\documentclass[conference]{IEEEtran}
\IEEEoverridecommandlockouts
\usepackage{cite}
\usepackage{amsmath,amssymb,amsfonts,amsthm}
\usepackage{algorithmic}
\usepackage{graphicx}
\usepackage{textcomp}
\usepackage{xcolor}
\def\BibTeX{{\rm B\kern-.05em{\sc i\kern-.025em b}\kern-.08em
    T\kern-.1667em\lower.7ex\hbox{E}\kern-.125emX}}

\usepackage{thm-restate}

\usepackage{xspace}

\usepackage{kmath}

\usepackage{misc}
\usepackage{enumitem}
\usepackage{tikz}
\usetikzlibrary{calc}

\usepackage{tikz}
\usetikzlibrary{automata}
\usetikzlibrary{positioning,fit,calc}
\usetikzlibrary{decorations.pathreplacing}
\usetikzlibrary{matrix}

\tikzset{>=latex, 
	point/.style = {circle,draw,thick,minimum size=2mm,inner sep=0pt},
	point1/.style = {circle,draw,thick,minimum size=6mm,inner sep=0pt},
	hm/.style = {dotted,semithick},
	role/.style = {thick},
	tree/.style = {rounded corners=10pt, dashed, fill opacity=0.5, fill=nullscolour},
	wiggly/.style={thick,
	},
	query/.style={thick},
	itria/.style={
  draw,dashed,shape border uses incircle,
  isosceles triangle,shape border rotate=90,yshift=-1.45cm},
  square/.style={regular polygon,regular polygon sides=4}
}

\newcommand{\sub}{\textit{sub}}

\newcommand{\I}{\mathcal{I}}

\newcommand{\var}{\textit{var}}

\newcommand{\sig}{\mathit{sig}}

\newcommand{\FO}{\ensuremath{\mathsf{FO}}}
\newcommand{\MSO}{\ensuremath{\mathsf{MSO}}}

\newcommand{\FOT}{\ensuremath{\mathsf{FO^2}}}
\newcommand{\FOTNE}{\ensuremath{\mathsf{FO^2_\ne}}}

\renewcommand{\ML}{\ensuremath{\mathsf{ML}}}
\newcommand{\MLi}{\ensuremath{\mathsf{ML\!^i}}}
\newcommand{\MLu}{\ensuremath{\mathsf{ML\!^u}}}
\newcommand{\MLiu}{\ensuremath{\mathsf{ML\!^{i,u}}}}
\newcommand{\MLin}{\ensuremath{\mathsf{ML\!^{i,n}}}}
\newcommand{\MLinu}{\ensuremath{\mathsf{ML\!^{i,n,u}}}}

\newcommand{\MLn}{\ensuremath{\mathsf{ML\!^{n}}}}
\newcommand{\MLnu}{\ensuremath{\mathsf{ML\!^{n,u}}}}
\newcommand{\GML}{\ensuremath{\mathsf{GML}}}
\newcommand{\GMLi}{\ensuremath{\mathsf{GML\!^i}}}
\newcommand{\GMLu}{\ensuremath{\mathsf{GML\!^u}}}
\newcommand{\GMLiu}{\ensuremath{\mathsf{GML\!^{i,u}}}}
\newcommand{\GMLin}{\ensuremath{\mathsf{GML\!^{i,n}}}}
\newcommand{\GMLinu}{\ensuremath{\mathsf{GML\!^{i,n,u}}}}

\newcommand{\GMLn}{\ensuremath{\mathsf{GML\!^{n}}}}
\newcommand{\GMLnu}{\ensuremath{\mathsf{GML\!^{n,u}}}}

\newcommand{\bis}{\boldsymbol{\beta}}
\renewcommand{\L}{L}
\newcommand{\LS}{L'}
\newcommand{\CT}{\ensuremath{\mathsf{C^2}}}

\newcommand{\Du}{\Diamondblack}

\newcommand{\type}{\mathfrak{t}}

\newcommand{\ftype}{\mathfrak{st}}

\newcommand{\Types}{\boldsymbol{T}}
\newcommand{\mx}{k^\bullet}

\newcommand{\md}{\textit{md}}

\newcommand{\Flat}{\mathit{fl}}

\renewcommand{\S}{\bar{S}}

\newcommand{\reg}{\mathit{reg}}

\begin{document}

\title{Separation and Definability in Fragments of Two-Variable First-Order Logic with Counting
}

\author{\IEEEauthorblockN{1\textsuperscript{st} Given Name Surname}
\IEEEauthorblockA{\textit{dept. name of organization (of Aff.)} \\
\textit{name of organization (of Aff.)}\\
City, Country \\
email address or ORCID}
\and
\IEEEauthorblockN{2\textsuperscript{nd} Given Name Surname}
\IEEEauthorblockA{\textit{dept. name of organization (of Aff.)} \\
\textit{name of organization (of Aff.)}\\
City, Country \\
email address or ORCID}
\and
\IEEEauthorblockN{3\textsuperscript{rd} Given Name Surname}
\IEEEauthorblockA{\textit{dept. name of organization (of Aff.)} \\
\textit{name of organization (of Aff.)}\\
City, Country \\
email address or ORCID}
\and
\IEEEauthorblockN{4\textsuperscript{th} Given Name Surname}
\IEEEauthorblockA{\textit{dept. name of organization (of Aff.)} \\
\textit{name of organization (of Aff.)}\\
City, Country \\
email address or ORCID}
\and
\IEEEauthorblockN{5\textsuperscript{th} Given Name Surname}
\IEEEauthorblockA{\textit{dept. name of organization (of Aff.)} \\
\textit{name of organization (of Aff.)}\\
City, Country \\
email address or ORCID}
\and
\IEEEauthorblockN{6\textsuperscript{th} Given Name Surname}
\IEEEauthorblockA{\textit{dept. name of organization (of Aff.)} \\
\textit{name of organization (of Aff.)}\\
City, Country \\
email address or ORCID}
}

\author{
	\IEEEauthorblockN{Louwe Kuijer}
\IEEEauthorblockA{\textit{Department of Computer Science} \\
\textit{University of Liverpool}\\
Liverpool, UK \\
0000-0001-6696-9023}
\and
\IEEEauthorblockN{Tony Tan}
\IEEEauthorblockA{\textit{Department of Computer Science} \\
\textit{University of Liverpool}\\
Liverpool, UK \\
0009-0005-8341-2004}
\and
\IEEEauthorblockN{Frank Wolter}
\IEEEauthorblockA{\textit{Department of Computer Science} \\
\textit{University of Liverpool}\\
Liverpool, UK \\
0000-0002-4470-606X}
\and
\IEEEauthorblockN{Michael Zakharyaschev}
\IEEEauthorblockA{\textit{School of Computing and Mathematical Sciences} \\
\textit{Birkbeck, University of London}\\
London, UK \\
0000-0002-2210-5183}
}

\newtheorem{theorem}{Theorem}
\newtheorem{lemma}[theorem]{Lemma}
\newtheorem{corollary}[theorem]{Corollary}
\newtheorem{definition}[theorem]{Definition}
\newtheorem{example}[theorem]{Example}
\newtheorem{remark}[theorem]{Remark}
\newtheorem{proposition}[theorem]{Proposition}

\maketitle

\begin{abstract}
For fragments $\L$ of first-order logic (\FO) with counting quantifiers, we consider the definability problem, which asks whether a  given $\L$-formula can be equivalently expressed by a formula in some fragment of $\L$ without counting, and the more general separation problem asking whether two mutually exclusive $\L$-formulas can be separated in some counting-free fragment of $\L$. 
We show that separation is undecidable for the two-variable fragment of \FO{} extended with counting quantifiers and for the graded modal logic with inverse, nominals and universal modality. On the other hand, if inverse or nominals are dropped, separation becomes $\coNExpTime$- or $\TwoExpTime$-complete, depending on whether the universal modality is present. In contrast, definability can often be reduced in polynomial time to validity in $\L$. We also consider uniform separation and show that it often behaves similarly to definability.
\end{abstract}



\section{Introduction}\label{intro}

Extensions of decidable fragments of first-order logic (\FO) with counting quantifiers $\exists^{\ge k}x$ play an important role in computational logic. In modal logic (viewed as a fragment of \FO{} under the standard translation), counting was introduced in the early 1970s~\cite{goble1970grades,DBLP:journals/ndjfl/Fine72} as graded modalities $\Diamond^{\ge k}$ and the graded modal logic, $\GML$, has been investigated and applied ever since~\cite{DBLP:journals/jancl/Hoek92,DBLP:conf/lics/KazakovP09,DBLP:journals/tplp/BednarczykKW21}. In description logic (DL), counting quantifiers, known as qualified number restrictions,  are a core constructor used in almost all standard languages~\cite{BaaEtAl03a,Tob00}. A fundamental decidable fragment of \FO{} containing most decidable graded modal logics and standard DLs with qualified number restrictions is $\CT$, the two-variable first-order logic ($\FOT$) with counting quantifiers, in which satisfiability is  \NExpTime-complete~\cite{DBLP:conf/lics/GradelOR97,DBLP:conf/lics/PacholskiST97,DBLP:journals/jolli/Pratt-Hartmann05}; see also \cite{DBLP:journals/siglog/KieronskiPT18,Pratt23book}. Both $\GML$ and $\CT$ have been recently shown to characterise exactly the expressive power of certain graph neural networks (GNNs)~\cite{DBLP:conf/iclr/BarceloKM0RS20,DBLP:conf/lics/Grohe21}.

While the expressive power of logics with counting such as $\GML$ and $\CT$ is well-understood~\cite{DBLP:journals/sLogica/Rijke00,DBLP:journals/corr/abs-1910-00039,DBLP:journals/bsl/Grohe98}, the problem of defining and separating formulas with counting by counting-free ones has been left largely unexplored. This is in contrast to logics with fixpoints, where the definability and separation of fixpoint formulas by those without fixpoints has been investigated extensively~\cite{DBLP:conf/stoc/CosmadakisGKV88,DBLP:journals/jlp/HillebrandKMV95,DBLP:conf/stacs/Otto99,DBLP:conf/lics/BenediktCCB15,DBLP:journals/corr/abs-2406-01497}, and to automata theory, where the definability and separation of languages defined using expressive devices by means of less powerful ones has been studied in depth, with countless applications to logics over linear orders and trees~\cite{DBLP:conf/birthday/DiekertG08,DBLP:journals/corr/PlaceZ14,DBLP:conf/fossacs/ColcombetGM22,DBLP:conf/pods/Bojanzyk08}. 


The aim of this paper is to investigate the computational complexity of the following decision problems for most important logics $\L$ with counting and counting-free logics $\LS$:
\begin{description}
\item[$\L/\LS$-definability:] decide whether a given $\L$-formula $\varphi$ is equ\-ivalent to some $\LS$-formula, called  an $\LS$-\emph{definition} of $\varphi$;
	
	\item[$\L/\LS$-separation:] given two mutually exclusive $\L$-formulas $\varphi$ and $\psi$, decide whether there is an $\LS$-formula $\chi$---an \emph{$\LS$-separator} for $\varphi$ and $\psi$---such that $\varphi \models \chi$ and $\chi \models \neg \psi$;
	
	\item[uniform $\L/\LS$-separation:] given an $\L$-formula $\varphi$, decide if there is a logically strongest (modulo equivalence) $\LS$-formula $\chi$---a \emph{uniform $\LS$ separator} for $\varphi$---with $\varphi \models \chi$.    
\end{description}
In DL, for instance, $\varphi$ can be an ontology $\mathcal{O}$ and $\psi$ a concept $C$ that is not satisfiable with respect to $\mathcal{O}$, both given in a DL $\L$ with number restrictions. Then a separator ontology $\mathcal{O}'$ in a weaker, easier to comprehend language $\LS$ without number restrictions could provide an easier `explanation' of unsatisfiability as it inherits that $\mathcal{O}\models \mathcal{O}'$ and $C$ is not satisfiable under $\mathcal{O}'$. In the context of concept learning~\cite{DBLP:journals/ml/LehmannH10}, $\varphi$ and $\psi$ can represent positive and negative examples for a target concept $C$. Then any separator in an appropriately
chosen language $\LS$ without number restrictions could represent a qualitative rather than quantitative concept one aims to learn~\cite{DBLP:journals/ai/JungLPW22,DBLP:journals/tocl/ArtaleJMOW23}.

Note that $(a)$ $\L/\LS$-definability reduces to $\L/\LS$-separation for $\varphi$ and $\neg \varphi$, $(b)$ an $\LS$-definition of $\varphi$ is a uniform $\LS$-separator for $\varphi$, and $(c)$ a uniform $\LS$-separator for $\varphi$ is an $\L/\LS$-separator for $\varphi$ and any $\L$-formula $\psi$ $\LS$-separable from $\varphi$. 
The next example illustrates definability and uniform separation of $\GML$-formulas by those in plain modal logic $\ML$ (the statements can be checked using Theorems~\ref{prop:definability1} and \ref{thm:uniformsep1}). 

\begin{example}\label{ex:gmlml}\em 
The $\GML$-formula $\Diamond^{\geq n}p$ is $\ML$-definable iff \mbox{$n = 1$}. However, $\Diamond^{\geq n}p$ always has a uniform $\ML$-separator, namely, $\Diamond p$. On the other hand, for any $n\geq 1$, the $\GML$-formulas $\Diamond^{\leq n}p$ and $\Diamond^{=n}p$ do not have a uniform $\ML$-separator.
The $\GML$-formula $\neg \Diamond^{=n}p$ is not $\ML$-definable but has a uniform $\ML$-separator, $\top$, for any $n\geq 1$, 
which means that $\neg \Diamond^{=n}p$ does not have any non-trivial consequences in $\ML$. \hfill $\dashv$
\end{example}

$\L/\LS$-separation can be regarded as a variant of Craig interpolation \cite{craig_1957}, where some logical symbols (here, counting quantifiers) are not allowed to occur in interpolants. We also consider the bona fide generalisation of Craig interpolation:
  
\begin{description}
	\item[Craig $\L/\LS$-separation:] decide whether given $\L$-formulas $\varphi$ and $\psi$ have an $\LS$-separator $\chi$ that is built from common non-logical symbols of $\varphi$ and $\psi$.
\end{description}	

\begin{example}\em\label{ex1}
Consider the $\CT$-formulas $\varphi(x) = \exists^{=1}y \, R(x,y)$ and 
$\psi(x) = \exists^{=1}y\, \big(R(x,y) \wedge A(y)\big) \wedge \exists^{=1}y\, \big(R(x,y)\wedge \neg A(y)\big)$. 
	%
Clearly, $\varphi \models \neg \psi$. The $\FOT$-formula 
$$
\chi(x) = \forall y\, \big( R(x,y) \rightarrow A(y)\big) \vee \forall y\, \big(R(x,y) \rightarrow \neg A(y)\big)
$$ 
separates $\varphi(x)$ and $\psi(x)$. 
For $\psi'(x)= \exists^{=2}y\, R(x,y)$, we also have $\varphi\models \neg\psi'$, but $\varphi(x)$ and $\psi'(x)$ are not $\FOT$-separable.
On the other hand, there is no Craig $\FOT$-separator for $\varphi(x)$ and $\psi(x)$, which has to be built from predicate $R$ only, and so would separate $\varphi(x)$ and $\psi'(x)$ as well. \hfill $\dashv$
\end{example}

Our core decidability and complexity results are summarised in 
Table~\ref{table:results}. Here, $\FOT$ admits equality, $\FOTNE$ is equality-free, and $\ML\!^\mathsf{e}$ and $\GML\!^\mathsf{e}$ denote extensions of $\ML$ and $\GML$ with the constructors in $\mathsf{e} \subseteq \{\mathsf{i}, \mathsf{u}, \mathsf{n}\}$, where $\mathsf{i}$ stands for inverse modalities,  $\mathsf{u}$ the universal modality, and $\mathsf{n}$ nominals.

\begin{table}[tbp]
\caption{Decidability and complexity results}
\begin{center}
\begin{tabular}{c|c|c}
logics $\L/\LS$ & (Craig) separation & definability/uniform separation \\\hline
$\CT/\FOT$ & undecidable & \textbf{?}\\\hline
$\CT/\FOTNE$ &  &  \\
$\GMLinu / \MLinu$  & undecidable  & \coNExpTime-comp.\\ 
$\GMLin / \MLin$ &  & \\
\hline
$\GMLiu / \MLiu$  &    &    \\
$\GMLnu / \MLnu$  &   \TwoExpTime-comp. &  \ExpTime-comp.  \\
$\GMLu / \MLu$  &    &    \\\hline
$\GMLi / \MLi$  &    &   \\
$\GMLn / \MLn$  &   \coNExpTime-comp. &  \PSpace-comp. \\
$\GML / \ML$  &    &   
\end{tabular}
\label{table:results}
\end{center}	
\end{table}

Given an $\L$-formula $\varphi$, the proofs of the complexity results for $\L/\LS$-definability and uniform separation in Table~\ref{table:results} construct an $\LS$-formula $\varphi'$ that is a definition/uniform separator for $\varphi$ iff one exists and give polytime reductions to $\L$-validity (sometimes also over finite structures). The surprisingly simple construction is inspired by approaches to $\FO/\ML$ and $\MSO/\mu\ML$-definability via $\omega$-expansions~\cite{DBLP:journals/igpl/AndrekaBN95,DBLP:conf/concur/JaninW96}. 
For $\CT$ and nominal-free modal logics, the formula $\varphi'$ is constructed in polytime (if given as a DAG); if nominals are present, $\varphi'$ is of exponential size and more work is needed to prove a polytime reduction to (finite) $\L$-validity. We also show
tight complexity bounds for other pairs not mentioned in Table~\ref{table:results} explicitly, say $\CT/\MLiu$ and $\CT/\ML$. As
$\CT$ is exponentially more succinct than $\MLiu$ (even for DAG-representation), to obtain polytime reductions to validity, we work with a succinct version of $\MLiu$ inspired by the presentation of $\FOT$ as a modal logic~\cite{DBLP:conf/csl/LutzSW01}. 

The proofs for (Craig) $\L/\LS$-separation follow a completely different approach and use a basic characterisation of non-separation of $\varphi$ and $\psi$ in $\LS$ via the existence of $\LS$-bisimilar models satisfying $\varphi$ and $\psi$, for an appropriate notion of $\LS$-bisimulation. In this case, our complexity results hold for the pairs $\L/\LS$ in Table~\ref{table:results} and also all weakenings of $\LS$ mentioned in Table~\ref{table:results} (except for $\L = \GMLinu$ and $\LS$ without inverse modalities, where undecidability is open). The undecidability results are established by a reduction of the non-halting problem for Minsky machines, in which, remarkably, a single nominal fixing the root of a model is enough to obtain the undecidability for $\GMLinu$.
The technically most demanding proof is the \TwoExpTime{} upper bound for $\GMLiu / \MLiu$-separation, which uses a mosaic technique to show  that if one cannot fix the root of models using a nominal, the construction of bisimilar models can be regularised. The difference between \coNExpTime and \TwoExpTime is due to the presence/absence of the universal modality without which one can only look polynomially-deep into the bisimilar models.    


\smallskip
\noindent
\textit{Related Results.}
For a decidable \FO-fragment $\L$ and $\LS\in \{\ML,$ $\MLi\}$, the decidability of $\L/\LS$-definability follows from Otto's result that if an \FO-formula $\varphi(x)$ has an $\LS$-definition, then it has an $\LS$-definition of modal depth $\le 2^{n}-1$, where $n$ is the quantifier rank of $\varphi(x)$~\cite{DBLP:journals/apal/Otto04}. No elementary complexity bound is obtained using this approach. A \TwoExpTime-upper bound for $\GMLu/\MLu$-definability is shown in~\cite{LutEtAl11} in the context of DL TBoxes.
Independently from this paper, and using different methods, \TwoExpTime-completeness of $\GML/\ML$-separation is shown in~\cite{Jeanstacs}, even for $\GML$ with fixpoint-operators. 
For further discussion of separation and definability problems and their relationship to Craig interpolation, we refer the reader to~\cite{chapter:separation}.

The paper is organised as follows. All relevant logics and their semantics are defined in Sec.\,\ref{sec:logicsabst}. Model-theoretic characterisations of separation in terms of bisimulations are given in Sec.\,\ref{sec:definitions}, which also shows polytime reductions between some of our decision problems. The undecidability results for separation are proved in Sec.\,\ref{sec:undec} and the decidability ones in Secs.\,\ref{sec:gmlu},\,\ref{inverse-sep}. Definability is considered in Sec.\,\ref{sec:definability} and uniform separation in Sec.\,\ref{sec:uniformseparation}. Finally, Sec.\,\ref{future} discusses future work.

All omitted proofs can be found in the appendix.


\section{Logics}\label{sec:logicsabst}
\label{sec:logics}

First, we define the syntax and semantics of logics in Table~\ref{table:results}. Let $\sigma$ be a (possibly infinite) \emph{signature} of unary and binary predicate symbols and constants, let $\FO(\sigma)$ denote the set of first-order formulas with equality over $\sigma$, and let $\var$ be a set comprising two individual variables. Then
\begin{description}
\item[\rm$\FOT(\sigma)$,] the \emph{two-variable fragment of $\FO(\sigma)$}, is the set of formulas that are built from atoms of the form $\top(x)$, $A(x)$, $R(x,y)$, and $x=y$ with unary $A\in \sigma$, binary $R\in \sigma$, and $x,y \in \var$, using the Boolean connectives $\wedge$ and $\neg$ and quantifier $\exists x$ with $x\in \var$ (other Boolean connectives and $\forall x$ are regarded as standard abbreviations);

\item[\rm$\FOTNE(\sigma)$] is the \emph{equality-free fragment of $\FOT(\sigma)$};

\item[\rm$\CT(\sigma)$,] the \emph{two-variable fragment of $\FO(\sigma)$ with counting}, extends $\FOT(\sigma)$ with the counting quantifiers $\exists^{\ge k}x$, for $k \ge 1$ and $x\in \var$ (other counting quantifiers $\exists^{=k}x$, $\exists^{< k}x$, $\exists^{> k}x$, etc.\ are regarded as abbreviations).
\end{description}
Since constants, $c$, can be defined in $\FOT$ using fresh unary predicates $C$ as $\exists x\, \big( C(x) \land \forall y\, (C(y) \to x=y)\big)$, when dealing with $\FOT(\sigma)$ and $\CT(\sigma)$, we assume that $\sigma$ is constant-free.
 

We also consider a few fragments of $\FOT$ and $\CT$ corresponding to standard propositional (graded) modal logics. In this context, each unary $A \in\sigma$ is treated as a \emph{propositional variable} and each binary $R \in \sigma$ gives rise to \emph{graded modal operators} $\Diamond\!^{\ge k}_R$, for $k \ge 1$.  
\begin{description}
\item[$\ML(\sigma)$,] the basic \emph{modal logic} over $\sigma$, is defined as the set of formulas given by the following grammar with $A,R \in \sigma$:
$\varphi \ := \ \top \ \mid  \ A  \ \mid  \ \neg \varphi  \ \mid  \ \varphi \land \varphi' \ \mid \ \Diamond\!^{\ge 1}_R \varphi$.   Instead of $\Diamond\!^{\ge 1}_R \varphi$, we usually write $\Diamond_R \varphi$.
 
\item[$\GML(\sigma)$,] the \emph{graded modal logic} over $\sigma$, extends $\ML(\sigma)$ with graded modal operators $\Diamond\!^{\ge k}_R$, for any $k \ge 1$ and $R \in \sigma$. 
\end{description}
%
%
The popular extensions of $\ML$ and $\GML$ with the \emph{inverse} modal operators, \emph{universal modality}, and/or \emph{nominals} are denoted by $\ML\!^\mathsf{e}(\sigma)$ and $\GML\!^\mathsf{e}(\sigma)$, where $\mathsf{e} \subseteq \{\mathsf{i}, \mathsf{u}, \mathsf{n}\}$ and 
\begin{description}
\item[\quad$\mathsf{i}$] indicates that whenever $\Diamond\!^{\ge k}_R$ is admitted in the logic in question, then its \emph{inverse} $\Diamond\!^{\ge k}_{R^-}$ is also admitted;

\item[\quad$\mathsf{u}$] indicates admittance of the \emph{universal diamond} $\Du\varphi$;

\item[\quad$\mathsf{n}$] indicates that $\sigma$ may contain individual constants $c$, which give rise to atomic formulas $N_c$, known as \emph{nominals}.
\end{description}
Thus, $\GMLinu$ is our strongest graded modal logic. We see it as a fragment of $\CT$ through the lens of the \emph{standard translation} $\cdot^\ast$, which rephrases any $\GMLinu$-formula $\varphi$ as the $\CT$-formula $\varphi^{\ast}_x$ with one free variable $x \in \var$ defined inductively by taking 
\begin{align*}
&  \top^{\ast}_x  =  \top(x), \ \ A^{\ast}_x  =  A(x), \ \ {(N_c)}_x^\ast = (x = c), \ \ 
(\neg \varphi)_x^\ast = \neg \varphi_x^\ast, \\ 
& (\varphi \land \psi)_x^\ast = \varphi_x^\ast \land \psi_x^\ast, \ \ (\Du \varphi)_x^\ast = \top(x) \land \exists \bar x \, \varphi_{\bar x}^\ast , \\  
& 
(\Diamond^{\ge k}_R \varphi)_x^\ast = \exists^{\ge k} \bar x \, ( R(x,\bar x) \land \varphi_{\bar x}^\ast ), \,  
(\Diamond^{\ge k}_{R^-} \varphi)_x^\ast = \exists^{\ge k} \bar x \, ( R(\bar x, x) \land \varphi_{\bar x}^\ast )  
\end{align*} 
where $\bar x = y$, $\bar y = x$ and $\{x,y\} = \var$. 

The \emph{signature} of $\varphi$ is the set $\sig(\varphi)$ of predicate and cons\-tant symbols in $\varphi$. 
The set $\sub(\varphi)$ comprises all subformulas of $\varphi$ together with their negations; $\sub(\varphi,\psi) = \sub(\varphi) \cup \sub(\psi)$. We assume that formulas are given as DAGs and the parameter $k$ in $\exists^{\ge k}$ and $\Diamond\!^{\ge k}_R$ is encoded in binary, in which case the \emph{size} of $\varphi$ is bounded by $|\varphi| = |\sub(\varphi)| \cdot \log\mx$, where $\mx$ is the largest parameter in $\varphi$. Our complexity results remain the same if we adopt the unary encoding of parameters and define the size of $\varphi$ as the number of occurrences of symbols in it. 

A $\sigma$-\emph{structure} takes the form $\Amf=(\dom(\Amf),(Q^{\Amf})_{Q\in \sigma},(c^{\Amf})_{c\in \sigma})$ with a  \emph{domain} $\dom(\Amf) \ne \emptyset$, relations $Q^{\Amf}$ on $\dom(\Amf)$ of the same arity as $Q \in \sigma$, and \mbox{$c^{\Amf}\in \dom(\Amf)$}, for constants $c\in \sigma$. 
We write $\varphi \models \psi$ to say that $\varphi$ \emph{entails} $\psi$ in the sense that if $\varphi$ is true in some structure $\mathfrak A$ under some assignment of elements of $\dom(\Amf)$ to the variables in $\var$, then $\psi$ is also true in $\mathfrak A$ under the same assignment.  If $\varphi \models \psi$ and $\psi \models \varphi$, we write $\varphi \equiv \psi$, calling $\varphi$ and $\psi$ \emph{equivalent}.



The complexities of the satisfiability problems for the logics above are as follows (see~\cite{DL-Textbook,Pratt23book,DBLP:journals/jsyml/ArecesBM01} and references therein): 
\begin{itemize}
\item[--] $\FOT$, $\CT$, $\GMLinu$, and $\GMLin$ are $\NExpTime$-complete;

\item[--] $\MLu$, $\GMLu$, $\MLiu$, $\MLnu$, $\GMLnu$, $\GMLiu$ and $\MLinu$ are $\ExpTime$-complete;

\item[--] $\ML$, $\GML$, $\MLi$, $\GMLi$, $\MLn$, $\GMLn$ are $\PSpace$-complete.
\end{itemize}


\section{Separation and Bisimulation}\label{sec:definitions}

We now introduce the basic semantic instruments that are required for determining the decidability and complexity of the decision problems formulated in Sec.~\ref{intro} for the pairs $\L/\LS$ of logics such that
\begin{itemize}
\item[--] $\L$ is any logic with counting defined in Sec.~\ref{sec:logicsabst};

\item[--] $\LS \subseteq \L$ is any logic without counting from Sec.~\ref{sec:logicsabst}.
\end{itemize}
From now on, we assume that a signature $\sigma$ is fixed, that $\L$ and $\L'$ range over the logics with and, respectively, without counting defined in Sec.~\ref{sec:logicsabst}, and that $\L \supseteq \LS$ in every pair $\L/\LS$. (As $\sigma$ in $\CT(\sigma)$ is assumed to be constant-free, we do not consider pairs such as $\CT/\MLn$.) To simplify presentation and since most of our logics are modal, here we only consider formulas $\vp(x)$ with one free variable $x \in \var$.

A \emph{pointed structure} is a pair $\Amf,a$ with $a \in \dom(\Amf)$. We write $\mathfrak A \models \varphi(a)$ if $a$ satisfies $\varphi(x)$  in $\mathfrak A$. Given $\varrho \subseteq \sigma$ and  pointed $\sigma$-structures $\Amf,a$ and $\Bmf,b$, we write \mbox{$\Amf,a \equiv_{\L(\varrho)} \Bmf,b$} and call $\Amf,a$ and $\Bmf,b$  \emph{$L(\varrho)$-equivalent} in case \mbox{$\Amf\models \varphi(a)$} iff $\Bmf\models \varphi(b)$, for all $L(\varrho)$-formulas $\varphi(x)$.

%
%

A non-empty binary relation $\bis \subseteq \dom(\Amf)\times\dom(\Bmf)$ is an \emph{$\FOT(\varrho)$-bisimulation between \Amf and \Bmf} if $\bis$ is \emph{global}---that is, $\dom(\Amf)\subseteq \{a\mid (a,b)\in \bis\}$ and $\dom(\Bmf)\subseteq \{b\mid (a,b)\in \bis\}$---and for
every $(a,b)\in \bis$, the following conditions are satisfied: 
\begin{itemize}
\item[--] for every $a'\in\dom(\Amf)$, there is a $b'\in \dom(\Bmf)$ such that $(a',b')\in \bis$ and $(a,a')\mapsto (b,b')$ is a partial  $\varrho$-isomorphism between \Amf and \Bmf;
	
\item[--] for every $b'\in\dom(\Bmf)$, there is an $a'\in \dom(\Amf)$ such that $(a',b')\in \bis$ and $(a,a')\mapsto (b,b')$ is a partial  $\varrho$-isomorphism between \Amf and \Bmf.
\end{itemize}
We write $\Amf,a \sim_{\FOT(\varrho)} \Bmf,b$ if 
$a\mapsto b$ is a partial $\varrho$-isomorphism between \Amf and
\Bmf and there is an $\FOT(\varrho)$-bisimulation $\bis$ between $\Amf$ and
\Bmf such that $(a,b)\in \bis$.

A non-empty $\bis \subseteq \dom(\Amf)\times\dom(\Bmf)$ is an \emph{$\ML(\varrho)$-bisimulation between $\Amf$ and $\Bmf$} if the following conditions hold:
\begin{itemize}
\item[--] if $(a,b) \in \bis$, then $a\in A^{\Amf}$ iff $b\in A^{\Bmf}$ for all $A\in \varrho$;

\item[--] if $(a,b) \in \bis$ and $(a,a')\in R^{\Amf}$, for $R\in \varrho$, then there exists $b'\in\dom(\Bmf)$ such that $(a',b')\in \bis$ and $(b,b')\in R^{\Bmf}$;
	
\item[--] if $(a,b) \in \bis$ and $(b,b')\in R^{\Bmf}$, for $R\in \varrho$, then there exists $a'\in\dom(\Amf)$ such that $(a',b')\in \bis$ and $(a,a')\in R^{\Amf}$. 
\end{itemize}
Further, an $\ML(\varrho)$-bisimulation $\bis$ between $\Amf$, $\Bmf$ is called an 
\begin{description}
\item[\rm\emph{$\MLi(\varrho)$-bisimulation}] if the second and third conditions also hold for $R = P^{-}$ with $P\in \varrho$;

\item[\rm\emph{$\MLn(\varrho)$-bisimulation}] if $(c^{\Amf},c^{\Bmf})\in \bis$, for each constant $c\in \varrho$;

\item[\rm\emph{$\MLu(\varrho)$-bisimulation}] if it is global.
\end{description}
$\MLiu(\varrho)$-, $\MLinu(\varrho)$-, etc.\ bisimulations are those $\bis$ that satisfy all of the conditions for the respective superscripts.
For a modal logic $\LS$, we write $\Amf,a \sim_{\LS(\varrho)} \Bmf,b$ if there is an $\LS(\varrho)$-bisimulation $\bis$ between $\Amf$ and $\Bmf$ such that $(a,b)\in \bis$.

The following characterisation is well-known; see, e.g.,~\cite{DBLP:books/el/07/GorankoO07,DBLP:books/daglib/p/Gradel014,Pratt23book} (see also~\cite{modeltheory} for a discussion of $\omega$-saturated structures).

\begin{restatable}{lemma}{lemguardedbisim}\label{lem:guardedbisim}
For any $\LS$, $\varrho \subseteq \sigma$, and any pointed $\sigma$-structures $\Amf,a$ and $\Bmf,b$,
$$
\Amf,a \sim_{\LS(\varrho)} \Bmf,b \quad \text{ implies } \quad
\Amf,a \equiv_{\LS(\varrho)} \Bmf,b
$$
and, conversely, if structures $\Amf$ and $\Bmf$ are $\omega$-saturated, then
$$
\Amf,a \equiv_{\LS(\varrho)} \Bmf,b \quad \text{ implies } \quad 
\Amf,a \sim_{\LS(\varrho)} \Bmf,b.
$$
\end{restatable} 

An $\LS(\varrho)$-\emph{separator} for $\L$-formulas $\varphi(x)$ and $\psi(x)$ is any $\LS(\varrho)$-formula $\chi(x)$ such that $\varphi(x) \models \chi(x)$ and $\chi(x) \models \neg \psi(x)$.
The next criterion is similar to Robinson's joint consistency theorem~\cite{Robinson1960}:

\begin{lemma}\label{criterion}
For any pair $\L/\LS$, any $\L$-formulas $\varphi(x)$, $\psi(x)$ and any $\varrho \subseteq \sigma$, the following conditions are equivalent\textup{:}
\begin{itemize}
\item[--] $\varphi(x)$ and $\psi(x)$ do not have an $\LS(\varrho)$-separator\textup{;}

\item[--] there are pointed $\sigma$-structures $\Amf,a$ and $\Bmf,b$ such that 
$$
\Amf\models \varphi(a), \quad \Bmf\models \psi(b), \quad \Amf,a \sim_{\LS(\varrho)} \Bmf,b.
$$
\end{itemize} 
\end{lemma}
\begin{proof}
The proof is in line with the characterisations of Craig interpolant existence in~\cite{DBLP:conf/lics/JungW21,DBLP:journals/tocl/ArtaleJMOW23}.
\end{proof}


We next observe a few reductions between the decision problems from Sec.~\ref{intro}. The next lemma requires the notion of \emph{unary Craig $\L/\LS$-separation}, which lies between $\L/\LS$-separation and Craig $\L/\LS$-separation: an $\LS$-separator $\chi$ for $\varphi$ and $\psi$ is called a \emph{unary Craig $\L/\LS$-separator} if all unary $A \in \sig(\chi)$ occur in both $\varphi$ and $\psi$ (non-unary predicates may be arbitrary). In modal logic, where unary predicates are regarded as propositional variables, this notion has actually been frequently considered for Craig interpolation when one does not want to restrict the use of modal operators in interpolants. Also, unary Craig separation makes our undecidability and lower complexity bound proofs more transparent.


\begin{restatable}{lemma}{lred}\label{l:red}
For any pair $\L/\LS$, 
\begin{itemize}
\item[--] $\L/\LS$-separation is polytime reducible to Craig $\L/\LS$-separation and to unary Craig $\L/\LS$-separation\textup{;}

\item[--] unary Craig $\L/\LS$-separation is polytime reducible to $\L/\LS$-separation.
\end{itemize}
\end{restatable}

While the former reduction is easy, the proof of the latter relies on Lemma~\ref{criterion}. Essentially, one replaces non-shared unary predicates in $\varphi$ and $\psi$ by compound counting formulas over shared symbols. Note also that, for all $\L/\LS$, we prove exactly the same (mostly tight) complexity bounds for plain separation and Craig separation. However, we could not find a `black-box' polytime reduction of Craig separation to separation.  

Finally, we lift the well-established link between Craig interpolation and (projective) Beth definability~\cite{craig_1957,Bet56} to a link between Craig $\L/\LS$-separation and an appropriate notion of $\L/\LS$-definability. 
\begin{description}	
\item[relative $\L/\LS$-definability:] given a signature $\varrho$ and two $\L$-formulas $\varphi$ and $\psi$,  decide whether there exists an \emph{explicit $\LS(\varrho)$-definition of $\psi$ modulo} $\varphi$, i.e., an $\LS(\varrho)$-formula $\delta$ such that $\varphi \models \delta \leftrightarrow \psi$.
\end{description}

\begin{restatable}{lemma}{thmequivdef}
\label{thm:equivdef}
For any pair $\L/\LS$, Craig $\L/\LS$-separation and relative $\L/\LS$-definability are polytime reducible to each other. 
\end{restatable}

Using a standard trick from, e.g.,~\cite{MGabbay2005-MGAIAD}, the proof shows how the construction of separators can be used to construct definitions and the other way round. 



\section{Separation in $\CT$ and $\GMLinu$ is Undecidable}\label{sec:undec}

Our first result is negative. It shows that separation of $\CT$- and $\GMLinu$-formulas by means of (nearly all) smaller counting-free languages in Table~\ref{table:results} is undecidable:

\begin{restatable}{theorem}{undecidability}
\label{thm:RE}
The $\L/\LS$-separation and Craig $\L/\LS$-separation problems are RE-complete \textup{(}and so undecidable$)$ for
\begin{itemize}
\item[--] $\L = \CT$ and any counting-free $\LS$ in Table~\ref{table:results}, and also for

\item[--] $\L = \{\GMLinu, \GMLin\}$ and any $\LS = \ML\!^{\mathsf{e}}$ in Table~\ref{table:results} such that $\mathsf{i} \in \mathsf{e} \subseteq \{\mathsf{i}, \mathsf{n}, \mathsf{u}\}$.
\end{itemize}
\end{restatable}

We first explain the idea of the proof for $\GMLinu/\MLinu$. We encode the non-halting problem for Minsky (aka counter or register) machines in terms of Craig-separator  non-existence. A \emph{Minsky machine $M$ with two registers} $\reg_0$, $\reg_1$ (2RM, for short) is given by a set $Q = \{q_0, \dots, q_m\}$ of states with initial $q_0$ and final $q_m$, and instructions $I_i$, $i < m$, that can be of two forms: $I_i = +(\ell,q_j)$ and $I_i = -(\ell, q_j, q_l)$, for $\ell \in \{0,1\}$. The former says that, when in state $q_i$, $M$ increments register $\reg_\ell$ by one and moves to state $q_j$. The latter means: in state $q_i$, if $\reg_\ell > 0$, $M$ decrements $\reg_\ell$ by one and moves to state $q_l$; if $\reg_\ell = 0$, $M$ moves to state $q_j$. 

The \emph{halting problem} is to decide, given a 2RM $M$, whether starting in $q_0$ with $\reg_0 = \reg_1= 0$, $M$ will reach the final state $q_m$. This problem is RE-complete~\cite{minsky_67}. We assume without loss of generality that $q_0$ does not occur in $M$'s instructions and that $I_0$ is a decrement instruction.

Given such an $M$, we define two $\GMLinu$-formulas $\varphi_M$, $\psi_M$. Intuitively, $\varphi_M$ describes $M$'s successive  states, $\psi_M$ captures the content of the registers, and the bisimulation from the $\GMLinu/\MLinu$-separation criterion (Lemma~\ref{criterion}) makes the two sequences match, thereby simulating the computation of $M$.
The shared signature, $\varrho$, of $\varphi_M$ and $\psi_M$ comprises propositional variables $q_i \in Q$ and $e_\ell$, $\ell = 0,1$, saying that $\reg_\ell = 0$, a binary predicate $R$ for computation steps, and a constant $c$ for the starting point of the computation. Additionally, $\psi_M$ uses (non-shared) propositional variables $r_0$ and $r_1 \equiv \neg r_0$ for encoding the registers and auxiliary $u$ whose purpose is explained later.

To illustrate, suppose $M$ contains (among others) instructions $I_0 = -(0,q_1,q_1)$, $I_1 = +(0,q_4)$ and $I_4=-(0,q_2,q_4)$. The intended structures $\Amf,a \models \varphi_M$ and $\Bmf,b \models \psi_M$ look like in the picture below. The formula $\varphi_M$ generates an infinite $R$-chain\\ 
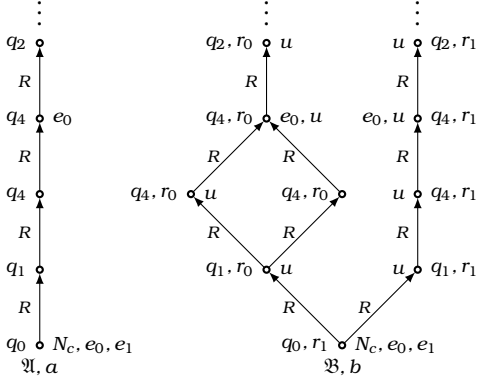
\centerline{
\begin{tikzpicture}[>=latex,line width=0.4pt,xscale = 1,yscale = 1]
\node[point,scale = 0.4,label=below:{\footnotesize $\Amf, a$},label=left:{\footnotesize $q_0$},label=right:{\footnotesize $N_c,e_0,e_1$}] (a0) at (0,0) {};
\node[point,scale = 0.4,label=left:{\footnotesize $q_1$}] (a1) at (0,1) {};
\node[point,scale = 0.4,label=left:{\footnotesize $q_4$}] (a2) at (0,2) {};
\node[point,scale = 0.4,label=left:{\footnotesize $q_4$},label=right:{\footnotesize $e_0$}] (a3) at (0,3) {};
\node[point,scale = 0.4,label=left:{\footnotesize $q_2$}] (a4) at (0,4) {};
\node[] at (0,4.5) {$\vdots$};
\draw[->,left] (a0) to node[] {\scriptsize $R$} (a1);
\draw[->,left] (a1) to node[] {\scriptsize $R$} (a2);
\draw[->,left] (a2) to node[] {\scriptsize $R$} (a3);
\draw[->,left] (a3) to node[] {\scriptsize $R$} (a4);
%
%
\node[point,scale = 0.4,label=below:{\footnotesize $\Bmf,b$},label=left:{\footnotesize $q_0,r_1$},label=right:{\footnotesize $N_c,e_0,e_1$}] (b0) at (4,0) {};
\node[point,scale = 0.4,label=left:{\footnotesize $q_1,r_0$},label=right:{\footnotesize $u$}] (l1) at (3,1) {};
\node[point,scale = 0.4,label=right:{\footnotesize $q_1,r_1$},label=left:{\footnotesize $u$}] (r1) at (5,1) {};
\node[point,scale = 0.4,label=left:{\footnotesize $q_4,r_0$},label=right:{\footnotesize $u$}] (l21) at (2,2) {};
\node[point,scale = 0.4,label=left:{\footnotesize $q_4,r_0$}] (l22) at (4,2) {};
\node[point,scale = 0.4,label=right:{\footnotesize $q_4,r_1$},label=left:{\footnotesize $u$}] (r2) at (5,2) {};
\node[point,scale = 0.4,label=right:{\footnotesize $q_4,r_1$},label=left:{\footnotesize $e_0,u$}] (r3) at (5,3) {};
\node[point,scale = 0.4,label=right:{\footnotesize $q_2,r_1$},label=left:{\footnotesize $u$}] (r4) at (5,4) {};
\node[point,scale = 0.4,label=left:{\footnotesize $q_4,r_0$},label=right:{\footnotesize $e_0,u$}] (l3) at (3,3) {};
\node[point,scale = 0.4,label=left:{\footnotesize $q_2,r_0$},label=right:{\footnotesize $u$}] (l4) at (3,4) {};
\draw[->,left] (b0) to node[] {\scriptsize $R$} (l1);
\draw[->,left] (b0) to node[] {\scriptsize $R$} (r1);
\draw[->,left] (l1) to node[] {\scriptsize $R$} (l21);
\draw[->,left] (l1) to node[] {\scriptsize $R$} (l22);
\draw[->,left] (r1) to node[] {\scriptsize $R$} (r2);
\draw[->,left] (r2) to node[] {\scriptsize $R$} (r3);
\draw[->,left] (r3) to node[] {\scriptsize $R$} (r4);
\draw[->,left] (l21) to node[] {\scriptsize $R$} (l3);
\draw[->,left] (l22) to node[] {\scriptsize $R$} (l3);
\draw[->,left] (l3) to node[] {\scriptsize $R$} (l4);
\node[] at (3,4.5) {$\vdots$};
\node[] at (5,4.5) {$\vdots$};
\end{tikzpicture}
}\\
starting from $a$, each of whose points satisfies exactly one of the  $q_i \in Q$. The order of these $q_i$ follows the order of states in the computation of $M$, without ever reaching $q_m$. The variables $e_\ell$ are supposed to indicate the computation steps where the current instruction is to decrement $\reg_\ell = 0$ (their behaviour depends on both $\varphi_M$ and $\psi_M$). Syntactically, 
\begin{align*}
\varphi_M =~ & N_c \wedge q_0 \land \blacksquare \bigvee_{i \leq m} q_i \land 
\blacksquare\bigwedge_{i \leq m}(q_i\rightarrow \bigwedge_{i<j\leq m}\neg q_j) \land{} \\
& \square_{R^-}\bot \land \blacksquare \Diamond\!^{= 1}_{R}\top \land 
\blacksquare(\neg N_c\rightarrow \Diamond\!^{= 1}_{R^-}\top) \land \blacksquare\neg q_m \land{}\\
& \bigwedge_{I_i = +(\ell,q_j)} \blacksquare (q_i\to \square_R q_j)  \land{}\\
& \bigwedge_{I_i = -(\ell,q_j,q_k)} \hspace*{-3mm} \blacksquare((q_i\wedge e_\ell)\rightarrow \square_R q_j) \land \blacksquare((q_i\wedge \neg e_\ell)\rightarrow \square_R q_k) .
\end{align*} 
The formula $\psi_M$ simulates the behaviour of the registers $\reg_\ell$ by generating an $R$-directed graph with root $b$ whose  nodes of \emph{depth} $k \ge 0$ (that is, located at distance $k$ from $b$) represent the value of $\reg_\ell$ at step $k$ of $M$'s computation in the sense that this value is $n$ iff the number of depth $k$ nodes that make $r_\ell$ true is $2^n$. We explain the intuition behind $\psi_M$ by dividing it into four groups of conjuncts. The first one comprises the same conjuncts as the first line of the definition of $\varphi_M$. 

The second group consists of the conjuncts
\begin{align*}
& e_0 \land e_1 \land \Diamond\!_R^{=2}\top \land \Diamond_R(u \wedge r_0) \wedge \Diamond_R (u \wedge r_1) \land{}\\
& \bigwedge_{\ell = 0,1} \blacksquare \big[r_\ell \rightarrow \big(\square_R r_\ell \wedge \square_{R^{-1}}(r_\ell \vee N_c)\big)\big] \land{}\\
& \blacksquare \big((u\wedge \neg N_c)\rightarrow (\Diamond\!_R^{=1} u \wedge\Diamond\!_{R^{-1}}^{=1} u) \big).
\end{align*}
It generates one $r_0$-node and one $r_1$-node of depth 1, makes sure that all $R$-successors and all non-root $R$-predecessors of an $r_\ell$-node are also $r_\ell$-nodes, and singles out one branch of $r_\ell$-nodes, for each $\ell = 0,1$, by making the auxiliary variable $u$ true at its nodes. Thus, for each $k \ge 1$, $u$ is true at exactly one $r_\ell$-node of depth $k$. This fact will be used to handle the variables $e_\ell$ that are needed for checking whether $\reg_\ell = 0$.

The third group simulates the instructions using the following  conjuncts (see the picture above for an illustration): 
\begin{align*}
& \bigwedge_{I_i = +(\ell,q_j)} \hspace*{-2mm} \blacksquare \big[ \big( (q_i\wedge r_\ell) \rightarrow (\Diamond\!_R^{=2}\top \wedge \square_R\Diamond\!_{R^{-1}}^{=1}\top) \big) \land{}\\[-7pt]
& \hspace*{3cm} \big( (q_i\wedge \neg r_\ell )\rightarrow (\Diamond\!_R^{=1}\top \wedge \square_R\Diamond\!_{R^{-1}}^{=1}\top) \big)\big] \land{}
\end{align*}
\begin{align*}
& \bigwedge_{\substack{i \ne 0\\I_i = -(\ell,q_j,q_l)}} \hspace*{-3mm} \blacksquare \big[ \big((q_i \wedge r_\ell \wedge e_\ell) \rightarrow (\Diamond\!_R^{=1}\top \wedge \square_R\Diamond\!_{R^{-1}}^{=1}\top) \big) \land{}\\[-11pt]
& \hspace*{2.3cm} \big( (q_i \wedge r_\ell \wedge \neg e_\ell) \rightarrow (\Diamond\!_R^{=1}\top \wedge \square_R\Diamond\!_{R^{-1}}^{=2}\top) \big) \land{}\\
& \hspace*{3.3cm} \big( (q_i \wedge \neg r_\ell) \rightarrow (\Diamond\!_R^{=1}\top \wedge \square_R\Diamond\!_{R^{-1}}^{=1}\top) \big) \big].
\end{align*}
For the increment instructions $I_i = +(\ell,q_j)$, each node of depth $k$ making $q_i$ and $r_\ell$ true has two $R$-successors, each of which has one $R$-predecessor. Each $\neg r_\ell$-node of depth $k$ has one $R$-successor. This doubles the number of $r_\ell$-nodes of depth $k$ while keeping the number of $\neg r_\ell$-nodes unchanged as required. 
For the decrement instructions $I_i = -(\ell,q_j,q_l)$, we use the variables $e_\ell$ to check if $\reg_\ell = 0$. If this is the case, each $r_\ell$-node of depth $k$ has one $R$-successor that has one $R$-predecessor, which keeps $\reg_\ell = 0$ at step $k+1$. Otherwise, each $r_\ell$-node of depth $k$ has a single $R$-successor that has two $R$-predecessors, which halves the number of such nodes.
The nodes of depth $k$ in the other register (i.e., those satisfying $\neg r_\ell$) have exactly one $R$-successor that has exactly one $R$-predecessor, thereby maintaining the value of the register.
Decrement instruction $I_0$ is treated separately by the second group that creates registers $\reg_\ell = 0$, $\ell = 0,1$, at step (depth) 1. 


Finally, we require the conjunct $\blacksquare \big( (e_\ell \land r_\ell) \rightarrow u \big)$ to make sure that, if $\Amf \models \varphi_M(a)$, $\Bmf \models \psi_M(b)$, and $\Amf, a \sim_{\MLinu(\varrho)} \Bmf, b$, then $e_\ell$ is true at nodes representing the computation steps with a decrement instruction and $\reg_\ell = 0$.

It is straightforward to see that, for a non-halting $M$,  the intended  pointed structures $\Amf,a$ and $\Bmf,b$ described above are such that $\Amf \models \varphi_M(a)$, $\Bmf \models \psi_M(b)$, and $\Amf,a \sim_{\MLinu(\varrho)} \Bmf,b$, with a $\MLinu(\varrho)$-bisimulation $\bis$ obtained by taking $(a',b') \in \bis$ iff $a'$ and $b'$ are at the same distance $k < \omega$ from $a$ in $\Amf$ and, respectively, $b$ in $\Bmf$. It follows that $\varphi_M$ and $\psi_M$ are not $\MLinu(\varrho)$-separable. 
Conversely, we prove in the appendix that if there are $\Amf,a$ and $\Bmf,b$ with $\Amf \models \varphi_M(a)$, $\Bmf \models \psi_M(b)$, and $\Amf,a \sim_{\MLinu(\varrho)} \Bmf,b$, then $M$ is non-halting. In fact, the proof shows that, in this case, $\Amf$ and $\Bmf$ must contain the intended structures as bisimilar sub-structures. 

By Lemma~\ref{criterion}, it follows that Craig $\GMLinu/\MLinu$-separation is RE-complete and undecidable. This is also true of unary Craig $\GMLinu/\MLinu$-separation because $\varphi_M$ and $\psi_M$ have the same binary predicates. And using the second claim of Lemma~\ref{l:red}, we obtain that $\GMLinu/\MLinu$-separation is RE-complete. 
We can adapt the construction above to prove Theorem~\ref{thm:RE} for $\Du$- and $\blacksquare$-free $\GMLin /\MLin$ using the \emph{spypoint} technique~\cite{DBLP:journals/jsyml/ArecesBM01}; see the appendix for details. To deal with the nominal-free $\LS$, we can simulate $N_c$ by adding to $\varphi_M$ the conjunct $N_c \leftrightarrow p$ with a fresh propositional variable $p$ and replace $N_c$ in $\psi_M$ with $p$, making it shared in $\varphi_M$ and $\psi_M$. 

Finally, the first claim of Theorem~\ref{thm:RE} is proved similarly to the graded modal case, with some minor modifications required to take care of $\FOT(\varrho)$-bisimulations, which are more restricted than $\MLinu(\varrho)$-bisimulations; see the appendix. 


The next two sections establish the positive tight complexity results for (Craig) separation in Table~\ref{table:results}, starting with the technically simpler case of graded modal logics without inverses.


\section{Separation in Fragments of $\GMLnu$}\label{sec:gmlu}

The aim of this section is to prove the following:

\begin{restatable}{theorem}{complexGMLnu}
\label{thm:gmlu}
Separation and Craig-separation are
\begin{description}
\item[\it\textup{$2\ExpTime$}-complete] for the pairs $\GMLnu/\MLnu$, $\GMLnu/\MLn$, $\GMLnu/\MLu$, $\GMLu/\MLu$, and $\GMLu/\ML$\textup{;}

\item[\it\textup{$\coNExpTime$}-complete] for $\GMLn/\MLn$ and $\GML/\ML$. 
\end{description}
\end{restatable}

We begin by establishing the upper bounds. By Lemma~\ref{l:red}, it suffices to consider Craig separation. Let $\varphi_{1}$, $\varphi_{2}$ be $\GMLnu$-formulas, $\sigma = \sig(\varphi_{1}) \cup \sig(\varphi_{2})$ and $\varrho = \sig(\varphi_{1}) \cap \sig(\varphi_{2})$. 
%
A \emph{type} (\emph{for $\varphi_{1}$ and $\varphi_{2}$}) is any set $\type \subseteq \sub(\varphi_{1},\varphi_{2})$, for which there is a pointed $\sigma$-structure $\Amf, a$ such that $\type=\type(\Amf,a)$, where 
\[
\type(\Amf,a) = \{ \chi\in \sub(\varphi_{1},\varphi_{2})\mid \Amf \models \chi(a)\}
\] 
is the \emph{type} \emph{realised by $a$ in $\Amf$}. Denote by $\Types$ the set of all such types. Clearly, $|\Types| \le 2^{|\sub(\varphi_{1},\varphi_{2})|}$ and $\Types$ can be computed in time exponential in $|\varphi_{1}| + |\varphi_{2}|$.
Types $\type_1$ and $\type_2$ are called \emph{$\Du$-equivalent}
in case $\Du\chi\in \type_1$ iff $\Du\chi\in \type_2$, for all $\Du\chi\in \sub(\varphi,\psi)$. Note that the types realised in the same structure are $\Du$-equivalent to each other. 

We check the existence of $\sigma$-structures $\Amf_1,a_1$ and $\Amf,a_2$ such that $\Amf_i\models \varphi_i(a_i)$, $i=1,2$,  and $\Amf_1,a_1 \sim_{\MLnu(\varrho)} \Amf_2,a_2$ using a mosaic-elimination procedure, a generalisation of the standard type-elimination procedure for satisfiability-checking; see, e.g.,~\cite[p.~72]{GabEtAl03}. 

A \emph{mosaic} (\emph{for $\varphi_{1}$ and $\varphi_{2}$}) is a pair $M=(M_{1},M_{2})$ of non-empty sets $M_{1},M_{2}$ of types satisfying the following conditions: 
%
\begin{itemize}
\item[--] $A\in \type$ iff $A\in \type'$, for all $A\in \varrho$ and $\type,\type'\in M_{1}\cup M_{2}$;
 
\item[--] $N_c\in \type$ iff $N_c\in \type'$, for all $c \in \varrho$ and $\type,\type'\in M_{1}\cup M_{2}$; 

\item[--] 
if $\type, \type' \in M_{i}$, $i \in \{1,2\}$, $N_c\in \type \cap \type'$, $c\in \sigma$, then $\type = \type'$; 
	
\item[--] all types in each $M_{i}$ are $\Du$-equivalent to each other.
\end{itemize}
%
We say that $N_c$ \emph{occurs} in $M$ if $N_c \in \type\in M_{i}$, for some $\type$ and $i\in \{1,2\}$.
As follows from the second and third items, if $N_c$ occurs in  $M$, for some $c\in \varrho$, then $M_1$ and $M_2$ are singletons.

There are at most doubly exponentially many mosaics, and they can be computed in double exponential time. Let $f$ be the function defined by taking $f(M) = 1$ if $M$ contains a type with a nominal 
and $f(M) = \mx$ otherwise, where $\mx$ is the maximal $k$ such that $\Diamond_R^{\ge k}$ occurs in $\varphi_{1}$ or $\varphi_{2}$, for some $R$.

The definition of mosaic above is an abstraction from the following semantical notion. Given structures $\Amf_{1}$ and $\Amf_{2}$, the
\emph{mosaic defined by $a\in \dom(\Amf_{i})$ in $\Amf_{1},\Amf_{2}$}, $i \in \{1,2\}$, 
is the pair $M(\Amf_{1},\Amf_{2},a)=(M_{1}(\Amf_{1},a),M_{2}(\Amf_{2},a))$, where, for $j \in \{1,2\}$,
\[
M_{j}(\Amf_{j},a) = \{ \type(\Amf_{j},b) \mid b\in
\dom(\Amf_{j}), \ \Amf_{i},a \sim_{\MLnu(\varrho)}\Amf_{j},b\}.
\] 
%
Clearly, $M(\Amf_{1},\Amf_{2},a)$ is a mosaic and $M(\Amf_{1},\Amf_{2},a) = M(\Amf_{1},\Amf_{2},a')$ if $a \sim_{\MLnu(\varrho)} a'$, for any $a,a' \in \dom(\Amf_{1}) \cup \dom(\Amf_{2})$.

We aim to build the required structures $\Amf_{i}$, $i=1,2$, and bisimulation $\bis$ between them from copies of suitable pairs $(\type,M)$ with a mosaic $M=(M_{1},M_{2})$ and a type $\type\in M_i$, for which we, roughly, aim to  set $(\type_{1},M) \bis (\type_{2},M)$. 
To satisfy the modalities $\Diamond_{R}^{\geq k}\chi$ in the $\Amf_{i}$, we require the following definition.
%
Let $M=(M_{1},M_{2})$ be a mosaic, $\mathcal{M}$ a set of mosaics, and some type in $M_1\cup M_2$ contain $\Diamond_{R}^{\geq k}\chi$. We call $\mathcal{M}$ an \emph{$R$-witness for $M$} if there exists a relation $R^{M,\mathcal{M}} \subseteq (\Delta_{1} \times \Gamma_{1})\cup(\Delta_{2}\times \Gamma_{2})$ with  
\begin{align*}
& \Delta_{i} =\{ (\type,M) \mid \type\in M_{i}\},\\ 
& \Gamma_{i}=\{ (\type',M',j) \mid \type'\in M_{i}',\ M'=(M_{1}',M_{2}')\in \mathcal{M},\ 1 \le j \le f(M')\}
\end{align*}
satisfying the following for $\Delta=\Delta_{1} \cup \Delta_{2}$ and $\Gamma= \Gamma_{1}\cup \Gamma_{2}$:
\begin{description}

\item[($\Du$-harmony)] all types in $M_{i}$ and $M_{i}'$, for any $M'=(M_{1}',M_{2}')$ in $\mathcal{M}$, are $\Du$-equivalent to each other, $i=1,2$;
	
\item[(witness)] for any $\Diamond^{\geq k}_R\chi \in \sub(\varphi_1,\varphi_2)$ and $(\type,M)\in \Delta$, we have 
$$
\mbox{}\hspace*{-7mm}\Diamond^{\geq k}_R\chi \in \type \text{ iff } \big| \{ (\type',M',j) \in \Gamma \mid \chi\in \type', (\type,M) R^{M,\mathcal{M}} (\type',M',j) \}\big| \ge k;
$$
%
	
\item[(bisim)] if $R\in \varrho$, $M'\in \mathcal{M}$, and $(\type_{0},M)\in \Delta$ has an $R^{M,\mathcal{M}}$-successor $(\type_{0}',M',j)\in \Gamma$, then each $(\type_{1},M)\in \Delta$ has an $R^{M,\mathcal{M}}$-successor of the form $(\type_{1}',M',j')\in \Gamma$. 
\end{description}
The intuition behind conditions {\bf ($\Du$-harmony)} and {\bf (witness)} should be clear, and {\bf (bisim)} is needed to ensure that if, for any fixed $M'$, all $(\type',M',j)\in \Gamma$ are mutually $\MLnu(\varrho)$-bisimilar, then all $(\type,M)\in \Delta$ are also $\MLnu(\varrho)$-bisimilar to each other. 
It should also be clear that if $\Amf_1, a_1 \sim_{\MLnu(\varrho)}\Amf_2,a_2$, then every mosaic $M(\Amf_{1},\Amf_{2},a)$, for $a \in \dom(\Amf_i)$, $i=1,2$, has an $R$-witness for any relevant $R$, namely, the set $\mathcal{M}$ of all mosaics $M(\Amf_{1},\Amf_{2},b)$ such that $a' R^{\Amf_j} b$, for some $j \in \{1,2\}$ and $a' \in \dom(\Amf_j)$ with $M(\Amf_{1},\Amf_{2},a) = M(\Amf_{1},\Amf_{2},a')$.
Observe also that every $R$-witness $\mathcal{M}$ for $M$ contains a subset $\mathcal{M}'\subseteq \mathcal{M}$ such that $\mathcal{M}'$ is an $R$-witness for $M$ and $|\mathcal{M}'|\leq \mx 2^{|\sub(\varphi_{1},\varphi_{2})|}$. 
%
So, from now on, we assume all witnesses to be of size $\leq \mx 2^{|\sub(\varphi_{1},\varphi_{2})|}$. 

We can now define the mosaic elimination procedure for a set $\mathfrak{S}$ of mosaics. Call $M=(M_{1},M_{2})\in \mathfrak{S}$ \emph{bad} in $\mathfrak{S}$ if at least one of the following conditions is \emph{not} satisfied for $i = 1,2$:
\begin{itemize}
\item[--] if $\Diamond_{R}^{\geq k} \chi\in \type\in M_{i}$, then there is an $R$-witness for $M$ in $\mathfrak{S}$;
	
\item[--] if $\Du\chi \in \type\in M_{i}$, then there are $(M_{1}',M_{2}')\in \mathfrak{S}$ and $\type'\in M'_{i}$ such that $\chi\in \type' \in M_{i}'$ and $\type,\type'$ are $\Du$-equivalent.
\end{itemize}
Given a set $\mathfrak{S}$ of mosaics and $M\in \mathfrak{S}$, it can be decided in double exponential time in $|\varphi_{1}|+|\varphi_{2}|$ whether $M$ is bad in $\mathfrak{S}$.

Let $\mathfrak{S}_{0}$ be the set of all mosaics. Let $h$ be a function that, for any $c \in \sigma$ and $i \in \{1,2\}$, picks one mosaic $h(N_c,i) = (M_1,M_2)$ such that $N_c \in \type$, for some $\type \in M_i$, and if $N_{c'} \in \type'$, for some $\type' \in M_j$, $j \in \{1,2\}$ and $c'\in \sigma$, then $h(N_c,i) = h(N_{c'},j)$.
By the definition of mosaics, $h(N_c,1)=h(N_c,2)$ for all $c\in \varrho$.

Given such $h$, let $\mathfrak{S}_{0}^{h} \subseteq \mathfrak{S}_{0}$ comprise all mosaics of the form $h(N_c,i)$ and those that do not have occurrences of nominals. Starting from $\mathfrak{S}_{0}^{h}$, we construct inductively sets $\mathfrak{S}_{n+1}^{h}$, $n \ge 0$, by looking through the mosaics in $\mathfrak{S}_{n}^{h}$ and eliminating the bad ones. We stop when $\mathfrak{S}_{n+1}^{h} = \mathfrak{S}_{n}^{h}$ and set $\mathfrak{S}_{\ast}^{h}= \mathfrak{S}_{n}^{h}$. This can be done in double-exponential time (see the appendix).


\begin{restatable}{lemma}{mosaicsGMLnu}
\label{lem:main1} 
The following conditions are equivalent\textup{:}
\begin{description}
\item[\rm $(i)$] there exist pointed $\sigma$-structures $\Amf_i,a_i$, $i=1,2$, such that $\Amf_i\models \varphi_i(a_i)$ and $\Amf_1,a_1 \sim_{\MLnu(\varrho)} \Amf_2,a_2$ \textup{(}and so $\varphi_1$ and $\varphi_2$ are not $\LS$-separable\textup{)}\textup{;}
		
\item[\rm $(ii)$] there exist $h$, $M^* =(M^*_{1},M^*_{2})\in \mathfrak{S}_{\ast}^{h}$ and $\type^*_i \in M^*_i$, $i=1,2$, such that $\varphi_{i}\in \type^*_{i}\in M^*_{i}$, $h(N_c,j)=(M_{1},M_{2})\in \mathfrak{S}_{\ast}^{h}$, for all $c\in \sigma$ and $j=1,2$, and all types in $M^*_{j}\cup M_{j}$ are $\Du$-equivalent to each other.
\end{description} 
\end{restatable}

Here, we only sketch how, given $(ii)$, to construct the required $\Amf_i,a_i$. For any $M = (M_1,M_2)\in \mathfrak{S}_{\ast}^{h}$ and $\Diamond_{R}^{\geq k}\chi\in \type\in M_i$, for some $\Diamond_{R}^{\geq k}\chi$, $\type$ and $i \in \{1,2\}$, we take an $R$-witness $\mathcal{M}_{R,M}\subseteq \mathfrak{S}_{\ast}^{h}$ for $M$ with the  relation $R^{M,\mathcal{M}_{R,M}}$ and define the $\dom(\Amf_{i})$ as the sets of all words
\begin{equation}\label{eq1m}
s=(\type_{0},M_{0},i_{0})R_{1}(\type_{1},M_{1},i_{1}) R_2 \dots R_{n}(\type_{n},M_{n},i_{n}) 
\end{equation}
such that $\omega > n\geq 0$, $M_{j}=(M_{j,1},M_{j,2})\in \mathfrak{S}_{\ast}^{h}$, for $j\leq n$, 
\begin{itemize}
\item[--] $\type_{j}\in M_{j,i}$ and $1 \le i_j \le f(M_{j,i})$, for all $j\leq n$; 

\item[--] $M_{j+1}\in \mathcal{M}_{R_{j+1},M_{j}}$ and no nominal occurs in $M_{j+1}$;
		
\item[--] all types in $M_{k}^{\ast}$ and $M_{0,k}$ are $\Du$-equivalent, $k=1,2$.
\end{itemize}
For a propositional variable $A$ and $s$ in~\eqref{eq1m}, we set $s\in A^{\Amf_{i}}$ iff $A\in \type_{n}$; for $N_c$, we set $c^{\Amf_{i}} = s$ iff $n=0$, $N_c\in \type_{0}$ and $h(N_c,i) = M_0$, and for a binary $R \in \sigma$, we set:
\begin{itemize}
\item[--] $(s',s)\in R^{\Amf_{i}}$  if  $s'$ is obtained by dropping $R_{n}(\type_{n},M_{n},i_{n})$ from $s$, $R=R_{n}$ and
$
(\type_{n-1},M_{n-1})R_{n}^{M_{n-1},\mathcal{M}_{R_n,M_{n-1}}}
(\type_{n},M_{n},i_{n});
$

\item[--] $(s,(\type,M,j))\in R^{\Amf_{i}}$ if $(\type,M,j)\in \dom(\Amf_{i})$, a nominal occurs in $M$ and $(\type_{n},M_{n})R^{M_{n},\mathcal{M}_{R,M_{n}}} (\type,M,j)$.
\end{itemize}
Finally, the bisimulation $\bis$ between the $\Amf_{i}$ is defined by taking $s\bis s'$, for $s$ in~\eqref{eq1m} and 
$s'=(\type_{0}',M_{0}',i_{0}')R_{1}' \dots R_{m}'(\type_{m}',M_{m}',i_{m}')$, 
if $n=m$, $R_{i}=R_{i}'$ and $M_{i}=M_{i}'$, for all $i\leq n$.

Lemma~\ref{lem:main1} gives a $2\ExpTime$ upper bound for deciding (Craig) $\GMLnu/\MLnu$-separation, and the construction can be modified to show the same bound for the other logics with $\Du$. 
To establish the \coNExpTime upper bounds of Theorem~\ref{thm:gmlu} for  the $\Du$-free languages, we restrict the $\Amf_{i}$ to $s$ in  \eqref{eq1m} with $n$ not exceeding the modal depth of $\varphi_i$ and such that either $\type_0 \in \{\type^*_1, \type^*_2\}$ and $M_0 = M^*$ or $h(N_c,i) = M_0$, for some $c \in \sigma$.
The size of the resulting structures is exponential in $|\varphi_{1}|+|\varphi_{2}|$ and their existence can clearly be checked in \NExpTime. 

By Lemma~\ref{l:red}, it suffices to prove the matching lower bounds for unary Craig separation. 
The underlying idea, inspired by~\cite{DBLP:journals/tocl/ArtaleJMOW23}, is to construct polysize formulas $\varphi_n$ and $\psi_n$, $n < \omega$, in the required language such that if $\Amf\models \varphi_n(a)$, $\Bmf\models \psi_n(b)$, and $\Amf,a \sim_{\LS(\varrho)} \Bmf,b$, then there is $E\subseteq \dom(\Bmf)$
of size $2^{n}$ whose elements are $\LS(\varrho)$-bisimilar and satisfy all pairwise distinct combinations of variables $A_{0},\dots,A_{n-1} \notin \varrho$. 
For example, $\varphi_n = \Diamond_{R}^{=1}\dots \Diamond_{R}^{=1}\top$ with $n$-many $\Diamond_{R}^{=1}$ is such that $\Amf \models \varphi_n(a)$ iff $\Amf$ has a unique $R$-chain of length $n$ starting from $a$ and  
\begin{multline*}
\psi_n = \bigwedge_{0 \le i <n} \Big[ \Box_{R}^{i}\Diamond_{R}^{=2}\top \land \Box_{R}^{i}(\Diamond_{R}A_{i} \wedge \Diamond_{R}\neg A_{i}) \land{}\\[-3pt] 
\bigwedge_{0 \le j < i} \Box_{R}^{i}\big((A_{j} \rightarrow \Box_{R}A_{j}) \wedge (\neg A_{j} \rightarrow \Box_{R}\neg A_{j}) \big) \Big].
\end{multline*}
is such that $\Bmf \models \psi_n(b)$ iff the $n$-long $R$-chains from $b$ comprise a binary tree of depth $n$ whose leaves satisfy binary numbers encoded by the $A_{i}$. So, if $\Amf\models \varphi_n(a)$, $\Bmf\models \psi_n(b)$, and $\Amf,a\sim_{\ML(\varrho)} \Bmf,b$, then there is a unique $a'$ in $\Amf$ reachable via an $R$-chain of length $n$ from $a$, and there are leaves $E = \{e_{0},\dots,e_{2^{n}-1}\}$ of the binary tree with root $b$ in $\Bmf$ such that $e_{i}$ satisfies the binary encoding of $i$ via $A_{0},\dots,A_{n-1}$, and so 
$\Amf,{a'} \sim_{\LS(\varrho)}  \Bmf,e_{0} \sim_{\LS(\varrho)}\cdots \sim_{\LS(\varrho)}\Bmf,e_{2^{n}-1}$.
The set $E$ can now be used, following~\cite[Secs.\ 7, 9]{DBLP:journals/tocl/ArtaleJMOW23}, to encode exponential space-bounded alternating Turing machines for the \TwoExpTime-lower bound, and  $2^{n}\times 2^{n}$-torus tiling for the $\coNExpTime$ one.


\section{Separation in Fragments of $\GMLiu$}\label{inverse-sep}

Our next aim is to generalise the mosaic technique of Sec.~\ref{sec:gmlu} to nominal-free graded modal logics with inverse and universal modalities and prove the following:

\begin{restatable}{theorem}{mosaicGMLiu}
\label{thm:gmliu}
Separation and Craig-separation are
\begin{description}
\item[\it\textup{$\TwoExpTime$}-complete] for all pairs of the form $\GMLiu/\LS$, where $\LS \in \{\MLiu, \MLi,\MLu, \ML\}$, and 

\item[\it\textup{$\coNExpTime$}-complete] for $\GMLi/\MLi$ and $\GMLi/\ML$. 
\end{description}
\end{restatable}

To show the upper bound for Craig $\GMLiu/\MLiu$-separation, suppose $\varphi_{1}$ and $\varphi_{2}$ are $\GMLiu$-formulas, $\sigma = \sig(\varphi_1) \cup \sig(\varphi_2)$ and $\varrho = \sig(\varphi_{1}) \cap \sig(\varphi_{2})$. Our task is to check the existence of pointed $\sigma$-structures $\Amf_{i},a_{i}$, $i=1,2$, and an $\MLiu(\varrho)$-bisimulation $\bis$ between $\Amf_{1}$ and $\Amf_2$ such that $\Amf_{1}\models\varphi_{1}(a_{1})$, $\Amf_{2}\models \varphi_{2}(a_{2})$ and $a_{1}\bis a_{2}$. 
We use the notion of type from Sec.~\ref{sec:gmlu} adapted in the obvious way for $\GMLiu$.

\begin{example}\em\label{ex1sec6}
Consider the $\GMLiu$-formulas
\begin{align*}
& \varphi_1 = B \land \blacksquare \big(B \rightarrow( \neg A \wedge \Diamond^{=1}_{R}\top \wedge \neg \Diamond_{R^{-}}\top \wedge \Diamond_R A)\big) \land{}\\
& \hspace*{1.82cm} \blacksquare \big(A \rightarrow (\Diamond_{R}^{=1}\top \wedge \Diamond_{R^{-}}^{=1}\top \wedge \Diamond_{R}\neg A) \big) \land{}\\ 
& \hspace*{2.8cm} \blacksquare \big( (\neg A \wedge \neg B) \rightarrow (\Diamond_{R}^{=1}\top \wedge \Diamond_{R^{-}}^{=1}\top \wedge \Diamond_{R} A) \big),\\
& \varphi_2 = \Diamond^{=2}_{R}A \land \blacksquare \big(A \rightarrow \Diamond_{R}^{=2}\neg A \wedge (\neg \Diamond_{R^{-1}} B \rightarrow
	\Diamond_{R^{-}}^{=2}\neg A) \big) \land{} \\
& \hspace*{3.9cm}	\blacksquare \big( (\neg A \wedge \neg B) \rightarrow (\Diamond_{R}^{=2}A \wedge \Diamond_{R^{-}}^{=2}A) \big).
\end{align*}	
It is not hard to see that $\varphi_{1}\models \neg \varphi_{2}$. The picture below shows structures $\Amf_1$ and $\Amf_2$ (with all arrows labelled by $R$) that satisfy $\varphi_1$ and $\varphi_2$ at their respective roots. The relation $\bis$ defined by taking $a \bis b$ iff $a$ and $b$ are in the same dashed ellipse is an $\MLiu$-bisimulation between the $\Amf_i$, and so, by Lemma~\ref{criterion}, $\varphi_{1}$ and $\varphi_{2}$ are not $\MLiu$-separable. Note that all $\type(\Amf_1, a_{2n})$ with $n >0$ coincide as well as all $\type(\Amf_1, a_{2n+1})$ with $n > 1$. However, we have $B \in \type(\Amf_1, a_0) \ne \type(\Amf_1, a_2) \notni B$ and $\Diamond_{R^-} B \in \type(\Amf_1, a_1) \ne \type(\Amf_1, a_3) \notni \Diamond_{R^-} B$. Also, all $\type(\Amf_2, e_{2nl})$ with $n > 0$, $l=1,2$ coincide as well as all $\type(\Amf_2, e_{2n+1l})$, for $n > 1$, $l=1,2$. \hfill $\dashv$\\
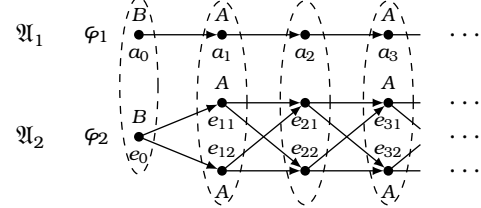
\centerline{
\begin{tikzpicture}[>=latex,line width=0.5pt,xscale = 1.1,yscale = 0.9]
\node[]  at (-1.3,2.5) {$\mathfrak A_1$};
\node[point,fill=black,scale = 0.5,label=left:{$\varphi_1$\hspace*{2mm}\mbox{}},label=above:{\footnotesize $B$},label=below:{\footnotesize $a_0$}] (x1) at (0,2.5) {};
\node[point,fill=black,scale = 0.5,label=above:{\footnotesize $A$},label=below:{\footnotesize $a_1$}] (x2) at (1,2.5) {};
\node[point,fill=black,scale = 0.5,label=below:{\footnotesize $a_2$}] (x3) at (2,2.5) {};
\node[point,fill=black,scale = 0.5,label=above:{\footnotesize $A$},label=below:{\footnotesize $a_3$}] (x4) at (3,2.5) {};
\node[label=right:{$\dots$}] (x5) at (3.5,2.5) {};
\draw[->] (x1) to (x2);
\draw[->] (x2) to (x3);
\draw[->] (x3) to (x4);
\draw[-] (x4) to (x5);
\node[]  at (-1.3,1) {$\mathfrak A_2$};
\node[point,fill=black,scale = 0.5,label=left:{$\varphi_2$\hspace*{2mm}\mbox{}},label=above:{\footnotesize $B$},label=below:{\footnotesize $e_0$}] (y1) at (0,1) {};
\node[point,fill=black,scale = 0.5,label=above:{\footnotesize $A$},label=below:{\footnotesize $e_{11}$}] (y21) at (1,1.5) {};
\node[point,fill=black,scale = 0.5,label=below:{\footnotesize $A$},label=above:{\footnotesize $e_{12}$}] (y22) at (1,0.5) {};
\node[point,fill=black,scale = 0.5,label=below:{\footnotesize $e_{21}$}] (y31) at (2,1.5) {};
\node[point,fill=black,scale = 0.5,label=above:{\footnotesize $e_{22}$}] (y32) at (2,0.5) {};
\node[point,fill=black,scale = 0.5,label=above:{\footnotesize $A$},label=below:{\footnotesize $e_{31}$}] (y41) at (3,1.5) {};
\node[point,fill=black,scale = 0.5,label=below:{\footnotesize $A$},label=above:{\footnotesize $e_{32}$}] (y42) at (3,0.5) {};
\node[label=right:{$\dots$}] (y51) at (3.5,1.5) {};
\node[label=right:{$\dots$}] (y51') at (3.5,1) {};
\node[label=right:{$\dots$}] (y52) at (3.5,0.5) {};
\draw[->] (y1) to (y21);
\draw[->] (y1) to (y22);
\draw[->] (y21) to (y31);
\draw[->] (y22) to (y32);
\draw[->] (y21) to (y32);
\draw[->] (y22) to (y31);
\draw[->] (y31) to (y41);
\draw[->] (y32) to (y42);
\draw[->] (y31) to (y42);
\draw[->] (y32) to (y41);
\draw[-] (y41) to (y51);
\draw[-] (y41) to (y51');
\draw[-] (y42) to (y52);
\draw[-] (y42) to (y51');
\draw[dashed] (0,1.75) ellipse (2.3 mm and 1.3 cm ) ;
\draw[dashed] (1,1.5) ellipse (3 mm and 1.5 cm ) ;
\draw[dashed] (2,1.5) ellipse (3 mm and 1.5 cm ) ;
\draw[dashed] (3,1.5) ellipse (3 mm and 1.5 cm ) ;
\end{tikzpicture}
}
\end{example}

As in Sec.~\ref{sec:gmlu}, $\Amf_{1},\Amf_{2}$ and $\bis$ (if they exist) are assembled from mosaics, which have to be appropriately modified to deal with inverse modalities. To this end, we supply each node with some information about the types of its `parents' and `children'\!\!. Following~\cite[Sec.\ 8.3]{Pratt23book}, we call the result a star-type. 

To give a precise definition, we require a few auxiliary notions. 
A \emph{successor type set} $\frak{s}$  contains, for each type $\type$ and each $S\in \{R,R^{-}\}$ with $R\in \sigma$, exactly one expression of the form $\Diamond_{S}^{=n}\type$, for $0\leq n\leq \mx$, or $\Diamond_{S}^{>\mx}\type$. Let $a\in \dom(\Amf)$, $D\subseteq \dom(\Amf)$ and $n(a,S,\type, D) = |\{b\in D \mid (a,b)\in S^{\Amf},  \type=\type(\Amf,b)\}|$. The \emph{successor type set $\frak{s}(\Amf,a,D)$ defined by $a$ and $D$} contains 
\begin{itemize}
\item[--] $\Diamond_{S}^{=n}\type\in \frak{s}(\Amf,a,D)$ if $n = n(a,S,\type,D) \leq \mx$, and 

\item[--] $\Diamond_{S}^{> \mx}\type\in \mathfrak{s}(\Amf,a,D)$ if $n(a,S,\type,D) > \mx$.
\end{itemize}  
%
%
If we are only interested in the existence of an $S$-successor with a type $\type$, we use $\Diamond_{S}^{>0}\type\in \frak{s}$ as a shorthand for `\mbox{$\Diamond_{S}^{=n}\type\in \frak{s}$} 
for $n >0$ or $\Diamond_{S}^{> \mx}\type\in \mathfrak{s}$'\!. We do not mention elements $\Diamond^{=0}_S \type$ in $\frak{s}$ and  omit $\Amf$ from $\frak{s}(\Amf,a,D)$, $\type(\Amf,a)$, etc., if understood.


A \emph{star-type} takes the form $\ftype = (\type,\frak{s}_{p},\frak{s}_{c})$, where $\type$ is a type,  $\frak{s}_{p}$ a \emph{parent} successor type set, and $\frak{s}_{c}$ a \emph{child} successor type set. To be coherent with the semantics of graded modalities, $\ftype$ should satisfy some natural conditions defined below. 
%
%
The \emph{$S$-profile $P_{S}(\frak{s})$ of $\frak{s}\in \{\frak{s}_{p},\frak{s}_{c}\}$} contains, for any $\chi\in \sub(\varphi_1,\varphi_2)$,
\begin{itemize}
\item[--] $\Diamond_{S}^{=n}\chi$ if $n$ is the sum of all $k$ with $\Diamond_{S}^{=k}\type' \in \frak{s}$ and $\chi\in \type'$;

\item[--] $\Diamond_{S}^{>\mx}\chi$ if this sum is $>\mx$ or $\Diamond_{S}^{>\mx} t'\in \frak{s}$, for some $\type' \ni \chi$.
\end{itemize}
The \emph{$S$-profile $P_{S}(\frak{s}_{p},\frak{s}_{c})$ of $\frak{s}_{p}$ and $\frak{s}_{c}$} contains, for all formulas $\chi\in \sub(\varphi_1,\varphi_2)$,
\begin{itemize}
\item[--] $\Diamond_{S}^{=n}\chi$ if $n=n_{1}+n_{2} \leq \mx$ for $\Diamond_{S}^{=n_{1}}\chi$ in the $S$-profile of $\frak{s}_{p}$ and $\Diamond_{S}^{=n_{2}}\chi$ in the $S$-profile of $\frak{s}_{c}$;

\item[--] $\Diamond_{S}^{>\mx}\chi$ otherwise.
\end{itemize}
A star-type $\ftype=(\type,\frak{s}_{p},\frak{s}_{c})$ is called \emph{coherent} if, for any formula $\Diamond_{S}^{\geq n}\chi\in \sub(\varphi_1,\varphi_2)$, we have $\Diamond_{S}^{\geq n}\chi\in \type$ iff $P_{S}(\frak{s}_{p},\frak{s}_{c})$ contains either $\Diamond_{S}^{=n'}\chi$, for
some $n'\geq n$, or $\Diamond_{S}^{>\mx}\chi$. Henceforth, we assume all star-types to be coherent. 

The \emph{star-type realised by $a\in \dom(\Amf)$ and $D\subseteq \dom(\Amf)$ in $\Amf$} is the triple 
$\ftype(a,D) = \big( \type(a),\frak{s}(a,D),
\frak{s}(a,N(a)\setminus D) \big)$, where $N(a)= \{ b\in \dom(\Amf)\mid (a,b)\in S^{\Amf},\ S\in \{R,R^{-}\},\ R\in \sigma\}$ is the \emph{neighbourhood of $a$ in $\Amf$}.

%



\begin{example}\em 
The picture below illustrates 4 star-types realised 
by $\Amf_1$, $\Amf_2$ from Example~\ref{ex1sec6}, where the parent/child successor type set is on the left/right-hand side of the centre node $\bullet$. \hfill $\dashv$  
\end{example}

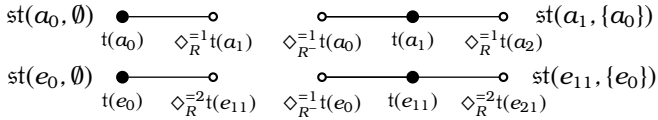
\centerline{
\begin{tikzpicture}[>=latex,line width=0.5pt,xscale = 1.2,yscale = 1]
\node[]  at (-0.8,0) {$\ftype(a_0,\emptyset)$};
\node[point,fill=black,scale = 0.7,label=below:{\footnotesize $\type(a_0)$}] (x1) at (0,0) {};
\node[point,scale = 0.5,label=below:{\footnotesize $\Diamond^{=1}_R \type(a_1)$}] (x2) at (1,0) {};
\draw[-] (x1) to (x2);
\node[]  at (5.2,0) {$\ftype(a_1,\{a_0\})$};
\node[point,fill=black,scale = 0.7,label=below:{\footnotesize $\type(a_1)$}] (y1) at (3.2,0) {};
\node[point,scale = 0.5,label=below:{\footnotesize $\Diamond^{=1}_{R^-} \type(a_0)$}] (y2) at (2.2,0) {};
\node[point,scale = 0.5,label=below:{\footnotesize $\Diamond^{=1}_R  \type(a_2)$}] (y3) at (4.2,0) {};
\draw[-] (y1) to (y2);
\draw[-] (y2) to (y3);
\node[]  at (-0.8,-0.8) {$\ftype(e_0,\emptyset)$};
\node[point,fill=black,scale = 0.7,label=below:{\footnotesize $\type(e_0)$}] (u1) at (0,-0.8) {};
\node[point,scale = 0.5,label=below:{\footnotesize $\Diamond^{=2}_R \type(e_{11})$}] (u2) at (1,-0.8) {};
\draw[-] (u1) to (u2);
\node[]  at (5.2,-0.8) {$\ftype(e_{11},\{e_0\})$};
\node[point,fill=black,scale = 0.7,label=below:{\footnotesize $\type(e_{11})$}] (z1) at (3.2,-0.8) {};
\node[point,scale = 0.5,label=below:{\footnotesize $\Diamond^{=1}_{R^-} \type(e_0)$}] (z2) at (2.2,-0.8) {};
\node[point,scale = 0.5,label=below:{\footnotesize $\Diamond^{=2}_R  \type(e_{21})$}] (z3) at (4.2,-0.8) {};
\draw[-] (z1) to (z2);
\draw[-] (z2) to (z3);
\end{tikzpicture}
}

A \emph{mosaic} is a pair $M=(M_{1},M_{2})$ of sets of star-types such that $A\in \type$ iff $A\in \type'$, for all $(\type,\frak{s}_{p},\frak{s}_{c}), (\type',\frak{s}_{p}',\frak{s}_{c}')\in M_{1}\cup M_{2}$ and $A\in \varrho$, and 
for each $(\type,\frak{s}_{p},\frak{s}_{c})\in M_{1}\cup M_{2}$, all types in $\frak{s}_{p}\cup \frak{s}_{c}$ are $\Du$-equivalent to $\type$.
%
Given structures $\Amf_{1},\Amf_{2}$ with $a\in \dom(\Amf_{j_{1}})$ and $p\in \dom(\Amf_{j_{2}})$, for some $j_1,j_2 \in \{1,2\}$, the \emph{mosaic defined by $a$ and $p$} is 
$
M(\Amf_{1},\Amf_{2},a,p)= \big( M_{1}(\Amf_{1},a,p),M_{2}(\Amf_{2},a,p) \big),
$
where, for $ D_{i}=\{ d\in \dom(\Amf_{i}) \mid \Amf_{i},d\sim_{\MLiu(\varrho)}\Amf_{j_{2}},p\}$, $i=1,2$,
$$
M_{i}(\Amf_{i},a,p)\! =\! \{ \ftype(\Amf_{i},e,D_{i})\mid e\in
\dom(\Amf_{i}), \Amf_{i},e \sim_{\MLiu(\varrho)}\! \Amf_{j_{1}}\!,a\}.
$$
The star-types $\ftype(\Amf_{i},e,D_{i})$ comprising $M_{i}(\Amf_{i},a,p)$ are illustrated below, where all nodes in each dashed   ellipse are $\MLiu(\varrho)$-bisimilar to each other. By the definition, $a\sim_{\MLiu(\varrho)}a'$ and 
$p\sim_{\MLiu(\varrho)}p'$ imply $M(\Amf_{1},\Amf_{2},a,p) = M(\Amf_{1},\Amf_{2},a',p')$.\\[5pt]
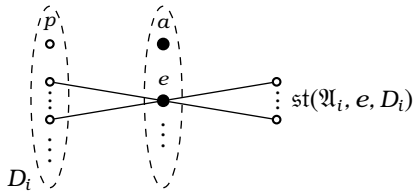
\centerline{
\begin{tikzpicture}[>=latex,line width=0.5pt,xscale = 1,yscale = 1]
\node[point,scale = 0.5,label=above:{\footnotesize $p$}] (x0) at (0,0) {};
\node[point,fill=black,scale = 0.7,label=above:{\footnotesize $a$}] (y0) at (1.5,0) {};
\node[point,fill=black,scale = 0.7,label=above:{\footnotesize $e$}] (y1) at (1.5,-0.75) {};
\node[point,scale = 0.5] (x1) at (0,-0.5) {};
\node[point,scale = 0.5] (x2) at (0,-1) {};
\node[label=right:{.}] (x) at (-0.29,-0.65) {};
\node[label=right:{.}] (x) at (-0.29,-0.75) {};
\node[label=right:{.}] (x) at (-0.29,-0.85) {};
\node[point,scale = 0.5] (z1) at (3,-0.5) {};
\node[point,scale = 0.5] (z2) at (3,-1) {};
\node[label=right:{.}] (z) at (2.72,-0.65) {};
\node[label=right:{.}] (z) at (2.72,-0.75) {};
\node[label=right:{.}] (z) at (2.72,-0.85) {};
\node[label=right:{$\vdots$}] (x3) at (-0.3,-1.3) {};
\node[label=right:{$\vdots$}] (y3) at (1.2,-1.1) {};
\draw[-] (x1) to (y1);
\draw[-] (x2) to (y1);
\draw[-] (z1) to (y1);
\draw[-] (z2) to (y1);
\draw[dashed] (0,-0.7) ellipse (2.5 mm and 1.2 cm ) ;
\draw[dashed] (1.5,-0.7) ellipse (2.5 mm and 1.2 cm ) ;
\node[]  at (-0.4,-1.8) {$D_i$};
\node[]  at (4,-0.75) {$\ftype(\Amf_i,e,D_{i})$};
\end{tikzpicture}
}
\\
In a \emph{root mosaic} $M$, we have $\Diamond_{S}^{=0}\type'\in \frak{s}_{p}$, for all types $\type'$, all $S$, and all $(\type,\frak{s}_{p},\frak{s}_{c})\in M_{1}\cup M_{2}$. 
We obtain a root mosaic $M(\Amf_{1},\Amf_{2},a,\emptyset) = (M_1(\Amf_{1},a,\emptyset),M_2(\Amf_{2},a,\emptyset))$ with $a\in \dom(\Amf_{j})$ and $j \in \{1,2\}$ by taking, for $i = 1,2$,
$$
M_{i}(\Amf_{i},a,\emptyset) = \{ \ftype(\Amf_{i},e,\emptyset)\mid e\in
\dom(\Amf_{i}), \Amf_{i},e \sim_{\MLiu(\varrho)}\Amf_{j},a\}. 
$$

\begin{example}\em\label{ex3sec6}
For $\Amf_1,\Amf_2$ from Example~\ref{ex1sec6} and $n \ge 1$, 
\begin{align*}
& M(\Amf_{1},\Amf_{2},a_0,\emptyset) = \big(\{\ftype(a_0,\emptyset)\}, \{\ftype(e_0,\emptyset)\}\big),\\
& M(\Amf_{1},\Amf_{2},a_1,a_0) = \big(\{\ftype(a_1,\{a_0\})\}, \{\ftype(e_{11},\{e_0\})\} \big),\\
& M(\Amf_{1},\Amf_{2},a_{n+1},a_n) = \big(\{\ftype(a_{n+1},\{a_n\})\}, \{\ftype(e_{(n+1)1},\{e_{n1},e_{n2}\})\} \big).
\end{align*}
For $n \ge 3$, $M(\Amf_{1},\Amf_{2},a_{n+1},a_{n}) = M(\Amf_{1},\Amf_{2},a_{n+3},a_{n+2})$.
\hfill $\dashv$
\end{example}



As in Sec.~\ref{sec:gmlu}, we use mosaics for building  the required structures $\Amf_{1}$ and $\Amf_{2}$. To this end, we need the following notion of witness. Let $M=(M_{1},M_{2})$ be a mosaic, $\mathcal{M}$ a set of mosaics, $\Delta=\Delta_{1}\cup \Delta_{2}$ and $\Gamma=\Gamma_{1} \cup \Gamma_{2}$, where, for $i=1,2$,
\begin{align*}
& \Delta_{i} = \{(\ftype,j)\mid j<\omega ,\ \ftype\in M_{i}\},\\ 
&\Gamma_{i} = \{(\ftype',j',M')\mid j<\omega,\ \ftype'\in M_{i}',\ M'=(M_{1}',M_{2}')\in \mathcal{M}\}.
\end{align*}
%
%
We call $\mathcal{M}$ a \emph{witness} for $M$ if, for each $S\in \{R,R^{-}\}$ with $R\in \sigma$, there exists a relation
$$
S^{M,\mathcal{M}} \subseteq (\Delta_{1} \times \Gamma_{1})\cup (\Delta_{2} \cup \Gamma_{2})
$$
such that the following conditions are satisfied:
\begin{description}
\item[(witness$_1$)] if $(\ftype,j)\in \Delta$, $\ftype=(\type,\frak{s}_{p},\frak{s}_{c})$ and $\Diamond_{S}^{=n}\type'\in \frak{s}_{c}$, then $(\ftype,j)$ has exactly $n$ different $S^{M,\mathcal{M}}$-successors of type $\type'$---i.e., of the form $((\type',\frak{s}_{p}',\frak{s}_{c}'),j',M')$---in $\Gamma$; if $\Diamond_{S}^{>\mx}\type'\in \frak{s}_{c}$, $(\ftype,j)$ has $> \mx$ different $S^{M,\mathcal{M}}$-successors of type $\type'$ in $\Gamma$;


\item[(witness$_2$)] if $(\ftype',k,M')\in \Gamma$, $\ftype'= (\type',\frak{s}_{p}',\frak{s}_{c}')$ and $\Diamond_{S}^{=n}\type\in \frak{s}_{p}'$, then $(\ftype',k,M')$ has exactly $n$ different $S^{M,\mathcal{M}}$-successors of type $\type$ in $\Delta$; if $\Diamond_{S}^{>\mx}\type\in \frak{s}_{p}'$, then $(\ftype',k,M')$ has $> \mx$  $S^{M,\mathcal{M}}$-successors of type $\type$ in $\Delta$;


\item[(bisim$_1$)] if $R\in \varrho$, $S \in \{R,R^-\}$ and some element of $\Delta$ has an $S^{M,\mathcal{M}}$-successor $(\ftype',j',M') \in \Gamma$, then every element of $\Delta$ has an $S^{M,\mathcal{M}}$-successor of the form $(\ftype'',j'',M') \in \Gamma$;

\item[(bisim$_2$)] if $R\in \varrho$, $S \in \{R,R^-\}$ and some $(\ftype',j',M') \in \Gamma$ is an $S^{M,\mathcal{M}}$-successor of an element of $\Delta$, then all $(\ftype'',j'',M')$ in $\Gamma$ are  $S^{M,\mathcal{M}}$-successors of some element of $\Delta$.
\end{description}
Witnesses for mosaics defined by (disjoint) structures $\Amf_{1}$ and  $\Amf_{2}$ can be  extracted as follows. If $M=M(\Amf_{1},\Amf_{2},a,\emptyset)$ is a root mosaic, then 
$$
\mathcal{M}=\{ M(\Amf_{1},\Amf_{2},b,a) \mid b\in N(a'), \text{ for some } a'\sim_{\MLiu(\varrho)} a\}
$$
is a witness for $M$. If $M=M(\Amf_{1},\Amf_{2},a,p)$, for $a,p\in \dom(\Amf_{i})$, then we obtain a witness $\mathcal{M}$ for $M$ by taking 
\begin{multline*}
\mathcal{M}=\{ M(\Amf_{1},\Amf_{2},b,a) \mid \Amf_{i},  b\not\sim_{\MLiu(\varrho)}\Amf_{i},p,\  b \in N(a'),\\ \text{ for some } a'\sim_{\MLiu(\varrho)} a\}.
\end{multline*}
Alternative and possibly larger witnesses can be obtained by dropping the restrictions `$b \in N(a')$, for some $a'\sim_{\MLiu(\varrho)} a$'\!.

\begin{example}\em
Below are the mosaic-witness pairs defined by $\Amf_1$ and $\Amf_2$ from Example \ref{ex1sec6} (see also Example~\ref{ex3sec6}):
\begin{align*}
& 
M(\Amf_{1},\Amf_{2},a_0,\emptyset), \ \{M(\Amf_{1},\Amf_{2},a_1,a_0)\},\\
& 
M(\Amf_{1},\Amf_{2},a_n,a_{n-1}), \ \{M(\Amf_{1},\Amf_{2},a_{n+1},a_n)\}, \text{ for } n \ge 1.
\end{align*}
Relations $R^{M,\mathcal{M}}$ for these pairs $M$, $\mathcal{M}$ can be read from $\Amf_1,\Amf_2$ and do not need infinitely-many copies of star-types from the definition of witness. Those copies will prove to be a convenient technical tool later on in this section. \hfill $\dashv$ 
\end{example}

Compared to the mosaic construction in Sec.~\ref{sec:gmlu}, it is much harder now to show that $(\dag)$ it suffices to use mosaics and witnesses of exponential size and that $(\ddag)$ checking whether a given exponential set of such mosaics is a witness for a given mosaic can be done in double exponential time. Here, by the \emph{size} of a mosaic $M = (M_1,M_2)$ we mean $|M| = |M_1| + |M_2|$.

Call successor sets $\frak{s}$ and $\frak{s}'$ \emph{profile-equivalent} if they have the same $S$-profiles for all $S\in \{R,R^{-}\}$ with $R\in \sigma$. Star-types $\ftype=(\type,\frak{s}_{p},\frak{s}_{c})$ and $\ftype'=(\type',\frak{s}_{p}',\frak{s}_{c}')$ are \emph{parent profile-equivalent} if $\type=\type'$ and $\frak{s}_{p}$, $\frak{s}_{p}'$ are profile equivalent; $\ftype$ and $\ftype'$ are \emph{child profile-equivalent} if $\type=\type'$ and $\frak{s}_{c}$, $\frak{s}_{c}'$ are profile equivalent.
Mosaics $M=(M_{1},M_{2})$ and $M'=(M_{1}',M_{2}')$ are \emph{child profile-equivalent} if, for any $\ftype \in M_{i}$, there exists a child-profile equivalent $\ftype'\in M_{i}'$, and vice versa, for $i=1,2$; $M$ and $M'$ are \emph{parent profile-equivalent} if, for any $\ftype \in M_{i}$, there exists a parent-profile equivalent $\ftype'\in M_{i}'$, and vice versa, for $i=1,2$.
Observe that the number of child/parent profile non-equivalent star-types does not exceed $(|\varphi_1|+|\varphi_2|)^{4(|\varphi_1|+|\varphi_2|)}$. The following lemma, proved in the appendix, establishes property $(\dag)$ above using a syntactic characterisation of witnesses, which is efficient in the sense that it also guarantees $(\ddag)$:

\begin{restatable}{lemma}{expcardinality}
\label{lem:cardin}
There is a polytime computable number \mbox{$k_{\varphi_{1}\varphi_{2}} = O(2^{|\varphi_{1}|+|\varphi_{2}|})$} such that, for any mosaic $M=(M_1,M_2)$ and witness $\mathcal{M}$ for $M$, there exist a mosaic $N=(N_1,N_2)$ and a witness $\mathcal{N}$ for $N$ satisfying the following\textup{:}
\begin{itemize}
\item[--] $|N|$, $|\mathcal{N}|$, and all $|N'|$ with $N' \in \mathcal{N}$ do not exceed $k_{\varphi_{1}\varphi_{2}}$,


\item[--] $M$ and $N$ are parent profile-equivalent,

\item[--] every $N'\in\mathcal{N}$ has a child profile-equivalent $M'\in \mathcal{M}$,

\item[--] $N_i \subseteq M_i$ and, for every $N'=(N_1',N_2')\in \mathcal{N}$, there is $M'=(M_1',M_2')\in \mathcal{M}$ with $N_i'\subseteq M_i'$, for $i=1,2$.
\end{itemize}
\end{restatable}

Thus, in our mosaic elimination procedure establishing the $2\ExpTime$ upper bound in Theorem~\ref{thm:gmliu}, instead of any mosaic $M$ and its witness $\mathcal{M}$ we can use exponential-size $N$ and $\mathcal{N}$ provided by Lemma~\ref{lem:cardin}.
Let $\mathfrak{S}_{0}$ be the set of all mosaics $M=(M_{1},M_{2})$ with non-empty $M_{1}$, $M_{2}$ and $|M|\leq k_{\varphi_1\varphi_2}^{2}$. Clearly $\mathfrak{S}_{0}$ can be computed in double exponential time in $|\varphi_1|+|\varphi_2|$. 

Let $\mathfrak{S}\subseteq \mathfrak{S}_{0}$. We call $M\in \mathfrak{S}$ \emph{bad} in $\mathfrak{S}$ if at least one of the following conditions does not hold:
\begin{itemize}
\item[--] there is a witness $\mathcal{M}\subseteq \mathfrak{S}$ for $M$ with $|\mathcal{M}| \leq k_{\varphi_1\varphi_2}$;

\item[--] if $\Du\chi\in \type$ and $(\type,\frak{s}_{p},\frak{s}_{c})\in M_{i}$, for $(M_{1},M_{2})\in \mathfrak{S}$, $i \in \{1,2\}$, then there is a root mosaic $(M_{1}',M_{2}')\in \mathfrak{S}$ such that $\chi\in \type'$, for some $(\type',\frak{s}_{p}',\frak{s}_{c}')\in M_{i}'$ with $\Du$-equivalent $\type$, $\type'$.
\end{itemize}
Using ($\dag)$ and $(\ddag)$, we now compute in double exponential time a sequence $\mathfrak{S}_{0},\mathfrak{S}_{1},\dots$, where each $\mathfrak{S}_{i+1}$ eliminates all bad mosaics from $\mathfrak{S}_{i}$, stop when $\mathfrak{S}_{n+1} = \mathfrak{S}_{n}$ and set $\mathfrak{S}^{\ast} = \mathfrak{S}_{n}$. 

\begin{restatable}{lemma}{mainlemmaforgmliu}
\label{lem:gmliuall} 
The following conditions are equivalent\textup{:}
\begin{description}
\item[\rm $(i)$] there exist pointed $\sigma$-structures $\Amf_{i},a_{i}$, for $i=1,2$, such that   $\Amf_i \models \varphi_{i}(a_i)$ and $\Amf_{1},a_{1}\sim_{\MLiu(\varrho)}\Amf_{2},a_{2}$\textup{;}

\item[\rm $(ii)$] there exist a root mosaic $M^* = (M^*_{1},M^*_{2})\in \mathfrak{S}^{\ast}$ with $\varphi_{1}\in \type$ and $\varphi_{2}\in \type'$, for some  $(\type,\frak{s}_{p},\frak{s}_{c})\in M^*_{1}$ and $(\type',\frak{s}_{p}',\frak{s}_{c}')\in M^*_{2}$.
\end{description} 
\end{restatable}
\begin{proof}
The proof of $(i) \Rightarrow (ii)$ needs more work compared to the proof of Lemma~\ref{lem:main1}. Given $\Amf_{i},a_{i}$, $i=1,2$, satisfying $(i)$, we consider the set of all pairs $p=(M, \mathcal{M}_{M})$ such that $\mathcal{M}_{M}$ is a witness for $M$ given by the $\Amf_{i}$. The union of all such $\{M\}\cup \mathcal{M}_{M}$ would not contain bad mosaics if $M$ and $\mathcal{M}_{M}$ were of size $\leq k_{\varphi_1\varphi_2}$, which is not necessarily the case. For $p=(M, \mathcal{M}_{M})$ of the wrong size, we use Lemma~\ref{lem:cardin} to  obtain $(M^{p},\mathcal{M}_{M}^{p})$ satisfying the conditions of that lemma. However, these pairs are not necessarily defined by the $\Amf_{i}$. We show in the appendix that one can construct from the $(M^{p},\mathcal{M}_{M}^{p})$ a set $\mathfrak{S}\subseteq \mathfrak{S}_{0}$ that does not contain any bad mosaics.

$(ii) \Rightarrow (i)$ 
%
%
As in the proof of Lemma~\ref{lem:main1}, we build the required $\Amf_{i},a_{i}$, $i=1,2$, and $\MLiu(\varrho)$-bisimulation $\bis$ between them using $\mathfrak{S}^*$ provided by $(ii)$, though the construction is somewhat different. For every $M\in \mathfrak{S}^{\ast}$, take a witness $\mathcal{M}_{M}\subseteq \mathfrak{S}^{\ast}$ for $M$ together with the corresponding relations $S^{M,\mathcal{M}_{M}} \subseteq \Delta^{M} \times \Gamma^{M}$ between $\Delta^{M}=\Delta_1^M \cup \Delta_2^M$ and $\Gamma^{M}=\Gamma_1^M \cup \Gamma_2^M$.
For $i=1,2$, the domain $\dom(\Amf_{i})$ of $\Amf_{i}$ comprises the words of the form 
\begin{equation}\label{eq1second+}
s = M_{0} \dots M_{n-1} (\ftype_{n},j_{n},M_{n}),
\end{equation}
%
	%
where $M_{j}=(M_{j,1},M_{j,2})\in \mathfrak{S}^{\ast}$, $0 \le j \le n$, $M_{0}$ is a root mosaic, with any type in $M_{0,i}$ being $\Du$-equivalent to any type in $M^{\ast}_{i}$, $M_{j+1}\in \mathcal{M}_{M_{j}}$, for all $j < n$, $\ftype_{n} \in M_{n,i}$ and $j_{n}<\omega$. 
It follows from the definition of (bad) mosaics that all types used in the construction of  $\dom(\Amf_{i})$ are $\Du$-equivalent to each other.

Let $\type_n$ be the type of $\ftype_{n}$, in which case $\type_n$ is also called the \emph{type} of $s$. We interpret the propositional variables $A \in \sigma$ by taking $s\in A^{\Amf_{i}}$ if $A\in \type_{n}$. For any $R\in \sigma$, we set 
\begin{itemize}
\item[--] $(s',s)\in R^{\Amf_{i}}$ for any $s$ in \eqref{eq1second+} and any $s'$ of the form
\begin{equation}\label{eq1second'+}
s' = M_{0} \dots M_{n-2} (\ftype_{n-1},j_{n-1},M_{n-1})
\end{equation}
such that $(\ftype_{n-1},j_{n-1})R^{M_{n-1},\mathcal{M}_{M_{n-1}}}(\ftype_{n},j_{n},M_{n})$;

\item[--] $(s,s')\in R^{\Amf_{i}}$ for any $s$ of the form \eqref{eq1second+} and any $s'$ of the form \eqref{eq1second'+}
such that $(\ftype_{n-1},j_{n-1}){(R^{-})}^{M_{n-1},\mathcal{M}_{M_{n-1}}}(\ftype_{n},j_{n},M_{n})$.
\end{itemize}
%
%
%
%
One can show by induction (see the appendix for details) that, for any $\chi\in\sub(\varphi_1,\varphi_2)$ and any $s$ in \eqref{eq1second+}, we have $\Amf_{i}\models \chi(s)$ iff $\chi\in \type_{n}$. 



We next define $\bis\subseteq \dom(\Amf_{1})\times\dom(\Amf_{2})$ by taking $s\bis s'$, for $s \in \dom(\Amf_{1})$ of the from \eqref{eq1second+} and $s'\in \dom(\Amf_{2})$ of the form
\begin{equation*} 
s'= M_{0}'\dots M_{m-1}'(\ftype_{m}',j_{m}',M_{m}')
\end{equation*}
if $n=m$ and $M_{j}=M_{j}'$ for all $j\leq n$. It follows from non-emptiness of mosaics in $\mathfrak{S}^*$, the definition of  $R^{\Amf_{i}}$ and conditions {\bf (bisim$_1$)} and {\bf (bisim$_2$)} that $\bis$ is a $\MLiu(\varrho)$-bisimulation. 
\end{proof}

Lemmas~\ref{criterion}, \ref{l:red}, \ref{lem:gmliuall} give the $\TwoExpTime$ upper bound for deciding $\GMLiu$/$\MLiu$-Craig separation in Theorem~\ref{thm:gmliu}.
%
%
%
We obtain the $\TwoExpTime$ upper bound for $\GMLiu/\MLi$ by dropping the condition that $M_{1}$ and $M_{2}$ in $M=(M_{1},M_{2})$ are both non-empty. The one for $\GMLiu/\MLu$ is obtained by treating the inverse $R^{-}$ with $R\in \varrho$ in the same way as $R\in \sigma\setminus\varrho$. Both of these modifications give the upper bound for $\GMLiu/\ML$.

It remains to establish the $\coNExpTime$ upper bounds for $\GMLi/\MLi$  and $\GMLi/\ML$.
%
%
%
%
%
Suppose $\varphi_{1},\varphi_{2}$ are $\GMLi$-formulas. 
We again ignore everything related to $\Du$ and admit mosaics $M = (M_{1},M_{2})$ with $M_{1}=\emptyset$ or $M_{2}=\emptyset$ (but not both). Let $m = \max \{\md(\varphi_1), \md(\varphi_2)\}$. Then it follows from the proof of Lemma~\ref{lem:gmliuall} and the definition of $\bis$ in it that there exist $\Amf_{i},a_{i}$, $i=1,2$, with $\Amf_i \models \varphi_{i}(a_{i})$ and $\Amf_{1},a_{1} \sim_{\MLi(\varrho)} \Amf_{2},a_{2}$ iff there are words of the form 
$s = M_{0}M_{1}\dots M_{n}$ with $n\leq \md(\varphi,\psi)$ and  $M_{j}=(M_{j,1},M_{j,2})$ such that
\begin{itemize}
\item[--] $M_{j}$ is a mosaic of size $\leq k_{\varphi_{1}\varphi_{2}}^{2}$, for $j\leq n$;

\item[--] $M_{0}=M^{\ast}$, for a fixed root mosaic $M^*=(M_{1}^{\ast},M_{2}^{\ast})$ such that there are $(\type,\frak{s}_{p},\frak{s}_{c})\in M_{1}^{\ast}$ and $(\type',\frak{s}_{p}',\frak{s}_{c}')\in M_{2}^*$ with $\varphi_{1}\in \type$ and $\varphi_{2}\in \type'$;

\item[--] $M_{j+1}\in \mathcal{M}_{M_{j}}$, for a witness $\mathcal{M}_{M_{j}}$ for $M_{j}$ with $|\mathcal{M}_{M_{j}}|\leq k_{\varphi_1 \varphi_2}$ and all $j< n$.
\end{itemize}
A \NExpTime{} algorithm deciding non-separation could guess words of the form $s = M_{0}M_{1}\dots M_{n}$ and then check in exponential time that they satisfy the three conditions above. 

The matching lower bounds are proved as in the previous section. This completes the proof of Theorem~\ref{thm:gmliu}.


\section{Definability}\label{sec:definability}

In this section, we show that, for all pairs $\L /\LS$ of logics in Table~\ref{table:results}, the $\L/\LS$-definability problem can be reduced in polynomial time to validity in $\L$. The complexity results for definability in Table~\ref{table:results} follow then from the known complexities of $L$-validity. In the proof, given any $\L$-formula $\varphi$, we construct an $\LS$-formula $\Flat(\varphi)$---mostly of polynomial DAG-size---such that $\Flat(\varphi)$ is a definition of $\varphi$ iff $\varphi$ is $\LS$-definable. 
In the cases of $\CT/\LS$-definability with $\LS\in \{\MLiu,\MLi,\ML\}$,
equivalent polynomial DAG-size formulas do not always exist. We nevertheless obtain polytime reductions to $\CT$-validity by 
working with a succinct version of $\MLiu$.



Let $\Amf = (\dom(\Amf),(R^{\Amf})_{R\in \sigma},(c^{\Amf})_{c\in \sigma})$ be a $\sigma$-structure and let $\mathbb{N}^{\infty} = \{1,2, \dots\}\cup \{\omega\}$ with $k<\omega$, for all $k\in \mathbb{N}$. For $\kappa\in \mathbb{N}^{\infty}$, the \emph{$\kappa$-expansion} $\Amf^\kappa$ of $\Amf$ is defined by taking
\begin{itemize}
	\item[--] $\dom(\Amf^\kappa) = \bigcup_{a\in \dom(\Amf)}\bar{a}$, where
	$\bar{a} = \{(a,0)\}$ if $a=c^{\Amf}$, for some constant $c$, and $\bar{a} = \{(a,i)\mid 0\leq i < \kappa\}$ otherwise;
	
	\item[--] $A^{\Amf^\kappa} = \bigcup_{a\in A^\Amf} \bar{a}$, for a unary predicate $A\in \sigma$;
	
	\item[--] $R^{\Amf^\kappa} = \bigcup_{(a,b)\in R^\Amf}\bar{a}\times \bar{b}$, for a binary predicate $R\in \sigma$;
	
	\item[--] $c^{\Amf^\kappa} = (c^\Amf,0)$, for a constant $c\in \sigma$.
\end{itemize}
In other words, $\Amf^\kappa$ creates $\kappa$-many copies of each element of $\dom(\Amf)$ except for those that interpret constants in $\sigma$, in which case we only keep a single copy. The next lemma follows directly from the definition by a straightforward induction:

\begin{lemma}\label{lemma:bisim_omega}
	If $\L$ is equality-free $\FO$ or a fragment of $\MLinu$, then, for any  pointed $\sigma$-structure $\Amf,a$ and any $i<\kappa$, we have $\Amf,a \equiv_{\L(\sigma)} \Amf^\kappa, (a,i)$.
\end{lemma}

%


We remind the reader that, in this paper, $\sigma$ in $\L(\sigma)$ may contain constants only if $\L$ is a modal logic with nominals. We first consider the case when $\L$ ranges over equality-free $\CT$ and nominal-free modal logics.
The \emph{flattening} $\Flat(\varphi)$ of a formula $\varphi$ in equality-free $\CT$ is defined by replacing every $\exists^{\geq k}x$ with $\exists x$; similarly, the \emph{flattening} $\Flat(\varphi)$ of a $\GMLiu$-formula $\varphi$ 
is defined by replacing 
%
%
every $\Diamond^{\ge k}_R$ in $\varphi$ with $\Diamond_R$. 
As before, we denote by $\mx$ the maximal number $k$ such that $\exists^{\geq k}x$ or, respectively, some $\Diamond_R^{\ge k}$ occurs in the given formula $\varphi$. The next observation follows directly from the definitions: 

\begin{restatable}{lemma}{lemmaequivalenceomega}
	\label{lemma:equivalenceomega}
Suppose that $\sigma$ is constant-free, $\varphi$ is a $\GMLiu(\sigma)$-formula or $\varphi(x)$ is an equality-free $\CT(\sigma)$-formula, and that $\mx<\kappa<\kappa'\in \mathbb{N}^{\infty}$. Then \mbox{$\Amf^\kappa \models \varphi(a,i)$} iff $\Amf^\kappa \models \Flat(\varphi)(a,i)$
	iff $\Amf^{\kappa'} \models \Flat(\varphi)(a,i)$ iff \mbox{$\Amf^{\kappa'}\models \varphi(a,i)$}, for all pointed $\sigma$-structures $\Amf,a$ and all $i <\kappa$.
\end{restatable}

We can now show that $\Flat(\varphi)$ is an $\L'$-formula defining $\varphi$ whenever such an $\L'$-definition exists.

\begin{theorem}\label{prop:definability1}
	Let $\L/\LS$ be either $\CT/\FOT$ with equality-free $\CT$ and $\FOT$, or any pair of modal nominal-free logics from Table~\ref{table:results}. Then an $\L$-formula $\varphi$ is $\LS$-definable iff $\models \varphi \leftrightarrow \Flat(\varphi)$.
\end{theorem}
\begin{proof}
$(\Rightarrow)$ Suppose $\varphi$ is definable by an $\LS$-formula $\varphi'$ and $\sigma=\sig(\varphi)\cup \sig(\varphi')$. By definition, $\sigma$ contains no constants. Take any pointed $\sigma$-structure $\Amf,a$. By the choice of $\varphi'$ and Lemmas~\ref{lemma:bisim_omega} and~\ref{lemma:equivalenceomega}, $\Amf\models \varphi(a)$ iff $\Amf\models \varphi'(a)$ iff $\Amf^\omega \models \varphi'(a,0)$ iff $\Amf^\omega \models \varphi(a,0)$ iff $\Amf^\omega \models \Flat(\varphi)(a,0)$ iff $\Amf \models \Flat(\varphi)(a)$, which yields  $\models \varphi\leftrightarrow \Flat(\varphi)$. The implication $(\Leftarrow)$ is trivial.
\end{proof}

By using the criterion of Theorem~\ref{prop:definability1}, one can readily verify the claims about definability stated in  Example~\ref{ex:gmlml}.  
%
%

As $\Flat(\varphi)$ is constructed in polytime, we obtain:

\begin{corollary}
	Let $\L/\LS$ be either $\CT/\FOT$ without equality or any pair of modal nominal-free logics from Table~\ref{table:results}.
	Then $\L/\LS$-definability is polynomial-time reducible to $\L$-validity.
\end{corollary}

We apply a similar technique to the modal logics with nominals, which  requires a slightly more complicated transformation than merely replacing $k$ in $\Diamond\!^{\ge k}_R$ by $1$. The reason is the definition of $\bar{a}$ in $\Amf^\kappa$, for $a = c^{\Amf}$ with $c \in \sigma$, where $a$ does not have multiple copies unlike the non-constant elements.  

Let $\varphi$ be a $\GMLinu(\sigma)$-formula. Denote by $\sigma_{c}$ the set of constants in $\sigma$, assuming without loss of generality that $\sigma_{c}$ is finite. Define the (\emph{nominal aware}) \emph{flattening} $\Flat_{\sigma_c}(\varphi)$ of $\varphi$ in the same way as $\Flat(\varphi)$ except that 
\begin{align*}
	&\Flat_{\sigma_{c}}(\Diamond\!^{\ge k}_R \varphi) = \Diamond_R \big( \Flat_{\sigma_{c}}(\varphi) \land \bigwedge_{c\in \sigma_{c}} \neg N_c \big)\vee{} \\
	&\hspace{1.7cm} \bigvee_{\varrho \subseteq\sigma_{c}, |\varrho|=k} \ \bigwedge_{c\in \varrho}\Diamond_R \big( \Flat_{\sigma_{c}}(\varphi) \land N_c\wedge \hspace*{-2mm} \bigwedge_{d\in \varrho\setminus \{c\}} \hspace*{-2mm} \neg N_d \big).
\end{align*} 
%
It is straightforward to lift Lemma~\ref{lemma:equivalenceomega} and 
Theorem~\ref{prop:definability1} to pairs $\L/\LS$ from Table~\ref{table:results} with $\L$ admitting nominals and the function $\Flat_{\sigma_{c}}(\varphi)$. In this case, we do not immediately get a polytime reduction of definability to validity since $\Flat_{\sigma_{c}}(\Diamond\!^{\ge k}_R\varphi)$ contains exponentially-many disjuncts and 
because $\Flat_{\sigma_{c}}(\varphi)$ occurs more than once in $\Flat_{\sigma_{c}}(\Diamond\!^{\ge k}_R\varphi)$. A polytime reduction is obtained, however, by replacing the second disjunct of $\Flat_{\sigma_{c}}(\Diamond\!^{\ge k}_R\varphi)$ by the equivalent $\L$-formula
$\Diamond_{R}^{\geq k} \big(\Flat_{\sigma_{c}}(\varphi) \wedge \bigvee_{c\in \sigma_{c}}N_{c}\big)$ and by working with a DAG-representation.
\begin{corollary}
	Let $\L/\LS$ be any pair of modal logics from Table~\ref{table:results} such that $L$ admits nominals. Then $\L/\LS$-definability is polytime reducible to $\L$-validity.
\end{corollary}
We next consider the pair $\CT/\FOTNE$. Observe that the flattening defined above does not work now even if we admit equality in both logics. For example, $\varphi=\exists^{\geq 2}y \, \big(x=y \wedge A(y)\big)$ is unsatisfiable (and so equivalent to an $\FOTNE$-formula) but it is not equivalent to $\exists y\, \big(x=y \wedge A(y)\big)$. 
To define $\Flat(\exists^{\geq k}x \, \varphi)$ in this case so that Lemma~\ref{lemma:equivalenceomega} holds, we distinguish between $k=1$ and $k\geq 2$, decomposing $\exists^{\geq k}x \, \psi$ into  
$$
\exists^{\geq k}x \, \big[ \big((x=y) \wedge \psi\big) \vee \big((x\not=y) \wedge \psi\big) \big],
$$ 
where $\{x,y\} = \var$. 
For $k=1$, we rewrite the disjuncts separately into formulas without equality, and for $k=2$, we only take a flattening of the second disjunct without equality; details are given in the appendix.  One can then again lift Theorem~\ref{prop:definability1} and obtain a polytime
reduction of $\CT/\FOTNE$-definability to $\CT$-validity.
We note that the decidability status of $\CT/\FOT$-definability remains open. 

In contrast to separation, the upper bounds for definability shown above cannot be easily extended from $\L/\LS$ to $\L/\L''$, where $\ML\subseteq \L''\subseteq \LS$.
A fundamental difference from the pairs of logics considered above is that  equivalent formulas in the weaker logic are sometimes exponentially larger even in DAG-size as demonstrated by the following:

\begin{restatable}{lemma}{lemsuccinctness}
	\label{lem:succinctness}
The formulas  
$$
\varphi_{n}(x) = \exists y \, \big[ R(x,y) \wedge \bigwedge_{1\leq i \leq n} \big(A_{i}(x) \leftrightarrow A_{i}(y) \big) \big], \ \ n < \omega,
$$ 
are $\ML$-definable, but any equivalent $\MLiu$-formulas are of exponential DAG-size.
\end{restatable}
\begin{proof}
	Suppose otherwise, that is, $\chi_{n}\in \MLiu$ is of polynomial DAG-size 
	and $\Amf\models \varphi_{n}(a)$ iff $\Amf\models \chi_{n}(a)$ holds for all $\Amf,a$.
	Consider the model $\Amf_{n}=(W_{n},R^{\Amf_{n}},A_{0}^{\Amf_{n}},\ldots,A_{n-1}^{\Amf_{n}})$ with domain $W_{n}=\{w\} \cup X_{n} \cup X_{n}'$, where  
	$X_{n}=\{0,\ldots,2^{n}-1\}$ and $X_{n}'=\{0,\ldots,2^{n}-1\}\times \{1\}$, and
	\begin{itemize}
		\item[--] $R^{\Amf_{n}} = (\{w\} \times X_{n}) \cup (X_{n} \times X_{n}')$,
		
		\item[--] $A_{i}^{\Amf_{n}}$ is such that the nodes $k\in X_{n}$ and $(k,1)\in X_{n}'$ satisfy $k$ encoded in binary using $A_{0},\ldots,A_{n-1}$.   
	\end{itemize}
	Clearly, $\Amf_{n}\models \Box_{R}\varphi_{n}(w)$, and so $\Amf_{n}\models \Box_{R}\chi_{n}(w)$. Then, for every $\Diamond_{R} \chi\in \text{sub}(\chi_{n})$ such that $\Amf_{n}\models \Diamond_{R}\chi(v)$ for some $v\in X_{n}$, we can select from $X_{n}'$ a $v_{\chi}\in X_{n}'$ with $\Amf_{n}\models\chi(v_{\chi})$. Denote the set of selected nodes by $X_{n}''$.
	Then the restriction of $\Amf_{n}$ to the set $\{w\}\cup X_{n}\cup X_{n}''$ satisfies $\Box_{R}\chi$ and $X_{n}''\subsetneq X_{n}'$ for sufficiently large $n$ as the size of $X_{n}''$ is bounded by the number of subformulas of $\chi_{n}$. By the semantics of $\Box_R\varphi_{n}$, we clearly have
	$\Amf_{n}\not\models\Box_{R}\varphi_{n}(w)$ for all $n>0$, which is  a contradiction. 
\end{proof}
We now introduce a more succinct version of $\MLiu$, which is included in the two-variable guarded fragment of first-order logic (and so satisfiability is still \ExpTime-complete). Fix variables $x,y$. Then the formulas of suc$\MLiu$, \emph{succinct $\MLiu$}, are constructed from atoms $A(z)$ and $\top(z)$ with $z\in \{x,y\}$ 
by taking the closure under the following rules:
\begin{enumerate}
	\item if $\varphi,\psi\in \text{suc}\MLiu$, then $\neg\varphi,\varphi\wedge \psi \in \text{suc}\MLiu$;
	
	\item if $\chi(x,y)\in \text{suc}\MLiu$ and $x,y$ are both free in $\chi(x,y)$, then $\exists y\, (R(x,y) \wedge \chi(x,y))\in \text{suc}\MLiu$, $\exists x\, (R(y,x) \wedge \chi(x,y))\in \text{suc}\MLiu$;
	
	\item if $\chi(x,y)\in \text{suc}\MLiu$ and $x,y$ are both free in $\chi(x,y)$, then $\exists y\, (R(y,x) \wedge \chi(x,y))\in \text{suc}\MLiu$, $\exists x\, (R(x,y) \wedge \chi(x,y))\in \text{suc}\MLiu$;
	\item if $\chi(x,y)\in \text{suc}\MLiu$ and $x,y$ are both free in $\chi(x,y)$, then
	$\exists y \, \chi(x,y)\in \text{suc}\MLiu$ and $\exists x \, \chi(x,y)\in \text{suc}\MLiu$.
\end{enumerate} 
The languages suc$\MLi$ and suc$\MLu$ are obtained by dropping items 4) and 3), respectively. The language suc$\ML$ is defined by dropping 3) and 4).

\begin{restatable}{lemma}{lemtransss}
	\label{lem:lemtransss}
	Let $\L\in \{\ML,\MLi,\MLu,\MLiu\}$. Then suc$\L$ and $\L$ have the same expressive power in the sense that
	\begin{enumerate}
		\item[$(i)$] for every $\varphi\in \L$, one can construct in polynomial time a formula $\varphi'(x)\in \text{suc}\L$
		such that $\Amf\models \varphi(a)$ iff $\Amf\models \varphi'(a)$ for all pointed $\Amf,a$\textup{;}
		
		\item[$(ii)$] for every $\varphi(x)\in \text{suc}\L$, one can construct in exponential time a formula $\varphi'\in \L$
		such that $\Amf\models \varphi(a)$ iff $\Amf\models \varphi'(a)$ for all pointed $\Amf,a$.
	\end{enumerate}
\end{restatable}

Now, for every formula $\varphi(x)\in \FOTNE$ with exactly one free variable $x$,  we define inductively the \emph{modal reduct} $m(\varphi(x))$ in suc$\MLiu$ 
by taking $m(A(x))=A(x)$, $m(\neg \varphi) = \neg m(\varphi)$, $m(\varphi\wedge\psi) = m(\varphi)\wedge m(\psi)$, and by defining $m(\exists x \, \psi)$ as follows. Suppose $\psi$ is a Boolean combination $\beta$ of the form 
$$
\psi = \beta\big(\rho_{1},\ldots,\rho_{k_{1}}, \gamma_{1}(x),\dots,\gamma_{k_{2}}(x),\xi_{1}(y),\dots,\xi_{k_{3}}(y)\big),
$$
where
\begin{itemize}
	\item[--] the $\rho_{i}$ are binary atoms $R(x,y)$, $R(y,x)$; 
	
	\item[--] the $\gamma_{i}(x)$ are unary atoms of the form $A(x)$, or $R(x,x)$, or a formula of the form $\exists y\, \gamma_{i}'$; 
	
	\item[--] the $\xi_{i}(y)$ are unary atoms of the form $A(y)$, or $R(y,y)$, or a formula of the form $\exists x \, \xi_{i}'$.
\end{itemize}
We may assume that $x$ and $y$ are free in $\psi$. For $R\in \sig(\varphi)$, we construct from $\psi$ the formulas 
\begin{itemize}
	\item[--] $\psi_{R(x,y)}$ by replacing all $\rho_{i}$ of the form $S(x,y)$ and $S(y,x)$, for $S\not=R$, or of the form $R(y,x)$ with $\neg\top(x)$, all $\gamma_{i}$ of the form $S(x,x)$ with $\neg\top(x)$, all $\xi_{i}$ of the form $S(y,y)$ with $\neg\top(y)$, and all $\rho_{i}$ of the form $R(x,y)$ with $\top(x)$;
	
	\item[--] $\psi_{R(y,x)}$ by replacing all $\rho_{i}$ of the form $S(x,y)$ and $S(y,x)$, for $S\not=R$, or of the form $R(x,y)$ with $\neg \top(x)$, all $\gamma_{i}$ of the form $S(x,x)$ with $\neg\top(x)$, all $\xi_{i}$ of the form $S(y,y)$ with $\neg\top(y)$, and replace all $\rho_{i}$ of the form $R(y,x)$ with $\top(x)$;
	
	\item[--] $\psi_{N}$ by replacing all $\rho_{i}$ of the form $S(x,y)$ and $S(y,x)$ and all
	$\gamma_{i}$ of the form $S(x,x)$ and all $\xi_{i}$ of the form $S(y,y)$ with $\neg\top(x)$. 
\end{itemize}
Next, consider the formula 
\begin{multline*}
	\psi_{1} = \bigvee_{R\in \sig(\varphi)} \big (R(x,y) \wedge \psi \big) 
	\vee \bigvee_{R\in \sig(\varphi)} \big( R(y,x) \wedge \psi\big) \vee{} \\ 
	\big[\psi \wedge \bigwedge_{R\in \sig(\varphi)} \big(\neg R(x,y) \wedge \neg R(y,x) \big) \big],
\end{multline*}
which is equivalent to $\psi$. Then we obtain $m(\exists x\, \psi)$ from 
$\exists x\,\psi_{1}$ by setting: 
\begin{multline*}
	m(\exists x\,\psi) =  \bigvee_{R\in \sig(\varphi)}\exists x\, \big( R(x,y) \wedge m(\psi_{R(x,y)})\big) \vee{} \\
\bigvee_{R\in \sig(\varphi)} \exists x\, \big( R(y,x) \wedge m(\psi_{R(y,x)}) \big) \vee  \exists x\, m(\psi_{N}).
\end{multline*}
Clearly, $m(\varphi)$ is in suc$\MLiu$ and of polynomial DAG-size in the size of $\varphi$. We sketch how to prove that $\varphi(x)$ is suc$\MLiu$-definable iff $\models \varphi \leftrightarrow m(\varphi)$:
given any structure $\mathfrak{A}$, one
takes its \emph{twin-unfolding} $\Amf^{\ast}$ that  consists of the standard unfoldings of $\Amf$ along binary relations and their inverses, with at least two unrelated copies of any node. The structures $\Amf$ and $\Amf^{\ast}$ cannot be distinguished in suc$\MLiu$. In fact, one can prove that $\Amf^{\ast}\models \varphi(a)$ iff $\Amf^{\ast}\models m(\varphi)(a)$, for all pointed $\Amf,a$. Then one can prove analogues of Lemmas~\ref{lemma:bisim_omega} and \ref{lemma:equivalenceomega} to show that $\varphi(x)$ is suc$\MLiu$-definable iff $\models \varphi \leftrightarrow m(\varphi(x))$. Overall, this shows that $\CT/\MLiu$ definability is polytime reducible to $\CT$-validity.

For $\text{suc}\MLiu/\text{suc}\MLi$- and $\text{suc}\MLiu/\text{suc}\ML$-definability, we  
say that $\Amf_{0}$ is a \emph{generated substructure} of $\Amf_{1}$ if it is a substructure of $\Amf_{1}$ and whenever $a\in \text{dom}(\Amf_{0})$ and $(a,a')\in R^{\Amf_{1}}$ or
$(a',a)\in R^{\Amf_{1}}$, then $a'\in \dom(\Amf_{0})$. We call 
$\Amf_{0}$ a \emph{forward generated substructure} of $\Amf_{1}$ if it is a substructure of $\Amf_{1}$ and whenever $a\in \text{dom}(\Amf_{0})$ and $(a,a')\in R^{\Amf_{1}}$, then $a'\in \dom(\Amf_{0})$.
We say that $\varphi$ is \emph{invariant under $($forward$)$ generated substructures} if, for all pointed $\Amf,a$ and all (forward) generated substructures $\Amf'$ of $\Amf$ with $a\in \dom(\Amf')$, we have $\Amf\models\varphi(a)$ iff $\Amf'\models \varphi'(a)$.
We then show that 
\begin{enumerate}
	\item $\varphi\in \MLiu$ is $\MLi$-definable iff $\varphi$ is invariant under generated substructures\textup{;}
		
		
		\item $\varphi\in \MLiu$ is $\ML$-definable iff $\varphi$ is invariant under forward generated substructures.
	\end{enumerate}
By encoding the problem of deciding whether a formula in $\text{suc}\MLiu$
is invariant under (forward) generated substructures as a validity problem in
$\text{suc}\MLiu$ we obtain the following overall result.
\begin{theorem}\label{thm:succ}
Let $\LS\in \{\MLiu,\MLi,\ML\}$. Then $\CT/\LS$-defina\-bility is polytime reducible to $\CT$-validity.
\end{theorem}


\section{Uniform Separation}\label{sec:uniformseparation}

Let $\L,\LS$ be a pair of logics, $\sigma$ a signature,
and $\varphi$ an $\L(\sigma)$-formula. An $\LS(\sigma)$-formula $\varphi'$ is called an \emph{$\LS(\sigma)$-uniform separator for $\varphi$}
if it is a logically strongest $\LS(\sigma)$-formula that follows from $\varphi$ in the sense that $\varphi\models \varphi'\models \chi$, for every $\LS(\sigma)$-formula $\chi$ with $\varphi\models \chi$. The \emph{uniform $\L/\LS(\sigma)$-separation problem} is to decide, given an $\L(\sigma)$-formula $\varphi$, whether it has an $\LS(\sigma)$-uniform separator. If such a separator for $\varphi$ exists, it is unique up to logical equivalence and separates any $\LS(\sigma)$-separable $\L(\sigma)$-formulas $\varphi$ and $\psi$. Note, however, that uniform $\LS(\sigma)$-separators for distinct $\sigma$ can be non-equivalent and that a uniform $\L(\sigma)$-separator might only exist for some $\sigma$.

We show that, for $\L,\LS$ from Table~\ref{table:results}, the flattenings $\Flat$ and $\Flat_{\sigma_{c}}$ defined in the previous section provide uniform separators iff such separators exist. The following example is easily verified using the results to be established below.
\begin{example}\em
	Let $\varphi=\Diamond_{R}^{\geq k}A$. The formula $\Flat(\varphi)= \Diamond_{R} A$ is a uniform $\MLinu(\sigma)$-separator for $\varphi$ if $\sigma$ contains infinitely-many unary and binary predicates (including $R$ and $A$) and no constants. If $\sigma_{\ell}=\sigma\cup \sigma_{c}^{\ell}$, where $\sigma_{c}^{\ell}$ contains $\ell$-many constants, then $\varphi_{\ell}=\Flat_{\sigma_{c}^{\ell}}(\varphi)$ is a uniform $\MLinu(\sigma_{\ell})$-separator for $\varphi$. Note that  $\varphi_{\ell'}\models \varphi_{\ell}$, for $\sigma_{c}^{\ell'}\supseteq \sigma_{c}^{\ell}$, but the converse entailment does not hold. There is no uniform $\MLinu(\sigma_{\omega})$-separator for $\varphi$ if $\sigma_{\omega}$ contains infinitely-many constants.
	\hfill $\dashv$
\end{example}

Following what is customary for uniform interpolation and most robust~\cite{Visser96,DBLP:series/lncs/KonevLWW09,DBLP:journals/apal/KowalskiM19},
we assume that $\sigma$ contains countably-infinite sets of unary and binary predicates. If $\L$ admits nominals, we consider $\sigma$ with either a finite or an infinite set of constants, denoted $\sigma_{c}$. Unless otherwise indicated, we assume $\sigma$ to be fixed. 
We first consider nominal-free logics.

\begin{restatable}{theorem}{thmuniformsepone}
	\label{thm:uniformsep1}
	Let $\L/\LS$ be either $\CT/\FOTNE$ or any pair of nominal-free modal logics from Table~\ref{table:results}. Then, for any signature $\sigma$ with $\sigma_{c}=\emptyset$ and any $\L(\sigma)$-formula $\varphi$, the following conditions are equivalent\textup{:}
	\begin{enumerate}
		\item[$(a)$] $\varphi$ has a uniform $\LS(\sigma)$-separator\textup{;}
		
		\item[$(b)$] $\Flat(\varphi)$ is a uniform $\LS(\sigma)$-separator for $\varphi$\textup{;}
		
		\item[$(c)$] $\varphi\models \Flat(\varphi)$\textup{;}
		
		\item[$(d)$] $\varphi$ is preserved under $\omega$-expansions of $\sigma$-structures $\Amf$ in the sense that $\Amf\models \varphi(a)$ implies $\Amf^{\omega}\models \varphi(a,0)$, for any pointed $\sigma$-structure $\Amf,a$.
	\end{enumerate}
\end{restatable}
\begin{proof}
	We prove $(a) \Rightarrow (d)$ for $\GMLu/\MLu$. The implications $(d)\Rightarrow (c) \Rightarrow (b) \Rightarrow (a)$ can be proved using Lemmas~\ref{lemma:bisim_omega} and \ref{lemma:equivalenceomega}; the remaining pairs $\L/\LS$ are considered in the appendix.
	For a proof by contradiction, suppose $\varphi'$ is a uniform $\MLu(\sigma)$-separator for $\varphi$ but there is a pointed $\sigma$-structure $\Amf,a$ such that $\Amf\models \varphi(a)$ and $\Amf^{\omega}\not\models \varphi(a,0)$. By Lemma~\ref{lemma:equivalenceomega}, $\Amf^{\kappa}\not\models\varphi(a,0)$ for all $\kappa > \mx$. By Lemma~\ref{lemma:bisim_omega}, $\Amf^{\kappa}\models\varphi'(a,0)$ for all $\kappa> \mx$.
	
	For any signature $\varrho$ and any $k\geq 0$, let $\MLu_{k}(\varrho)$ be the set of all $\MLu(\varrho)$-formulas of modal depth at most $k$. If $\varrho$ is finite, then $\MLu_{k}(\varrho)$ is finite modulo logical equivalence. Take propositional variables $A_{0},\ldots,A_{\mx}$ in $\sigma\setminus(\sig(\varphi)\cup \sig(\varphi'))$ and let $\varrho=\{A_{0},\ldots,A_{\mx}\} \cup \sig(\varphi)\cup \sig(\varphi')$. Define $\Bmf^{\mx+1}$ as $\Amf^{\mx+1}$ with the exception that, for $0\leq i \leq \mx$,
\[
	A_{i}^{\Bmf^{\mx+1}}= \{(a,i) \mid a\in \dom(\Amf)\}.
\]
Consider the set $\Gamma$ of $\MLu_{\md(\varphi)+2}(\varrho)$-formulas $\chi$ such that $\Bmf^{\mx+1}\models\chi(a,0)$. Let $\gamma = \bigwedge_{\chi\in \Gamma}\chi$. 
	We show that $\gamma \models \varphi \leftrightarrow \Flat(\varphi)$. 
	Observe that $\Gamma$ contains 
	$\blacksquare(A_{i}\rightarrow \neg A_{j})$, for $0\leq i < j \leq \mx$, 
	and, for every $\chi\in \sub(\Flat(\varphi))$ and $R\in \varrho$, it also contains
	$$
	\blacksquare \big[\Diamond_{R} \chi \rightarrow \big(\Diamond_{R}(\chi \wedge A_{0}) \wedge \dots \wedge \Diamond_{R}(\chi \wedge A_{\mx}) \big) \big].
	$$
	It follows that $\gamma\models \blacksquare(\Diamond_{R} \chi \rightarrow \Diamond_{R}^{\geq \mx+1}\chi)$, for all $\chi\in \sub(\Flat(\varphi))$ and $R\in \varrho$, 
	and one can show by induction that $\gamma \models \varphi \leftrightarrow \Flat(\varphi)$. 
	
	It follows that $\varphi\models \gamma \rightarrow \Flat(\varphi)$, and so, by the definition of uniform $\MLu(\sigma)$-separators, we have $\varphi'\models \gamma \rightarrow \Flat(\varphi)$. In view of $\Bmf^{\mx+1}\models \varphi'(a,0)$ and $\Bmf^{\mx+1}\models \gamma(a,0)$, we obtain $\Bmf^{\mx+1}\models \Flat(\varphi)(a,0)$. Therefore, $\Amf^{\mx+1}\models \Flat(\varphi)(a,0)$ as $\Flat(\varphi)$ does not contain $A_{0},\ldots,A_{\mx}$. It follows from Lemma~\ref{lemma:equivalenceomega} that $\Amf^{\mx+1}\models\varphi(a,0)$, which is impossible.
\end{proof}

By Theorem~\ref{thm:uniformsep1}~(b), uniform separator existence is polytime reducible to $\L$-validity. 
We next discuss the case of modal logics with nominals. We first note a straightforward extension of Theorem~\ref{thm:uniformsep1} to signatures $\sigma$ with \emph{finite} $\sigma_{c} \ne \emptyset$:

\begin{theorem}\label{thm:uniformsep2}
Let $\L$ and $\LS$ be any pair of modal logics with nominals from Table~\ref{table:results}. Let $\sigma$ be a signature with finite $\sigma_{c} \ne \emptyset$. Then conditions $(a)$--$(d)$ of Theorem~\ref{thm:uniformsep1} with $\Flat$ replaced by $\Flat_{\sigma_{c}}$ are equivalent to each other.
\end{theorem}

We illustrate Theorem~\ref{thm:uniformsep2} by an example.

\begin{example}\em 
	Suppose $\sigma_{c}=\emptyset$ and consider the $\GMLn$-formula $\varphi=\Diamond_{R}A \wedge \Diamond_{S}A \wedge (\Diamond_{R}^{\geq 2}A \leftrightarrow \Diamond_{S}^{\geq 2}A)$. Then $\Flat(\varphi)\equiv \Diamond_{R}A \wedge \Diamond_{S}A$, $\varphi\models \Flat(\varphi)$, and so $\Flat(\varphi)$ is a uniform $\MLn(\sigma)$-separator for $\varphi$.
	However, if $\sigma_{c} = \{c_1,c_2\}$, then $\varphi \models \neg\Flat_{\sigma_{c}}(\varphi)$ as shown below,\\
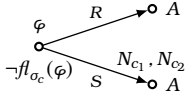
\centerline{
\begin{tikzpicture}[>=latex,line width=0.5pt,xscale = 1,yscale = 1]
\node[point,scale = 0.5,label=above:{\footnotesize $\varphi$},label=below:{\footnotesize $\neg\Flat_{\sigma_{c}}(\varphi)$}] (0) at (0,0) {};
\node[point,scale = 0.5,label=right:{\footnotesize $A$}] (2) at (1.5,0.5) {};
\node[point,scale = 0.5,label=right:{\footnotesize $A$},label=above:{\footnotesize $N_{c_1},N_{c_2}$}] (1) at (1.5,-0.5) {};
\draw[->,below] (0) to node[] {\scriptsize $S$} (1);
\draw[->,above] (0) to node[] {\scriptsize $R$} (2);
\end{tikzpicture}
}
\\
	%
	and so there is no uniform $\MLn(\sigma)$-separator for $\varphi$. \hfill $\dashv$
\end{example}

Finally, we consider the case when the set $\sigma_{c}$ of constants in $\sigma$ is infinite. We write $\varphi\models_{\it fin}\psi$ if 
$\Amf\models \varphi(a)$ implies $\Amf\models \psi(a)$, for all \emph{finite} pointed $\sigma$-structures $\Amf,a$. Note that $\varphi\models_{\it fin}\psi$ iff $\varphi\models\psi$, for all counting-free logics considered in this paper and also for $\GML$,
$\GMLn$, $\GMLu$, and $\GMLnu$ because these logics enjoy the \emph{finite model property} (FMP). All other logics we deal with here do not have the FMP~\cite{handbookDL,Pratt23book}.

\begin{restatable}{theorem}{thmuniformsepthree}\label{thm:uniformsep3}
Let $\L$ and $\LS$ be any pair of modal logics with nominals from Table~\ref{table:results}. Let $\sigma$ be a signature with $\sigma_{c}$ infinite and $\sigma_{0}\subseteq \sigma_c$ finite, non-empty, and containing all constants in an $\L(\sigma)$-formula $\varphi$. Then the following conditions are equivalent\textup{:}
	\begin{enumerate}
		\item[$(a)$] $\varphi$ has a  uniform $\LS(\sigma)$-separator\textup{;}
		
		\item[$(b)$] $\Flat_{\sigma_{0}}(\varphi)$ is a uniform $\LS(\sigma)$-separator for $\varphi$\textup{;}
		
		\item[$(c)$] $\varphi\models\Flat_{\sigma_{0}}(\varphi)$ and $\Flat_{\sigma_{0}}(\varphi)\models_{\it fin} \varphi$.		
	\end{enumerate}
\end{restatable}
\begin{proof}
	We sketch the proof of (b) $\Rightarrow$ (c) for $\GMLinu/\MLinu$; $(c) \Rightarrow (b)$ is straightforward and $(a) \Leftrightarrow (b)$ is shown using Theorem~\ref{thm:uniformsep2}. Observe that, for a finite pointed $\sigma_{f}$-structure $\Amf,a$ with finite $\sigma_{f}$, one can construct an $\MLinu(\sigma)$-formula $\text{diag}_{\sigma_{f}}(\Amf,a)$ with fresh constants for elements of $\dom(\Amf)$ such that $(i)$ if $\Bmf\models \text{diag}_{\sigma_{f}}(\Amf,a)(b)$, then $\Amf,a \cong_{\sigma_{f}}\Bmf,b$ and $(ii)$ by interpreting the fresh constants in $\text{diag}_{\sigma_{f}}(\Amf,a)$ as elements of $\dom(\Amf)$,
	$\Amf'\models \text{diag}_{\sigma_{f}}(\Amf,a)(a)$ for the resulting structure $\Amf'$.
Let $\Flat_{\sigma_{0}}(\varphi)$ be a uniform $\LS(\sigma)$-separator for $\varphi$, so 
	$\varphi\models\Flat_{\sigma_{0}}(\varphi)$. To show 
	$\Flat_{\sigma_{0}}(\varphi)\models_{\it fin}\varphi$, assume this is not the case and  there is a finite pointed $\sig(\varphi)\cup\sigma_{0}$-structure $\Amf,a$ with $\Amf\models \Flat_{\sigma_{0}}(\varphi)(a)$ but
	$\Amf\not\models \varphi(a)$. 
	Then $\varphi\models \neg \text{diag}_{\sigma_{f}}(\Amf,a)$ for $\sigma_{f}= \sig(\varphi)\cup\sigma_{0}$.
	By the definition of uniform separators, $\Flat_{\sigma_{0}}(\varphi)\models\neg \text{diag}_{\sigma_{f}}(\Amf,a)$, and so $\Amf\not\models \Flat_{\sigma_{0}}(\varphi)(a)$, which is a contradiction.  
\end{proof}

It follows that, for the pairs $\L/\LS$ considered in Theorem~\ref{thm:uniformsep3}, uniform $\LS(\sigma)$-separator existence for $\sigma$ containing infinitely many constants reduces in polytime to finite and general $\L$-validity. For logics $\L$ with the FMP, it follows from Theorem~\ref{thm:uniformsep3} that uniform $\LS$-separator existence and $\LS$-definability coincide. For $\L$ without the FMP, this is not the case.

\begin{example}\em 
	%
	Consider $\GMLin/\MLin$. Let $\varphi$ be a satisfiable $\GMLin$-formula that is not finitely satisfiable. We claim that $\top$ is a uniform $\MLin(\sigma)$-separator for $\neg \varphi$. Indeed, if $\neg\varphi \models \chi$, for some $\MLin$-formula $\chi$, then $\neg\chi \models \varphi$, and so $\neg\chi$ is not finitely satisfiable. By the FMP of $\MLin$, this can only happen if $\neg\chi$ is not satisfiable, that is, $\chi \equiv \top$. On the other hand, $\neg \varphi$ is not $\MLin(\sigma)$-definable since $\MLin$ enjoys the FMP and so $\varphi$ would be finitely satisfiable iff it is satisfiable. \hfill $\dashv$ \end{example}

We close with a brief discussion of uniform $\CT/\MLiu$-separation, which reduces to uniform $\FOTNE/\MLiu$-separation since we have already analysed uniform $\CT/\FOTNE$-separation. Surprisingly, we can again use the same function $m(\varphi)$ that characterises $\FOTNE/\MLiu$-definability. In this case, we show in the appendix that, for any $\FOTNE(\sigma)$-formula $\varphi$, the following conditions are equivalent\textup{:}
	\begin{enumerate}
		\item[$(a)$] $\varphi$ has a uniform suc$\MLiu(\sigma)$-separator\textup{;}
		
		\item[$(b)$] $m(\varphi)$ is a uniform suc$\MLiu(\sigma)$-separator for $\varphi$\textup{;}
		
		\item[$(c)$] $\varphi\models m(\varphi)$.
    \end{enumerate}
It follows that uniform $\FOTNE/\MLiu(\sigma)$-separation can be reduced in polytime to $\FOTNE$-validity. In contrast, for $\MLiu/\MLi$ and $\MLiu/\ML$, the definability results do not directly transfer to uniform separation; their analysis is left for future work.  


\section{Future Work}\label{future}

We have presented a comprehensive analysis of the complexity of separation and definability in fragments of $\CT$. While a rather elementary (but elegant) approach worked in many cases for definability and uniform separation, more labour was needed to understand separation. Nevertheless, the decidability of $\CT/\FOT$-definability and uniform separation remain open and will require a different technique compared to the decidability proof for $\CT/\FOTNE$. 

As far as separation is concerned, a main open problem is to find methods to construct separators and bound their size. Based on the mosaic technique introduced in this article, first results have been obtained~\cite{jung2025computation}. Using a mosaic elimination algorithm, it is shown that for $\GML/\ML$- and $\GMLu/\MLu$-separation, one can always construct separators of elementary size if separators exist at all. 

It would also be of interest to investigate the problems considered here over finite models. If $\L$ has the FMP, then there is no difference, so only $\CT$ and extensions of $\GMLiu$ are affected. We conjecture that it is straightforward to adapt our undecidability proofs by working with appropriate `finitary' $n$-bisimulations (as compactness cannot be used, Theorem~\ref{criterion} does not hold and has to be replaced by a finitary approximation). The definability proofs are easily adapted, and so are the proofs for uniform separation. Observe that now definability coincides with uniform separation for all pairs of logics $\L/\LS$ with nominals in Table~\ref{table:results} if signatures containing infinitely-many constants are considered. Significant modifications are needed, however, to deal with $\GMLiu/\MLiu$-separation, where the current proof relies on realising star-types infinitely often.  


\cleardoublepage
\appendix

\section*{Proofs for Section~\ref{sec:definitions}}
%

\lred*

\begin{proof}
If $L \in \{\FOT, \CT\}$ and $\varphi(x), \psi(x) \in \L$, we define $\L$-formulas $\varphi'(x)$ and $\psi'(x)$ by taking
\begin{align*}
& \varphi'(x) ~=~ \varphi(x) \land \hspace*{-2mm} \bigwedge_{A \in \sig(\psi) \setminus \sig(\varphi)} \hspace*{-6mm} \big( A(x) \to A(x) \big) \land{} \\ 
&\hspace*{2.2cm} \bigwedge_{R \in \sig(\psi) \setminus \sig(\varphi)} \hspace*{-6mm} \big( R(x,x) \to R(x,x) \big) \land \hspace*{-2mm} \bigwedge_{c \in \sig(\psi) \setminus \sig(\varphi)} \hspace*{-6mm} (c=c),\\
& \psi'(x) ~=~ \psi(x) \land \hspace*{-2mm} \bigwedge_{A \in \sig(\varphi) \setminus \sig(\psi)} \hspace*{-6mm} \big( A(x) \to A(x) \big) \land{} \\ 
& \hspace*{2.2cm}\bigwedge_{R \in \sig(\varphi) \setminus \sig(\psi)} \hspace*{-6mm} \big( R(x,x) \to R(x,x) \big) \land \hspace*{-2mm} \bigwedge_{c \in \sig(\varphi) \setminus \sig(\psi)} \hspace*{-6mm} (c=c).
\end{align*}
Then $\varphi(x)$ and $\psi(x)$ are $\L/\LS$-separable iff $\varphi'(x)$ and $\psi'(x)$ are Craig $\L/\LS$-separable, and so unary Craig separable. 
If $\L$ is a modal logic, we take $\Diamond_R \top \to \Diamond_R \top$ and $N_c \to N_c$ in the last two conjuncts of $\varphi'(x)$ and $\psi'(x)$.

For the second claim, let $\Amf$ be a $\sigma$-structure. For a non-empty set $D\subseteq \dom(\Amf)$, denote by $\Amf_{|D}$ the \emph{restriction} of $\Amf$ to $D$ defined by taking $\dom(\Amf_{|D})=D$ and by relativising all predicates in $\sigma$ to $D$, that is,  by setting $(a_{1},\dots,a_{n})\in R^{\Amf_{|D}}$ iff $a_{1},\dots,a_{n}\in D$ and $(a_{1},\dots,a_{n})\in R^{\Amf}$, for all $a_{1},\dots,a_{n}\in \dom(\Amf)$ and all $R\in \sigma$. If $\sigma$ contains constants $c$, we demand that $c^{\Amf} \in D$ and set $c^{\Amf_{|D}} = c^{\Amf}$.
We remind the reader that the \emph{relativisation} $\varphi_{|A}(x)$ of a formula $\varphi(x)$ to a unary predicate $A \notin \sig(\varphi)$  has the following property \cite{modeltheory}: for all structures $\Amf$ and all $a\in \dom(\Amf)$, 
\begin{itemize}
\item[--] $\Amf\models \varphi_{|A}(a)$ implies  $A^{\Amf}\ne \emptyset$;

\item[--] $\Amf\models \varphi_{|A}(a)$ iff $\Amf_{|A^{\Amf}}\models \varphi(a)$.
\end{itemize}
For a $\CT$-formula $\varphi(x)$, its $\CT$-relativisation $\varphi(x)_{|A}$ is constructed in polynomial time. 

We are now in a position to give the reduction
for $\CT/\FOT$. Assume $\varphi(x)$ and $\psi(x)$ in $\CT$ are given. Let $A_{1},\ldots, A_{n}$ be the unary predicate symbols that either occur in $\varphi$ and not in $\psi$ or vice versa. Let $\varrho=(\sig(\varphi) \cup \sig(\psi))\setminus \{A_{1},\dots,A_{n}\}$. Let
$U$ be a fresh unary predicate symbol and $R$ a fresh binary predicate symbol.
Let $\varphi_{|U}$ and $\psi_{|U}$ be the relativisations of $\varphi$ and $\psi$ to $U$.
Consider the formulas $\varphi'$ and $\psi'$ obtained by 
replacing in $\varphi_{|U}$ and $\psi_{|U}$ every occurrence of $A_{i}$ by a formula $\varphi_{A_{i}}$ stating that there is an $R$-chain of length $i$ with the final node having exactly two $R$-successors. In $\GML$ terms,
$$
\varphi_{A_{i}} = \underbrace{\Diamond_{R} \dots \Diamond_{R}}_{i} \Diamond_{R}^{=2}\top.
$$
Let $\tau=\sig(\varphi')\cup \sig(\psi')$. By Lemma~\ref{criterion}, it suffices to show that the following conditions are equivalent:
\begin{enumerate}
\item there exist $\Amf,a$ and $\Bmf,b$ such that \mbox{$\Amf\models \varphi(a)$}, $\Bmf\models \psi(b)$, and $\Amf,a\sim_{\FOT(\varrho)}\Bmf,b$;

\item there exist $\Amf',a'$ and $\Bmf',b'$ such that $\Amf'\models \varphi'(a')$, \mbox{$\Bmf'\models \psi'(b')$}, and $\Amf',a'\sim_{\FOT(\tau)}\Bmf',b'$. 
\end{enumerate}
For the implication $(1) \Rightarrow (2)$, suppose  
$\Amf\models \varphi(a)$, $\Bmf\models \psi(b)$ and $\Amf,a\sim_{\FOT(\varrho)}\Amf,b$.
Define $\Amf',a'$ and $\Bmf',b'$ as follows: the domain of
$\Amf'$ contains $\dom(\Amf)$ and, for every $d\in \dom(\Amf)$,  mutually disjoint sets $A_{d}$ we define later.
We set $a'=a$. Similarly, the domain of
$\Bmf'$ contains $\dom(\Bmf)$ and, for every $d\in \dom(\Bmf)$,  mutually disjoint sets $B_{d}$ to be defined later. Set $b'=b$.
The symbols in $\varrho$ are interpreted in $\Amf'$ and $\Bmf'$ exactly as in $\Amf$ and $\Bmf$, respectively. Set $U^{\Amf'}=\dom(\Amf)$ and $U^{\Bmf'}=\dom(\Bmf)$.
It remains to interpret $R$. We interpret $R$ in $\Amf'$ in such a way that, for $d\in \dom(\Amf)$, we have $\Amf\models A_{i}(d)$ iff $\Amf'\models \varphi_{A_{i}}(d)$, for $i=1,\dots,n$. Hence, we use the additional points in $A_{d}$ as the domain of an $\{R\}$-structure $\Amf_{d}$ consisting of $n$-many $R$-chains all starting at $d$, with the $i$th $R$-chain having length $i$ and such that its final node has \emph{two} $R$-successors if $d\in A_{i}^{\Amf}$ and \emph{three} $R$-successors if $d\not\in A_{i}^{\Amf}$. For $d\in \dom(\Amf)$, the $\Amf_{d}$ are hooked to $\Amf$ at $d$ to obtain $\Amf'$. Observe that we have  $\Amf_{d}\models \varphi_{A_{i}}(d)$ iff $\Amf\models A_{i}(d)$, for $i=1,\dots,n$, and that $\Amf_{d},d\sim_{\FOT(\{R\})}\Amf_{d'},d'$, for any $d,d'\in \dom(\Amf)$, since $\FOT$ cannot distinguish between nodes having two or three $R$-successors.
The interpretation of $R$ in $\Bmf'$ is analogous.

It follows that $\Amf'\models \varphi'(a')$ and $\Bmf'\models \psi'(b')$. Moreover, any $\bis$ witnessing $\Amf,a\sim_{\FOT(\varrho)}\Bmf,b$ can be extended to a $\bis'$ witnessing $\Amf',a'\sim_{\FOT(\tau)}\Bmf',b'$ 
in a straightforward manner.

For $(2) \Rightarrow (1)$, assume that 
$\Amf'\models \varphi'(a')$, $\Bmf'\models \psi'(b')$, and $\Amf',a'\sim_{\FOT(\tau)}\Bmf',b'$. 
Let $\Amf=\Amf'_{|U^{\Amf'}}$ be expanded by setting 
$$
A_{i}^{\Amf}=\{d\in U^{\Amf'}\mid \Amf'\models \varphi_{A_{i}}(d)\},
$$
for $i=1,\dots,n$, and let $\Bmf =\Bmf'_{|U^{\Bmf'}}$ be expanded by setting 
$$
A_{i}^{\Bmf}=\{d\in U^{\Bmf'}\mid \Bmf'\models \varphi_{A_{i}}(d)\},
$$
for $i=1,\dots,n$. Let $a=a'$ and $b=b'$. By construction, we have $\Amf\models \varphi(a)$ and $\Bmf\models \psi(b)$. Moreover, let $\bis$ be a bisimulation witnessing $\Amf',a'\sim_{\FOT(\tau)}\Bmf',b'$. Then the restriction of $\bis$ to $U^{\Amf'}\times U^{\Bmf'}$ witnesses $\Amf,a\sim_{\FOT(\varrho)}\Bmf,b$. 

The polynomial reduction given above works for all pairs $\CT/\LS$ with $\LS \in \{\MLiu,\MLi,\MLu,\ML\}$ since the formulas $\varphi_{A_{i}}$ are in $\CT$.

Next consider $\GMLinu/\MLinu$. As before, we can construct the relativisation $\varphi_{|A}$ of a $\GMLinu$-formula $\varphi$ in $\GMLinu$ to a propositional variable $A$. To ensure that all nominals are interpreted in the relativised structure, we add to $\varphi_{|A}$ the conjunct $\Du (N_{c} \wedge A)$ for $c\in \sig(\varphi) \cup \sig(\psi)$. With this amendments, the proof above works for $\GMLinu/\LS$ and all $\LS\in \{\MLiu,\MLin,\MLnu,\MLu,\MLi,\MLn,\ML\}$.

The proof for the relevant pairs $\GMLiu/\LS$, $\GMLnu/\LS$, $\GMLu/\LS$, $\GMLi/\LS$, and $\GML/\LS$ is a straightforward modification of the proof above.

For the relevant pairs of the form $\GMLin/\LS$ and $\GMLn/\LS$, some care is required since one cannot now ensure that nominals are interpreted in the relativised structure by adding $\Du(N_{c} \wedge U)$ as a conjunct to $\varphi_{|U}$.
Instead, when constructing $\Amf$ and $\Bmf$ from $\Amf'_{|U^{\Amf'}}$ and $\Bmf'_{|U^{\Bmf'}}$ in the proof of $(2) \Rightarrow (1)$, we add new elements to the domains of $\Amf$ and $\Bmf$ to interpret the nominals not yet interpreted in $\Amf'_{|U^{\Amf'}}$ and $\Bmf'_{|U^{\Bmf'}}$.
\end{proof}

The following criterion of relative definability is proved similarly to Lemma~\ref{criterion}:

\begin{lemma}\label{p:pbdpsem}
For any relevant pair $\L/\LS$, any $\L(\sigma)$-formulas $\varphi(x)$, $\psi(x)$ and any $\varrho \subseteq \sigma$, the following conditions are equivalent\textup{:}
\begin{itemize}
\item[--]
there is no explicit $\LS(\varrho)$-definition of $\psi$ modulo $\varphi$\textup{;}

\item[--] there are pointed $\sigma$-structures $\Amf,a$ and $\Bmf,b$ such that 
$$
\Amf\models (\varphi \land \psi)(a), \quad \Bmf\models (\varphi \land \neg\psi)(b), \quad \Amf,a \sim_{\LS(\varrho)} \Bmf,b.
$$
\end{itemize}
\end{lemma}

\thmequivdef*
\begin{proof}
We only consider modal logics $\L$ and $\LS$; the first-order case is similar. Let $\varrho \subseteq \sigma$. We show that an  $\L$-formula $\psi$ has an explicit $\LS(\varrho)$-definition modulo $\varphi$ in $\L$ iff the formulas $\varphi \land \psi$ and $\varphi^\varrho \land\neg \psi^\varrho$ have a Craig $\L/\LS$-separator, where $\varphi^\varrho$, $\psi^\varrho$ are obtained by replacing each variable $A \notin \varrho$ with a fresh variable $A^\varrho$ (cf.~\cite{MGabbay2005-MGAIAD}). Indeed, suppose first that $\chi$ is an $\LS(\varrho)$-formula with $\varphi \models \psi\leftrightarrow\chi$. Let $\chi'$ be the result of replacing every variable in $\chi$ that does not occur in $\sig(\varphi \land \psi)$ by $\top$, so $\sig(\chi') \subseteq \sig(\varphi \land \psi) \cup \sig(\varphi^\varrho \land \psi^\varrho)$. We clearly still have $\varphi \models \psi\leftrightarrow\chi'$, and so $\varphi \land \psi \models \chi'$ and $\chi' \models \varphi \to \psi$. The latter implies $\chi' \models \varphi^\varrho \to \psi^\varrho$ and $\chi' \models \neg(\varphi^\varrho \land \neg \psi^\varrho)$. Therefore, $\chi'$ is a Craig $\L/\LS$-separator for $\varphi \land \psi$ and $\varphi^\varrho \land\neg \psi^\varrho$. 
Conversely, it is readily seen that any Craig $\LS$-separator of $\varphi \land \psi$ and $\varphi^\varrho \wedge\neg \psi^\varrho$ is an explicit $\LS(\varrho)$-definition of $\psi$ modulo $\varphi$.
	
For the other reduction, we observe first that the decision problem for $\L$ is polynomial time reducible to relative $\L/\LS$-definability: for $\psi = A \notin \sig(\varphi)$ and $\varrho = \emptyset$, we have $\models \neg \varphi$ iff there is an explicit $\LS(\varrho)$-definition of $\psi$ modulo $\varphi$.
We show now that $L$-formulas $\varphi$ and $\neg\psi$ have a Craig $\L/\LS$-separator iff $\models \varphi\to \psi$ and there is an $\LS(\sig(\varphi) \cap \sig(\psi))$-definition of $\psi$ modulo $\psi \to \varphi$. 

Indeed, if $\chi$ is a Craig $\L/\LS$-separator for $\varphi$ and $\neg\psi$, then we have $\sig(\chi) \subseteq \sig(\varphi) \cap \sig(\psi)$, $\varphi \models \chi$ and $
\chi \models \psi$. It follows that $\models \varphi \to \psi$ and $\psi \to \varphi \models \psi \leftrightarrow \chi$, which means that $\chi$ is an $\LS(\sig(\varphi) \cap \sig(\psi))$-definition of $\psi$ modulo $\psi \to \varphi$. Conversely, suppose $\models \varphi \to \psi$ but $\varphi$ and $\neg \psi$ do not have a Craig $\L/\LS$-separator. By Lemma~\ref{criterion}, there are structures $\Amf,a$ and $\Bmf,b$ such that $\Amf\models\varphi(a)$, $\Bmf\models\neg\psi(b)$ and $\Amf,a \sim_{\LS(\varrho)} \Bmf,b$. But then $\Amf\models ((\psi\to\varphi)\land\psi)(a)$ and $\Bmf \models ((\psi\to\varphi)\land\neg\psi)(b)$, and so, by Lemma~\ref{p:pbdpsem}, there is no explicit $\LS(\varrho)$-definition of $\psi$ modulo $\psi \to \varphi$.
\end{proof}


\section*{Proofs for Section~\ref{sec:undec}}

\undecidability*

\subsection{$\GMLinu/\MLinu$-separation}\label{sec:undec-modal}

We reduce the non-halting problem for 2RMs to the non-existence of a Craig $\GMLinu/\MLinu$-separator. 
Recall that a 2RM is a pair
$M=(Q,P)$, where $Q = q_0,\dots,q_{m}$ is a set of \emph{states}
and $P = I_0,\dots,I_{m-1}$ a sequence of \emph{instructions}.
By definition, $q_0$ is the \emph{initial state}, and $q_m$ the
\emph{halting state}.  For all $i < m$,
\begin{itemize}
\item[--] either $I_i=+(\ell,q_j)$ is an \emph{incrementation instruction} with a register $\ell \in \{0,1\}$ and the subsequent state $q_j$;

\item[--] or $I_i=-(\ell,q_j,q_k)$ is a \emph{decrementation instruction} with a register $\ell \in \{0,1\}$, the subsequent state  $q_j$ if register~$\ell$ contains~0, and $q_k$ otherwise.
\end{itemize}
A \emph{configuration} of $M$ is a triple $(q,n_0,n_1)$, where $q$ is the
current state and $n_0,n_1 \in \omega$ the register contents.  We
write $(q_i,n_0,n_1) \Rightarrow_M (q_j,n_0',n_1')$ if one of the
following holds:
\begin{itemize}
\item[--] $I_i = +(\ell,q_j)$, $n_\ell' = n_\ell +1$, and $n_{1-\ell}'= n_{1-\ell}$;

\item[--] $I_i = -(\ell,q_j,q_k)$, $n_\ell'=n_\ell=0$, and $n_{1-\ell}'=n_{1-\ell}$;

\item[--] $I_i = -(\ell,q_k,q_j)$, $n_\ell > 0$, $n'_\ell = n_\ell -1$, and $n_{1-\ell}' = n_{1-\ell}$.
\end{itemize}
The \emph{computation} of $M$ on input $(n_0,n_1) \in \omega^2$ is
the unique longest configuration sequence 
$$
(p^0,k_0^0,k_1^0)
\Rightarrow_M (p^1,k_0^1,k_1^1) \Rightarrow_M \cdots
$$ 
such that $p^0 =
q_0$, $k_0^0 = n_0$, and $k_1^0 = n_1$.
%
The \emph{halting problem} for 2RMs is to decide, given a 2RM $M$, whether
its computation on input $(0,0)$ is finite (which implies that its
last state is $q_m$).

We assume without loss of generality (by adding a `dummy' initial state if necessary) that the initial state $q_0$ only occurs initially and that $I_0$ is a decrementation instruction.

Since we are reducing the non-halting problem for 2RMs to the non-existence of a Craig $\GMLinu/\MLinu$-separator, any machine $M$ will be associated with two $\GMLinu$-formulas $\varphi_M$ and $\psi_M$. Models of $\varphi_M$ represent the progression of the states of $M$ over time, while models of $\psi_M$ represent the content of the registers over time. The two are bisimilar (and hence there is no separator) if the state progression and the register contents match, and hence the models of $\varphi_M$ and $\psi_M$ together represent the execution of $M$.

The shared signature, $\varrho$, of $\varphi_M$ and $\psi_M$ consists of propositional variables $q_i\in Q$, $e_0$ and $e_1$, along with constant $c$ and binary relation $R$. Furthermore, $\psi_M$ uses non-shared propositional variables $r_0$ and $u$. For ease of notation, we define $r_1 \equiv \neg r_0$.


The intended meaning of $q_i$ is `the machine is currently in state $q_i$'\!, while $e_0$ and $e_1$ are used to represent `register 0 is empty' and `register 1 is empty'\!, respectively. The relation $R$ is used to represent the passage of time (i.e., if $(x,y)\in R$, then $y$ is one time step later than $x$). The constant $c$ identifies the starting point of the computation. The variable $r_0$ is used to indicate that a particular element of the domain counts as a part of register 0, while $r_1$ indicates elements that count as part of register 1. Finally, $u$ stands for `unique'\!, and will be true for exactly one element of each register at each point in time.


The two formulas $\varphi_M$ and $\psi_M$ are both conjunctions, with the conjuncts best presented in a few groups. For $\varphi_M$, the first group of conjuncts is as follows:
\begin{enumerate}
\item \label{form:c}$N_c \wedge q_0$,

\item \label{form:one_of} $\blacksquare \bigvee_{1\leq i \leq m}q_i$,

	\item \label{form:Q_excl} $\blacksquare\bigwedge_{1\leq i \leq m}(q_i\rightarrow \bigwedge_{i<j\leq m}\neg q_j)$.
\end{enumerate}	
Subformulas \eqref{form:c}--\eqref{form:Q_excl} express that every element of the domain satisfies exactly one of $q_0, \dots, q_{m}$, and the constant $c$ is the initial point of evaluation and satisfies $q_0$.


The next group of conjuncts of $\varphi_M$ is as follows:
\begin{enumerate}[resume]
	\item \label{form:cont}$\blacksquare(\Diamond\!^{= 1}_{R}\top)$,
	
	\item \label{form:start} $\square_{R^-}\bot$,
	
	\item \label{form:back}$\blacksquare(\neg N_c\rightarrow \Diamond\!^{= 1}_{R^{-1}}\top)$.
\end{enumerate}
These conjuncts describe the relation $R$: every element has exactly one $R$-successor, the initial element has no $R$-predecessors, and every other element has exactly one $R$-predecessor. 

It follows from these conditions that if $\Amf\models \varphi_M(x_0)$, then $\Amf$ contains an infinite $R$-chain starting at $x_{0}$ and each element of the chain satisfies exactly one $q_i$. We will interpret these elements as representing a run of $M$.

The next group of conjuncts depends on the instructions of the machine $M$. If $I_i=+(\ell,q_j)$, we add the conjunct
\begin{enumerate}[resume]
	\item\label{form:incr} $\blacksquare(q_i\rightarrow \square_R q_j)$,
\end{enumerate}
and if $I_i = -(\ell,q_j,q_k)$ we add the conjuncts
\begin{enumerate}[resume]
	\item\label{form:decr1} $\blacksquare((q_i\wedge e_\ell)\rightarrow \square_Rq_j)$,
	
	\item\label{form:decr2} $\blacksquare((q_i\wedge \neg e_\ell)\rightarrow \square_R q_k)$.
\end{enumerate}
These conjuncts express that the instructions of $M$ are followed, where $e_\ell$ (with $\ell\in \{0,1\}$) indicates that register $\ell$ is empty. Note that we do not, at this stage, keep track of the value in each register. But \emph{if} $e_\ell$ holds iff register $\ell$ is empty, then these conjuncts ensure that every pointed model $\Amf,x_{0}$ of $\varphi_{M}$ encodes the run of $M$.

Finally, we add one more, very simple, conjunct
\begin{enumerate}[resume]
	\item\label{form:nonhalt} $\blacksquare\neg q_\ell$,
\end{enumerate}
which expresses that the halting state $q_\ell$ is not reached. So, if a model of $\varphi_M$ encodes a run of $M$, it must be a non-halting run.

The formula $\psi_M$ can similarly be divided into groups of conjuncts. 
The first group of conjuncts of $\psi_M$ is the same as the the first group of $\varphi_M$:
\begin{enumerate}
	\item \label{form2:c}$N_c\wedge q_0$,
	\item \label{form2:one_of} $\blacksquare\bigvee_{1\leq i \leq m}q_i$,
	\item \label{form2:Q_excl}$\blacksquare\bigwedge_{1\leq i \leq m}(q_i\rightarrow \bigwedge_{i<j\leq m}\neg q_j)$.
\end{enumerate}	
So, as with $\varphi_M$, in every model of $\psi_M$, all elements satisfy exactly one of $q_0, \dots, q_m$. The remaining conjuncts are quite different, however:
\begin{enumerate}[resume]
	\item \label{form2:init1}$\Diamond\!_R^{=2}\top$,
	
	\item \label{form2:init2}$\Diamond\!_R(u \wedge r_0) \wedge \Diamond\!_R(u\wedge r_1)$,
	
	\item \label{form2:init3}$e_0\wedge e_1$,
	
	\item \label{form2:prop1}$\blacksquare(r_\ell \rightarrow (\square_Rr_\ell \wedge \square_{R^{-1}}(e_\ell\vee N_c)))$, for $\ell\in \{0,1\}$,
	
	\item \label{form2:prop2}$\blacksquare((u\wedge \neg N_c)\rightarrow (\Diamond\!_R^{=1}u \wedge\Diamond\!_{R^{-1}}^{=1}u))$.
\end{enumerate}
Note that conjuncts~\eqref{form2:init1}--\eqref{form2:init3} do not start with $\blacksquare$, so they apply only to the initial point of evaluation which, due to conjunct~\eqref{form2:c}, is the constant $c$. The conjuncts then state that this element has exactly two $R$-successors, with one satisfying $u\wedge r_0$ and the other satisfying $u\wedge r_1$. Since we use $r_0$ and $r_1$ to denote the two registers, this means that a model of $\psi_M$ splits into two branches representing the two registers. Furthermore, conjunct~\eqref{form2:init3} says that $e_0$ and $e_1$ both hold at $c$, i.e., both registers are initially empty.

Conjuncts~\eqref{form2:prop1} and \eqref{form2:prop2} handle propagation: \eqref{form2:prop1} says that $r_\ell$ propagates forward and backward along $R$ (except for $c$), while \eqref{form2:prop2} says that $u$ propagates to exactly one $R$-successor and $R$-predecessor (except for $c$). So the branches representing the registers do not interact, and there is a unique $R$-path in each register of elements that satisfy $u$.

The next group of conjuncts depends on the instructions of $M$. For $I_i = +(\ell,q_j)$, we have conjuncts
\begin{enumerate}[resume]
	\item \label{form2:incr} $\blacksquare((q_i\wedge r_\ell)\rightarrow (\Diamond\!_R^{=2}\top \wedge \square_R\Diamond\!_{R^{-1}}^{=1}\top))$, 
	
	\item \label{form2:incr2} $\blacksquare((q_i\wedge \neg r_\ell)\rightarrow (\Diamond\!_R^{=1}\top \wedge \square_R\Diamond\!_{R^{-1}}^{=1}\top))$.
\end{enumerate}
For $i\not = 0$ and $I_i = -(\ell,q_j,q_k)$, we have
\begin{enumerate}[resume]
	\item \label{form2:decr1} $\blacksquare ((q_i \wedge r_\ell \wedge e_\ell)\rightarrow (\Diamond\!_R^{=1}\top \wedge \square_R\Diamond\!_{R^{-1}}^{=1}.\top))$,
	
	\item \label{form2:decr2} $\blacksquare ((q_i \wedge r_\ell \wedge \neg e_\ell)\rightarrow (\Diamond\!_R^{=1}\top \wedge \square_R\Diamond\!_{R^{-1}}^{=2}.\top))$,
	
	\item \label{form2:decr3} $\blacksquare ((q_i \wedge \neg r_\ell)\rightarrow (\Diamond\!_R^{=1}\top \wedge \square_R\Diamond\!_{R^{-1}}^{=1}.\top))$.
\end{enumerate}
Conjunct~\eqref{form2:incr} says that if $q_i$ is an incrementation instruction, then an element satisfying $q_i\wedge r_\ell$ will have two $R$-successors, each of which having one $R$-predecessor. As we use $R$ to represent time in the computation, this means that the number of $r_\ell$-elements will double in a state with an incrementation instruction. With respect to conjuncts~\eqref{form2:decr1} and \eqref{form2:decr2}, recall that we use $e_\ell$ as a marker for register $\ell$ being empty. Thus, these conjuncts say that if a decrementation instruction is given with an empty register (conjunct~\eqref{form2:decr1}), then for every $r_\ell$-element, there is one $R$-successor that has one $R$-predecessor, so the number of $r_\ell$-elements remains the same, and if a decrementation instruction is given with a non-empty register (conjunct~\eqref{form2:decr2}), then for every $r_\ell$-element, there is a single $R$-successor that has two $R$-predecessors, so the number of $r_\ell$-elements is halved.

Since we double for incrementation and halve for decrementation, a register value of $n$ in register $\ell$ will be represented by $2^n$ elements satisfying $r_\ell$.

Conjuncts~\eqref{form2:incr2} and \eqref{form2:decr3} state that elements in the other register (i.e., those satisfying $\neg r_\ell$) have exactly one successor that has exactly one predecessor, so the number of such elements remains the same.

Note that we treat $q_0$ separately. This is because we first need to create the two registers before we can increase or decrease the number of elements in each register. This is done by conjuncts~\eqref{form2:init1} and \eqref{form2:init2}. After one step of the computation, our two registers are initialised with one element each, which we interpret as the value $0$ for that register. This is why we assume that $q_0$ is decrementing (so the registers are empty after one step of the computation) and only occurs initially (so we do not perform the initialisation again at later times).

Finally, we add one more conjunct
\begin{enumerate}[resume]
	\item \label{form2:empty} $\blacksquare ((e_\ell\sqcap r_\ell)\rightarrow u)$,
\end{enumerate}
which states that the combination of $e_\ell$ and $r_\ell$ is possible only when $u$ also holds.

We show that an $\MLinu$-separator for $\varphi_M$ and $\psi_M$ exists iff $M$ is non-halting, from which it follows that separation is undecidable. More precisely, the reduction shows that non-separation is co-RE hard, and therefore the separation problem is RE hard.  

Before proving our reduction, let us first introduce the \emph{intended pointed models} for $\varphi_M$ and $\psi_M$, where $M$ is non-halting. A pointed model $\Amf,x_{0}$ of $\varphi_M$ will be used to represent a clock, and the state the machine is in at each time step. Recall that $\varphi_M$ guarantees that $\Amf$ follows the instructions of $M$ as long as $e_\ell$ holds when register $\ell$ is empty. In our intended models, we will make sure that (i) $e_\ell$ only holds if register $\ell$ is empty and (ii) if register $\ell$ is empty and the current instruction is to decrement register $\ell$, then $e_\ell$ holds (to indicate that decrementing is impossible). If register $\ell$ is empty and the current instruction is anything other than to decrement register $\ell$, we do not care whether $e_\ell$ holds.

A pointed model $\Bmf,y_{0}$ of $\psi_M$ is used to represent the content of the registers. Every point $t$ of the computation is represented by some element $x_t$ of $\Amf$, and this $x_t$ is $\MLinu$-bisimilar to a set of elements in $\Bmf$. These elements of $\Bmf$ can be divided into those that satisfy $r_0$ and those that satisfy $r_1$. If $x_t$ is $\MLinu$-bisimilar to $k$ elements that satisfy $r_0$ and $k'$ elements that satisfy $r_1$, then we say that registers 0 and 1 contain $\log_2k$ and $\log_2 k'$, respectively, at time $t$.

An example of (part of) these intended models is shown in the following image, reproduced from Section~\ref{sec:undec}.\\
\centerline{
\begin{tikzpicture}[>=latex,line width=0.4pt,xscale = 1,yscale = 1]
\node[point,scale = 0.4,label=below:{\footnotesize $\Amf, a$},label=left:{\footnotesize $q_0$},label=right:{\footnotesize $N_c,e_0,e_1$}] (a0) at (0,0) {};
\node[point,scale = 0.4,label=left:{\footnotesize $q_1$}] (a1) at (0,1) {};
\node[point,scale = 0.4,label=left:{\footnotesize $q_4$}] (a2) at (0,2) {};
\node[point,scale = 0.4,label=left:{\footnotesize $q_4$},label=right:{\footnotesize $e_0$}] (a3) at (0,3) {};
\node[point,scale = 0.4,label=left:{\footnotesize $q_2$}] (a4) at (0,4) {};
\node[] at (0,4.5) {$\vdots$};
\draw[->,left] (a0) to node[] {\scriptsize $R$} (a1);
\draw[->,left] (a1) to node[] {\scriptsize $R$} (a2);
\draw[->,left] (a2) to node[] {\scriptsize $R$} (a3);
\draw[->,left] (a3) to node[] {\scriptsize $R$} (a4);
%
%
\node[point,scale = 0.4,label=below:{\footnotesize $\Bmf,b$},label=left:{\footnotesize $q_0,r_1$},label=right:{\footnotesize $N_c,e_0,e_1$}] (b0) at (4,0) {};
\node[point,scale = 0.4,label=left:{\footnotesize $q_1,r_0$},label=right:{\footnotesize $u$}] (l1) at (3,1) {};
\node[point,scale = 0.4,label=right:{\footnotesize $q_1,r_1$},label=left:{\footnotesize $u$}] (r1) at (5,1) {};
\node[point,scale = 0.4,label=left:{\footnotesize $q_4,r_0$},label=right:{\footnotesize $u$}] (l21) at (2,2) {};
\node[point,scale = 0.4,label=left:{\footnotesize $q_4,r_0$}] (l22) at (4,2) {};
\node[point,scale = 0.4,label=right:{\footnotesize $q_4,r_1$},label=left:{\footnotesize $u$}] (r2) at (5,2) {};
\node[point,scale = 0.4,label=right:{\footnotesize $q_4,r_1$},label=left:{\footnotesize $e_0,u$}] (r3) at (5,3) {};
\node[point,scale = 0.4,label=right:{\footnotesize $q_2,r_1$},label=left:{\footnotesize $u$}] (r4) at (5,4) {};
\node[point,scale = 0.4,label=left:{\footnotesize $q_4,r_0$},label=right:{\footnotesize $e_0,u$}] (l3) at (3,3) {};
\node[point,scale = 0.4,label=left:{\footnotesize $q_2,r_0$},label=right:{\footnotesize $u$}] (l4) at (3,4) {};
\draw[->,left] (b0) to node[] {\scriptsize $R$} (l1);
\draw[->,left] (b0) to node[] {\scriptsize $R$} (r1);
\draw[->,left] (l1) to node[] {\scriptsize $R$} (l21);
\draw[->,left] (l1) to node[] {\scriptsize $R$} (l22);
\draw[->,left] (r1) to node[] {\scriptsize $R$} (r2);
\draw[->,left] (r2) to node[] {\scriptsize $R$} (r3);
\draw[->,left] (r3) to node[] {\scriptsize $R$} (r4);
\draw[->,left] (l21) to node[] {\scriptsize $R$} (l3);
\draw[->,left] (l22) to node[] {\scriptsize $R$} (l3);
\draw[->,left] (l3) to node[] {\scriptsize $R$} (l4);
\node[] at (3,4.5) {$\vdots$};
\node[] at (5,4.5) {$\vdots$};
\end{tikzpicture}
}\\
In this example, we have $I_0 = -(0,q_1,q_1)$, $I_1 = +(0,q_4)$ and $I_4=-(0,q_2,q_4)$. So the machine starts in state $q_0$ by decrementing the (already empty) register $0$, and moves to state $q_1$. In this state, it increments register 0, and continues to state $q_4$. In $q_4$, it first decrements register 0 while staying in $q_4$. Then it tries to decrement register 0 again, but now that register is empty so the next state is $q_2$ instead of $q_4$.

It is straightforward (if somewhat labour intensive) to verify that the intended pointed models satisfy $\varphi_M$ and $\psi_M$, and that they are $\MLinu(\varrho)$-bisimilar. 

\begin{proposition}
	If a machine $M$ is non-halting, then there are pointed models $\Amf,x_{0}$ of $\varphi_M$ and $\Bmf,y_{0}$ of $\psi_M$ such that $\Amf,x_{0}\sim_{\MLinu(\varrho)}\Bmf,y_{0}$. 
\end{proposition}

Next, we need to show the converse:

\begin{proposition}
	\label{prop:GMLinu-reduction-back}
	If there are pointed models $\Amf,x_{0}$ of $\varphi_M$ and $\Bmf,y_{0}$ of $\psi_M$ such that $\Amf,x_{0}\sim_{\MLinu(\varrho)}\Bmf,y_{0}$,
	then $M$ is non-halting.
\end{proposition}
\begin{proof}
We show that $\Amf$ and $\Bmf$ contain intended models as sub-models. Because it is a model of $\varphi_M$, $\Amf$ contains a sequence of elements $x_0, x_1, \dots$, where $x_{0}=c$, $x_{t+1}$ is the unique $R$-successor of $x_t$, and $x_t$ is the unique $R$-predecessor of $x_{t+1}$. Furthermore, on each $x_t$ exactly one $q_i$ (with $0\leq i \leq \ell$) holds.
	
	In model $\Bmf$, we also have $y_{0}=c$. Let $Y_0 = \{y_0\}$, and for every $t$, let $Y_{t+1}$ be the $R$-successors of the elements of $Y_t$. Let $\bis\subseteq \dom(\Amf)\times \dom(\Bmf)$ be the $\MLinu(\varrho)$-bisimulation between $\Amf,x_{0}$ and $\Bmf,y_{0}$. The proof now proceeds through a series of claims.
	
	\textbf{Claim 1.} For every $t\in \omega$, $(i)$ if $(x_t,y)\in \bis$, then $y\in Y_t$, and $(ii)$ if $y\in Y_t$, then $(x_t,y)\in \bis$.
	
	\textbf{Proof of Claim 1.} We proceed by induction on $t$. As base case, note that $x_0$ and $y_0$ are the constants $c$ in their respective models, so $x_0$ and $y_0$ can only be $\MLinu(\varrho)$-bisimilar to each other. Since $\bis$ is, by assumption, an $\MLinu(\varrho)$-bisimulation, we must have $(x_0,y_0)\in \bis$. Assume then as the induction hypothesis that $t>0$ and that the claim holds for all $t'<t$.
	
	Take any $y$ such that $(x_t,y)\in \bis$. Because $\bis$ is an $\MLinu(\varrho)$-bisimulation and $x_t$ has $x_{t-1}$ as an $R$-predecessor, $y$ must have some $y'$ as an $R$-predecessor such that $(t_{x-1},y')\in \bis$. By the induction hypothesis, $y'\in Y_{t-1}$. This implies that $y$ is an $R$-successor of an element from $Y_{t-1}$, and so, by definition, we have $y\in Y_t$. 
	
	Now, take any $y\in Y_t$. By the definition of $Y_t$, there is some $y'\in Y_{t-1}$ that is an $r$-predecessor of $y$. By the induction hypothesis, $(x_{t-1},y)\in \bis$. Because $\bis$ is a bisimulation and $y$ is an $R$-successor of $y'$, there must be some $x'$ that is an $R$-successor of $x_{t-1}$ such that $(x',y)\in \bis$. The only $R$-successor of $x_{t-1}$ is $x_t$, so $x'=x_t$, and hence $(x_t,y)\in \bis$. This completes the induction step, and thereby the proof of Claim 1.

\smallskip

	\textbf{Claim 2.} For every $y\in Y_t$, if $y'$ is an $R$-predecessor of $y$, then $y'\in Y_{t-1}$.
	
	\textbf{Proof of Claim 2.} By Claim 1, $(x_t,y)\in \bis$. Since $y'$ is an $R$-predecessor of $y$, we must have $(y',x')\in \bis$ for some $R$-predecessor $x'$ of $x_t$. The only $R$-predecessor of $x_t$ is $x_{t-1}$, so we have $(x_{t-1},y')\in \bis$. By Claim 1, this implies that $y'\in Y_{t-1}$. 

\smallskip

	\textbf{Claim 3.} For every $t$, there is exactly one $i_t$ such that $q_{i_t}$ holds on any of $x_t$ and $Y_t$. This $q_{i_t}$ holds on all of $x_t$ and $Y_t$. Furthermore, $e_\ell$ either holds on all of $x_t$ and $Y_t$, or on none.
	
	\textbf{Proof of Claim 3.} Because $\Amf$ is a model of $\varphi_M$, every $x_t$ satisfies exactly one $q_{i_t}$. The rest of the claim follows immediately from $\bis$ being a bisimulation.
	
	Before continuing with further claims, we define $Y_t^\ell$, for $t>0$ and $\ell\in \{0,1\}$, by taking $Y_t^\ell = \{y\in Y_t\mid \Bmf,y\models r_\ell\}$. If $\ell\in \{0,1\}$, we write $\overline{\ell}$ for $1-\ell$, i.e., $\overline{0}=1$ and $\overline{1}=0$.
	
	\textbf{Claim 4.} For $\ell\in \{0,1\}$ and $t>0$, the $R$-successors of $Y_t^\ell$ are in $Y_{t+1}^\ell$ and the $R$-predecessors of $Y_{t+1}^\ell$ are in $Y_t^\ell$.
	
	\textbf{Proof of Claim 4.} Conjunct~\eqref{form2:prop1} says that $r_\ell$ propagates forward and backward trough $R$.
	
	\textbf{Claim 5.} For every $t>0$,
	\begin{itemize}
		\item[--] if $I_{i_t}=+(\ell,q_j)$, then $i_{t+1}=j$ and $|Y_{t+1}^p\ell=2\times |Y_t^\ell|$,
		
		\item[--] if $I_{i_t}=-(\ell,q_j,q_k)$ and $e_\ell$ holds on $x_t$, then $i_{t+1}=j$ and $|Y_{t+1}^\ell|=|Y_t^\ell|$,
		
		\item[--] if $I_{i_t}=-(\ell,q_j,q_k)$ and $e_\ell$ does not hold on $x_t$, then $i_{t+1}=k$ and $|Y_{t+1}^\ell|=\frac{1}{2}\times|Y_t^\ell|$,
	\end{itemize}
	and, in each case, $|Y_{t+1}^{\overline{\ell}}|=|Y_t^{\overline{\ell}}|$.
	
	\textbf{Proof of Claim 5.} Conjuncts~\eqref{form:incr}, \eqref{form:decr1} and \eqref{form:decr2} of $\varphi_M$ guarantee that $i_{t+1}$ has the appropriate value, i.e., the state transitions as represented in the models $\Amf$ and $\Bmf$ follow the rules of $M$. 
	
	Consider now the sets $Y_t^\ell$ and $Y_{t+1}^\ell$. By Claim 4, all the $R$-successor of $Y_t^\ell$ are in $Y_{t+1}^\ell$, and all the $R$-predecessors of $Y_{t+1}^\ell$ are in $Y_t^\ell$. Furthermore, conjunct~\eqref{form2:incr} of $\psi_M$ implies that, if $I_{i+t}=+(\ell,q_j)$, then every $y\in Y_t^\ell$ has exactly two $R$-successors that each have exactly one $R$-predecessor. It follows that $|Y_{t+1}^\ell|=2\times |Y_t^\ell|$. Similarly, if $I_{i_t}=-(\ell,q_j,q_k)$ and $e_\ell$ holds, then conjunct~\eqref{form2:decr1} says that every $y\in Y_t^\ell$ has exactly one $R$-successor that has exactly one $R$-predecessor, so $|Y_{t+1}^\ell|=|Y_t^\ell|$. If $I_{i_t}=-(\ell,q_j,q_k)$ and $e_\ell$ does not hold, then conjunct~\eqref{form2:decr2} makes sure that every $y\in Y_t^\ell$ has exactly one $R$-successor that has exactly two $R$-predecessors, so $|Y_{t+1}^\ell|=\frac{1}{2}\times|Y_t^\ell|$.
	
	Finally, conjuncts~\eqref{form2:incr2}, \eqref{form2:decr3} guarantee that every $y\in Y_t^{\overline{\ell}}$ has exactly one $R$-successor that has exactly one $R$-predecessor, so $|Y_{t+1}^{\overline{\ell}}|=|Y_t^{\overline{\ell}}|$.

\smallskip

\textbf{Claim 6.} $|Y_1^0|=|Y_1^1|=1$.

	\textbf{Proof of Claim 6.} By conjuncts~\eqref{form2:init1} and \eqref{form2:init2} of $\psi_M$, $c$ has exactly two successors, one of which satisfies $u\wedge r_0$ while the other satisfies $u\wedge r_1$.

\smallskip
	
	\textbf{Claim 7.} For any $t>0$ and $\ell\in \{0,1\}$, there is exactly one $y\in Y_t^\ell$ that satisfies $u$.
	
	\textbf{Proof of Claim 7.} For $t=1$, the claim holds because, by conjunct~\eqref{form2:init2} of $\psi_M$, there is at least one $y\in Y_1^\ell$ that satisfies $u$. As $Y_1^\ell$ is a singleton (by Claim 6), there must be exactly one such $y$. Conjunct~\eqref{form2:prop2} implies that every $u$-element (other than $c$) has exactly one $R$-predecessor and one $R$-successor that satisfy $u$, so if $Y_{t}^\ell$ has exactly one element satisfying $u$, then so does $Y_{t+1}^\ell$.

\smallskip
	
	\textbf{Claim 8.} For any $t>0$ and $\ell\in \{0,1\}$,
	\begin{itemize}
		\item[--] if $e_\ell$ holds on $x_t$ and $Y_t$, then $Y_t^\ell$ is a singleton and
		\item[--] if $I_t = -(\ell,q_j,q_k)$ and $Y_t^\ell$ is a singleton, then $e_\ell$ holds on $x_t$ and $Y_t$.
	\end{itemize}
	
	\textbf{Proof of Claim 8.} Suppose that $e_\ell$ holds on $x_t$ and $Y_t$. Then $e_\ell\land r_\ell$ is true on every element of $Y_t^\ell$. By conjunct~\eqref{form2:empty}, this implies that $u$ is true on all of $Y_t^\ell$. Claim 7 says that only one element of $Y_t^p$ can satisfy $u$, so $Y_t^\ell$ must be a singleton.
	
	Suppose next that $I_t= - (\ell,q_j,q_k)$ and $Y_t^\ell$ is a singleton. Assume (towards a contradiction) that $e_\ell$ does not hold on $x_t$. Then, by Claim 5, $|Y_{t+1}^\ell|=\frac{1}{2}\times|Y_t^\ell|$. But $Y_{t+1}^\ell$ cannot have half as many elements as $Y_t^\ell$, since $Y_t^\ell$ is a singleton. From this contradiction, we conclude that $e_\ell$ holds on $x_t$ (and therefore also on $Y_t$).

\smallskip
	
	\textbf{Claim 9.} For any $t$ and $\ell\in \{0,1\}$, let $v^\ell_t$ be the value in register $\ell$ at time $t$ in the run of $M$ and $q_{i_t}$ the state of $M$ at time $t$. Then
	\begin{itemize}
		\item[--] $q_{i_t}$ holds on $x_t$ and $Y_t$, and
		
		\item[--] if $t>0$, then $|Y_t^\ell|=2^{v_t^\ell}$.
	\end{itemize}
	
	\textbf{Proof of Claim 9.} As we remarked when we introduced them, conjuncts~\eqref{form:incr}--\eqref{form:decr2} of $\varphi_M$ guarantee that the state transition instructions of $M$ are obeyed, if $e_\ell$ holds in the appropriate places. Claim 5 shows that $\Bmf$ follows the incrementation/decrementation instructions of $M$, again under the assumption that $e_\ell$ holds where appropriate. Finally, Claim 8 shows that $e_\ell$ holds where appropriate.

\smallskip
	
	\textbf{Claim 10.} $M$ is non-halting.
	
	\textbf{Proof of Claim 10.} Follows from conjunct~\eqref{form:nonhalt} of $\phi_M$ together with Claim 9.
	
	This completes the proof of Proposition~\ref{prop:GMLinu-reduction-back}.
\end{proof}

Thus, we have shown that Craig $\GMLinu/\MLinu$-separation is RE-complete. This is also true of unary Craig $\GMLinu/\MLinu$-separation because $\varphi_M$ and $\psi_M$ have the same binary predicates. And using the second claim of Lemma~\ref{l:red}, we obtain that $\GMLinu/\MLinu$-separation is RE-complete.

By using the so-called spy-point technique~\cite{DBLP:journals/jsyml/ArecesBM01}, we can adapt  the proof to the case of (Craig) $\GMLin/\MLin$-separation (without the universal diamond $\Du$ and its dual universal box $\blacksquare$). Indeed, let $\varphi_M$ and $\psi_M$ be as defined above, except that we replace $\blacksquare$ by $\square_U$ and add the following  conjuncts to both of the resulting formulas:
$$
\Diamond\!_U N_c, \quad \square_U (\square_R\Diamond_{U^{-1}}N_c \wedge \square_{R^{-1}}\Diamond_{U^{-1}}N_c). 
$$
These additional conjuncts guarantees that, while $U$ is not necessarily interpreted as the universal relation, all elements that are relevant in the proof are $U$-successors of $c$, and hence all conjuncts that start with $\square_U$ have the desired effect. 

To deal with the nominal-free $\LS$, we can simulate $N_c$ by adding to $\varphi_M$ the conjunct $N_c \leftrightarrow p$ with a fresh propositional variable $p$ and replace $N_c$ in $\psi_M$ with $p$, making it shared in $\varphi_M$ and $\psi_M$. 

This proves the second item of Theorem~\ref{thm:RE}.


\subsection{$\CT/\FOT$-separation}

We now show the first item of Theorem~\ref{thm:RE}. Our reduction is very similar to the one given in the previous subsection for $\GMLinu/\MLinu$-separation. We begin by discussing the required modifications.

We again construct two formulas, $\xi_M$ and $\chi_M$, such that the pointed model $\Amf,x_{0}$ of $\xi_M$ is essentially an $R$-chain
encoding the the run of $M$, and the model $\Bmf,y_{0}$ of $\chi_M$ represents the contents of the registers. Synchronisation between $\Amf$ and $\Bmf$ is done through the bisimulation between them. All of this is exactly as in the $\GMLinu/\MLinu$ case.
Observe, however, that now we deal with a $\FOT(\varrho)$-bisimulation, which is more constraining than an $\MLinu(\varrho)$-bisimulation. Recall that a $\FOT(\varrho)$-bisimulation requires that if $(x,y)\in \bis$, then, for every $x'$, there is an $y'$ such that $(x',y')\in \bis$ and $(x,x')\mapsto (y,y')$ is a partial isomorphism, and similarly, for every $y'$, there must be an $x'$ with the same properties. Here, $(x,x')\mapsto (y,y')$ being a partial isomorphism means that (i) for every unary predicate $P$ in the relevant signature, $P(x)$ iff $P(y)$ and $P(x')$ iff $P(y')$, (ii) for every binary predicate $R$ in the relevant signature, $R(x,x')$ iff $R(y,y')$ and $R(x',x)$ iff $R(y',y')$, and (iii) $x=x'$ iff $y=y'$.

Perhaps surprisingly, it is only condition (iii) that will give us trouble, and that requires us to make the current proof slightly more complicated than the $\GMLinu/\MLinu$ one. In the intended pointed models of $\varphi_M$ and $\psi_M$, we can already choose a $y'$ (or $x'$) that satisfies properties (i) and (ii). Condition (iii), however, is not satisfied for those models. 
This is because condition (iii) implies that an $\FOT$-bisimulation can never relate a singleton set to a multi-element set (i.e., if $(x,y)\in \bis$ and $(x,y')\in \bis$ with $y\not = y'$, then there must be a $x'\not = x$ such that $(x',y)\in \bis$ and $(x',y')\in \bis$). After all, if $(x,y), (x,y')\in \bis$ with $y\not = y'$, then there must be some $x'$ such that $(x',y')\in \bis$ and $(x,x')\mapsto (y,y')$ is a partial isomorphism. In particular, we must then have $x=x'$ iff $y=y'$, so from $y\not = y'$ it follows that $x\not = x'$.

In order to address this issue, we make a small modification to the intended models for $\varphi_M$ and $\psi_M$. Specifically, we create two copies of everything: the intended models of $\xi_M$ and $\chi_M$ are simply two disjoint copies of the models for $\varphi_M$ and $\psi_M$, respectively.
Now, let us consider the precise formulas $\xi_M$ and $\chi_M$ that achieve this effect. We begin with $\xi_M$, which is the conjunction of the following formulas:
\begin{enumerate}
	\item \label{form3:exists}$\exists^{=2} x\, q_0(x)$,
	\item $\forall x \, \bigvee_{1\leq i\leq m}q_i(x)$,
	\item $\forall x \, \bigwedge_{1\leq i \leq m}(q_i(x)\rightarrow \bigwedge_{i<j\leq m}\neg q_j(x))$,
	\item $\forall x \, (q_0(x)\rightarrow \forall y \neg R(y,x))$,
	\item $\forall x \, \exists^{=1}y R(x,y)$,
	\item $\forall x \, \neg q_0(x)\rightarrow \exists^{=1}y R(y,x)$;
\end{enumerate}
if $I_i=+(\ell,q_j)$, then 
\begin{enumerate}[resume]
	\item $\forall x \,(q_i(x)\rightarrow \forall y \, (R(x,y)\rightarrow q_j(y)))$,
\end{enumerate}
and if $I_i = -(\ell,q_j,q_k)$, then
\begin{enumerate}[resume]
	\item $\forall x \,((q_i(x)\wedge e_\ell(x))\rightarrow \forall y\,(R(x,y)\rightarrow q_j(y)))$,
	\item $\forall x \,((q_i(x)\wedge \neg e_\ell(x))\rightarrow \forall y\,(R(x,y)\rightarrow q_k(y)))$;
\end{enumerate}
and, finally,
\begin{enumerate}[resume]
	\item $\forall x \, \neg q_m(x)$
\end{enumerate}
The reader may notice that most conjuncts of $\xi_M$ are simply translations of the conjunct of $\varphi_M$ with the same number to $\CT$. The only exceptions, and this is a very minor exception, is conjuncts \eqref{form3:exists}. The corresponding $\varphi_M$ conjuncts guarantee the existence of one $q_0$ element, while here we guarantee the existence of two $q_0$ elements.

The conjuncts of $\chi_M$ are as follows:
\begin{enumerate}
	\item \label{form4:exists}$\exists^{=2} x\, q_0(x)$,
	\item $\forall x \, \bigvee_{1\leq i\leq m}q_i(x)$,
	\item $\forall x \, \bigwedge_{1\leq i \leq m}(q_i(x)\rightarrow \bigwedge_{i<j\leq m}\neg q_j(x))$,
	\item $\forall x \, (q_0(x)\rightarrow \exists^{=2}y \, R(x,y))$,
	\item $\forall x \, (q_0(x)\rightarrow \bigwedge_{\ell\in \{0,1\}}\exists^{=1}y \, (R(x,y)\wedge u(y)\wedge r_\ell(y)))$,
	\item $\forall x \, (q_0(x) \rightarrow (e_0(x)\wedge e_1(x)))$,
	\item $\forall x \, (r_\ell(x)\rightarrow (\forall y \, ((R(x,y)\vee R(y,x))\rightarrow (r_\ell(y)\vee q_0(y)))))$,
	\item $\forall x \, ((u(x)\wedge \neg q_0(x))\rightarrow (\exists^{=1}y\, (R(x,y)\wedge u(y))\wedge{}\\ \hspace*{4.7cm} \exists^{=1}y\, (R(y,x)\wedge u(y))))$;
\end{enumerate}
if $I_i = +(\ell,q_j)$, then
\begin{enumerate}[resume]
	\item $\forall x \, ((q_i(x)\wedge r_\ell(x))\rightarrow (\exists^{=2}y\, R(x,y) \wedge \forall y \,(R(x,y)\rightarrow{}\\ \hspace*{5.9cm}  \exists^{=1}x \, R(x,y))))$,
	\item $\forall x \, ((q_i(x)\wedge \neg r_\ell(x))\rightarrow (\exists^{=1}y\, R(x,y) \wedge \forall y \,(R(x,y)\rightarrow {}\\ \hspace*{5.9cm} \exists^{=1}x \, R(x,y))))$;
\end{enumerate}
if $I_i = -(\ell,q_j,q_j)$, then 
\begin{enumerate}[resume]
	\item $\forall x \, ((q_i(x)\wedge r_\ell(x) \wedge e_\ell)\rightarrow (\exists^{=1}y\, R(x,y) \wedge \forall y \,(R(x,y)\rightarrow {}\\ \hspace*{5.9cm} \exists^{=1}x \, R(x,y))))$,
	\item $\forall x \, ((q_i(x)\wedge r_\ell(x) \wedge \neg e_\ell)\rightarrow (\exists^{=1}y\, R(x,y) \wedge \forall y \,(R(x,y)\rightarrow {}\\ \hspace*{5.9cm} \exists^{=2}x \, R(x,y))))$,
	\item $\forall x \, ((q_i(x)\wedge \neg r_\ell(x))\rightarrow (\exists^{=1}y\, R(x,y) \wedge \forall y \,(R(x,y)\rightarrow {}\\ \hspace*{5.9cm}\exists^{=1}x \, R(x,y))))$;
\end{enumerate}
and finally
\begin{enumerate}[resume]
	\item $\forall x \, ((e_\ell(x)\wedge r_\ell(x))\rightarrow u(x))$.
\end{enumerate}

\begin{proposition}
	If $M$ is non-halting, then there are pointed models $\Amf,x_{0}$ and $\Bmf,y_{0}$ with $\Amf\models \xi_M(x_{0})$ and $\Bmf\models \chi_{M}(y_{0})$ 
	such that $\Amf,x_{0}\sim_{\FOT(\varrho)}\Bmf,y_{0}$.
\end{proposition}
\begin{proof}
	The models in question are, as discussed above, two disjoint copies of the intended models of $\varphi_M$ and $\psi_M$. 
%
It is straightforward to verify that these structures are indeed models of $\xi$ and $\chi$, and that they are $\FOT(\varrho)$-bisimilar.
\end{proof}

Next, the other direction:

\begin{proposition}
	If there are pointed structures $\Amf,x_{0}$ and $\Bmf,y_{0}$ such that $\Amf\models \xi_M(x_{0})$, $\Bmf\models \chi_{M}(y_{0})$ 
	and $\Amf,x_{0}\sim_{\FOT(\varrho)}\Bmf,y_{0}$,
	then $M$ is non-halting.
\end{proposition}
\begin{proof}
	This proof is mostly similar to that of Proposition~\ref{prop:GMLinu-reduction-back}, except for the first few claims, in which we show certain elements to be bisimilar. We therefore prove the first few claims in detail, and refer to the proof of Proposition~\ref{prop:GMLinu-reduction-back} for details on the remainder.
	
	Let $\bis$ be the bisimulation between $\Amf$ and $\Bmf$, and let $A_0$ be the elements in $\Amf$ that satisfy $q_0$ and $B_0$ the ones in $\Bmf$. Then, for $n\in \mathbb{N}$, let $A_{n+1}$ be the set of $R$-successors of elements of $A_n$ and $B_{n+1}$ the set of $R$-successors of $B_n$. 
	
	\textbf{Claim 1.} $\bis$ can only relate an element of $A_n$ to elements of $B_n$, and vice versa. 
	
	\textbf{Proof of Claim 1.} We proceed by induction. We show the base case for one direction, the other direction can be shown similarly. Suppose, towards a contradiction, that $(a,b)\in \bis$ with $a\in A_0$ and $b\not \in B_0$. Then, for any $x$, $(a,a)\mapsto (b,x)$ is not a partial isomorphism, since $q_0(a)$ but $\neg q_0(b)$. This contradicts $\bis$ being a bisimulation, so such $a$ and $b$ cannot exist.
	
	Assume then, as induction hypothesis, that the claim holds for all $n'<n$. Again, we show one direction: that $(a,b)\in \bis$ and $a\in A_n$. The element $a$ has an $a$-predecessor $a'\in A_{n-1}$. Because $\bis$ is a bisimulation, there must be a $b'$ such that $(b,b')\in \bis$ and $(a,a')\mapsto (b,b')$ is a partial isomorphism. We have $R(a',a)$, so we must also have $R(b',b)$. Furthermore, by the induction hypothesis, $(a',b')\in \bis$, together with $a'\in A_{n-1}$, implies that $b'\in B_{n-1}$. Hence $b$ is the $R$-successor of an element in $B_{n-1}$, which, by definition, means $b\in B_n$. This completes the induction step, thereby proving the claim.
	
	\textbf{Claim 2.} If $a\in A_{n+1}$ and $R(a',a)$, then $a'\in A_n$, and if $b\in B_{n+1}$ and $R(b',b)$, then $b'\in B_n$.
	
	\textbf{Proof of Claim 2.} In $\Amf$, every element, other than the two that satisfy $q_0$, has exactly one $R$-predecessor. So it immediately follows that every $R$-predecessor of $a\in A_{n+1}$ is in $A_n$. 
	Now, take any $b\in B_{n+1}$ and $b'$ such that $R(b',b)$. There must be some $a$ such that $(a,b)\in \bis$. By Claim 1, this implies that $a\in A_{n+1}$. Furthermore, there is some $a'$ such that $(a',b')\in \bis$ and $(a',a)\mapsto (b',b)$ is a partial isomorphism. As $R(b',b)$, this implies that $R(a',a)$ as well, so we have $a'\in A_n$. By Claim 1, this implies that $b'\in B_n$.
	
	\textbf{Claim 3.} If $a_0\in A_0$ and $b_0\in B_0$, then $(a_0,b_0)\in \bis$.
	
	\textbf{Proof of Claim 3.} Let $b_0$ be one of the two elements of $B_0$. Then there is at least one $a_0\in A_0$ such that $(a_0,b_0)\in\bis$. Let $a_0'$ be the other element of $A_0$. This $b_0$ has exactly two $R$-successors $b_1,b_1'\in B_1$ (with $b_1\not = b_1'$). There must be $a_1,a_1'$ such that $(a_1,b_1)\in\bis$, $(a_1',b_1')\in\bis$, and $(a_0,a_1)\mapsto (b_0,b_1)$ and $(a_0,a_1')\mapsto (b_0,b_1')$ are partial isomorphisms. From $R(b_0,b_1)$ and $R(b_0,b_1')$ it therefore follows that $R(a_0,a_1)$ and $R(a_0,a_0')$. As $a_0$ has exactly one $R$-successor, this implies that $a_0=a_0'$.
	
	Now, consider the pair $(b_1,b_1')$. We have $(a_1,b_1)\in\bis$, so there must be some $a_1''$ such that $(a_1'',b_1')\in\bis$ and $(a_1,a_1'')\mapsto (b_1,b_1')$ is a partial isomorphism. This implies that, in particular, $a_1=a_1''$ if and only if $b_1= b_1'$. We have $b_1\not = b_1'$, so $a_1\not = a_1''$. From $(a_1'', b_1')\in\bis$ and Claim 1, it follows that $a_1''\in A_1$. Because $a_1\not = a_1''$, this $a_1''$ must be the other element of $A_1$, which is the $R$-successor of $a_0'$. 
	
	Because $(a_1'', b_1')\in\bis$ and $R(a_0',a_1'')$, there must be some $b_0'$ such that $(a_0',b_0')\in\bis$ and $R(b_0',b_1')$. But $b_1'$ is an $R$-successor of $b_0$ and has only one $R$-predecessor, so $b_0'=b_0'$. Hence $(a_0',b_0')\in\bis$ implies $(a_0',b_0)\in\bis$. We have now shown that, for arbitrary $b_0\in B_0$, we have $(a_0,b_0)\in\bis$ and $(a_0',b_0)\in\bis$, with $a_0\not = a_0'$. The claim follows immediately.
	
	\textbf{Claim 4.} If $a_n\in A_n$ and $b_n\in B_n$, then $(a_n,b_n)\in\bis$.
	
	\textbf{Proof of Claim 4.} The proof is by induction. The base case is Claim 3. So assume as induction hypothesis that $n>0$ and that Claim 4 holds for all $n'<n$. Take any $a_n\in A_n$ and $b_n\in B_n$. There are $a_{n-1}\in A_{n-1}$ and $b_{n-1}\in B_{n-1}$ such that $R(a_{n-1},a_n)$ and $R(b_{n-1},b_n)$. By the induction hypothesis, $(a_{n-1},b_{n-1})\in\bis$. Now, there must be some $a_n'$ such that $(a_n',b_n)\in\bis$ and $(a_{n-1},a_n')\mapsto (b_{n-1},b_n)$ is a partial isomorphism. In particular, because $R(b_{n-1},b_n)$, we have $R(a_{n-1},a_n')$.
	
	We already had $R(a_{n-1},a_n)$, and $a_{n-1}$ has only one $R$-successor, so $a_n = a_n'$. As we have $(a_n',b_n)\in\bis$, this shows that $(a_n,b_n)\in\bis$, thereby completing the induction step and therefore proving the claim.
	
	The remainder of the proof proceeds in the same way as the proof of Proposition~\ref{prop:GMLinu-reduction-back}.
\end{proof}

Thus, we have established that (Craig) $\CT/\FOT$-separation is RE-complete. It remains to show how the above proofs can be modified for the pairs $\CT/L'$, where $\LS$ is any counting-free modal logic in Table~\ref{table:results}.
%
%
Because $\GMLinu$ can be viewed as a fragment of $\CT$, it follows from the second item of Theorem~\ref{thm:RE} that (Craig) $\CT/\LS$-separation is undecidable for all logics $\LS$  with inverse. For modal logics without inverse, we can use the fact that $\CT$ can define inverse relations: we can add the conjunct $\forall x \, \forall y\, (R(x,y)\leftrightarrow T(y,x))$. Then, in models of the relevant formulas, any modal logic bisimulation will respect the relation $T$ in addition to the relation $R$, and therefore also respect inverses. This completes the proof of Theorem~\ref{thm:RE}.


\section*{Proofs for Section~\ref{sec:gmlu}}


\complexGMLnu* 

\begin{proof}
First, we show the upper bounds. By Lemma~\ref{l:red}, it suffices to consider Craig separation. Let $\varphi_{1}$, $\varphi_{2}$ be $\GMLnu$-formulas, $\sigma = \sig(\varphi_{1}) \cup \sig(\varphi_{2})$ and $\varrho = \sig(\varphi_{1}) \cap \sig(\varphi_{2})$. Denote by $\mx$ the maximal number $k$ such that $\Diamond_R^{\ge k}$ occurs in $\varphi_{1}$ or $\varphi_{2}$, for some $R$. 
A \emph{type} (\emph{for $\varphi_{1}$ and $\varphi_{2}$}) is any set $\type \subseteq \sub(\varphi_{1},\varphi_{2})$, for which there is a pointed $\sigma$-structure $\Amf, a$ such that $\type=\type(\Amf,a)$, where 
\[
\type(\Amf,a) = \{ \chi\in \sub(\varphi_{1},\varphi_{2})\mid \Amf \models \chi(a)\}
\] 
is the \emph{type} \emph{realised by $a$ in $\Amf$}. Denote by $\Types$ the set of all such types. Clearly, $|\Types| \le 2^{|\varphi_{1}|+|\varphi_{2}|}$ and $\Types$ can be computed in time exponential in $|\varphi_{1}| + |\varphi_{2}|$.
Types $\type_1$ and $\type_2$ are called \emph{$\Du$-equivalent}
in case $\Du\chi\in \type_1$ iff $\Du\chi\in \type_2$, for all $\Du\chi\in \sub(\varphi,\psi)$. Note that the types realised in the same structure are $\Du$-equivalent to each other. 

We check the existence of $\sigma$-structures $\Amf_1,a_1$ and $\Amf,a_2$ such that $\Amf_i\models \varphi_i(a_i)$, $i=1,2$,  and $\Amf_1,a_1 \sim_{\MLnu(\varrho)} \Amf_2,a_2$ using a mosaic-elimination procedure, a generalisation of the standard type-elimination procedure for satisfiability-checking; see, e.g.,~\cite[p.~72]{GabEtAl03}.

A \emph{mosaic} (\emph{for $\varphi_{1}$ and $\varphi_{2}$}) is a pair $M=(M_{1},M_{2})$ of non-empty sets $M_{1},M_{2}$ of types such that the following hold: 
\begin{itemize}
\item[--] $A\in \type$ iff $A\in \type'$, for all $A\in \varrho$ and $\type,\type'\in M_{1}\cup M_{2}$;
 
\item[--] $N_c\in \type$ iff $N_c\in \type'$, for all $c \in \varrho$ and $\type,\type'\in M_{1}\cup M_{2}$; 

\item[--] 
if $\type, \type' \in M_{i}$, $i \in \{1,2\}$, $N_c\in \type \cap \type'$, then $\type = \type'$, for all $c\in \sigma$;
	
\item[--] all types in each $M_{i}$ are $\Du$-equivalent to each other.
\end{itemize}
%
We say that $N_c$ \emph{occurs} in $M$ if $N_c \in \type\in M_{i}$, for some $\type$ and $i\in \{1,2\}$.
As follows from the second and third items, if $N_c$ occurs in  $M$, for some $c\in \varrho$, then $M_1$ and $M_2$ are singletons.

There are at most doubly exponentially many mosaics, and they can be computed in double exponential time. Let $f$ be the function defined by taking $f(M) = 1$ if $M$ contains a type with a nominal 
and $f(M) = \mx$ otherwise.

The definition of mosaic above is an abstraction from the following semantical notion. Given structures $\Amf_{1}$ and $\Amf_{2}$, the
\emph{mosaic defined by $a\in \dom(\Amf_{i})$ in $\Amf_{1},\Amf_{2}$}, $i \in \{1,2\}$, 
is the pair $M(\Amf_{1},\Amf_{2},a)=(M_{1}(\Amf_{1},a),M_{2}(\Amf_{2},a))$, where, for $j \in \{1,2\}$,
\[
M_{j}(\Amf_{j},a) = \{ \type(\Amf_{j},b) \mid b\in
\dom(\Amf_{j}), \ \Amf_{i},a \sim_{\MLnu(\varrho)}\Amf_{j},b\}.
\] 
Clearly, $M(\Amf_{1},\Amf_{2},a)$ is a mosaic and $M(\Amf_{1},\Amf_{2},a) = M(\Amf_{1},\Amf_{2},a')$ if $a \sim_{\MLnu(\varrho)} a'$, for any $a,a' \in \dom(\Amf_{1}) \cup \dom(\Amf_{2})$.

We aim to build the required structures $\Amf_{i}$, $i=1,2$, and bisimulation $\bis$ between them---if this is at all possible---from copies of suitable pairs $(\type,M)$ with a mosaic $M=(M_{1},M_{2})$ and a type $\type\in M_i$, for which we, roughly, aim to  set $(\type_{1},M) \bis (\type_{2},M)$. 
To satisfy the graded modalities $\Diamond_{R}^{\geq k}\chi$ in the $\Amf_{i}$, we require the following definition.
%
Let $M=(M_{1},M_{2})$ be a mosaic, $\mathcal{M}$ a set of mosaics, and some type in $M_1\cup M_2$ contain a formula $\Diamond_{R}^{\geq k}\chi$. We call $\mathcal{M}$ an \emph{$R$-witness for $M$} if there exists a relation $R^{M,\mathcal{M}} \subseteq (\Delta_{1} \times \Gamma_{1})\cup(\Delta_{2}\times \Gamma_{2})$ with  
\begin{align*}
& \Delta_{i} =\{ (\type,M) \mid \type\in M_{i}\},\\ 
& \Gamma_{i}=\{ (\type',M',j) \mid \type'\in M_{i}',\ M'=(M_{1}',M_{2}')\in \mathcal{M},\ 1 \le j \le f(M')\}
\end{align*}
satisfying the following for $\Delta=\Delta_{1} \cup \Delta_{2}$ and $\Gamma= \Gamma_{1}\cup \Gamma_{2}$:
\begin{description}

\item[($\Du$-harmony)] all types in $M_{i}$ and $M_{i}'$, for any $M'=(M_{1}',M_{2}')$ in $\mathcal{M}$, are $\Du$-equivalent to each other, $i=1,2$;
	
\item[(witness)] for any $\Diamond^{\geq k}_R\chi \in \sub(\varphi_1,\varphi_2)$ and $(\type,M)\in \Delta$, we have 
$$
\mbox{}\hspace*{-7mm}\Diamond^{\geq k}_R\chi \in \type \text{ iff } \big| \{ (\type',M',j) \in \Gamma \mid \chi\in \type', (\type,M) R^{M,\mathcal{M}} (\type',M',j) \}\big| \ge k;
$$
%
	
\item[(bisim)] if $R\in \varrho$, $M'\in \mathcal{M}$, and $(\type_{0},M)\in \Delta$ has an $R^{M,\mathcal{M}}$-successor $(\type_{0}',M',j)\in \Gamma$, then each $(\type_{1},M)\in \Delta$ has an $R^{M,\mathcal{M}}$-successor of the form $(\type_{1}',M',j')\in \Gamma$. 
\end{description}
The intuition behind conditions {\bf ($\Du$-harmony)} and {\bf (witness)} should be clear, and {\bf (bisim)} is needed to ensure that if, for any fixed $M'$, all $(\type',M',j)\in \Gamma$ are mutually $\MLnu(\varrho)$-bisimilar, then all $(\type,M)\in \Delta$ are also $\MLnu(\varrho)$-bisimilar to each other. 
It should also be clear that if $\Amf_1, a_1 \sim_{\MLnu(\varrho)}\Amf_2,a_2$, then every mosaic $M(\Amf_{1},\Amf_{2},a)$, for $a \in \dom(\Amf_i)$, $i=1,2$, has an $R$-witness for any relevant $R$, namely, the set $\mathcal{M}$ of all mosaics $M(\Amf_{1},\Amf_{2},b)$ such that $a' R^{\Amf_j} b$, for some $j \in \{1,2\}$ and $a' \in \dom(\Amf_j)$ with $M(\Amf_{1},\Amf_{2},a) = M(\Amf_{1},\Amf_{2},a')$.
Observe also that every $R$-witness $\mathcal{M}$ for $M$ contains a subset $\mathcal{M}'\subseteq \mathcal{M}$ such that $\mathcal{M}'$ is an $R$-witness for $M$ and $|\mathcal{M}'|\leq \mx 2^{|\varphi_{1}|+|\varphi_{2}|}$. 
%
So, from now on, we assume all witnesses to be of size $\leq \mx 2^{|\varphi_{1}|+|\varphi_{2}|}$.

We can now define the mosaic elimination procedure for a set $\mathfrak{S}$ of mosaics. Call $M=(M_{1},M_{2})\in \mathfrak{S}$ \emph{bad} in $\mathfrak{S}$ if at least one of the following conditions is \emph{not} satisfied for $i = 1,2$:
\begin{itemize}
\item[--] if $\Diamond_{R}^{\geq k}\chi\in \type\in M_{i}$, then there is an $R$-witness $\mathcal{M}\subseteq \mathfrak{S}$ for $M$;
	
\item[--] if $\Du\chi \in \type\in M_{i}$, then there are $(M_{1}',M_{2}')\in \mathfrak{S}$ and $\type'\in M'_{i}$ such that $\chi\in \type' \in M_{i}'$ and $\type,\type'$ are $\Du$-equivalent.
\end{itemize}

\begin{lemma}\label{lem:decidebad1} 
Given a set $\mathfrak{S}$ of mosaics and $M\in \mathfrak{S}$, it can be decided in double exponential time in $|\varphi_{1}|+|\varphi_{2}|$ whether $M$ is bad in $\mathfrak{S}$.
\end{lemma}
\begin{proof}
There are at most doubly exponentially-many candidate sets $\mathcal{M}\subseteq  \mathfrak{S}$ with $|\mathcal{M}|\leq \mx 2^{|\varphi_{1}|+|\varphi_{2}|}$ and binary relations $R^{M,\mathcal{M}}\subseteq (\Delta_{1} \times \Gamma_{1})\cup(\Delta_{2}\times \Gamma_{2})$. We simply check that none of the pairs $\mathcal{M}$, $R^{M,\mathcal{M}}$ satisfies {\bf ($\Du$-harmony)}, {\bf (witness)} and {\bf (bisim)}.
\end{proof}

Let $\mathfrak{S}_{0}$ be the set of all mosaics. Let $h$ be a function that, for each pair $c \in \sigma$ and $i \in \{1,2\}$, selects one mosaic $h(N_c,i) = (M_1,M_2)$ such that
\begin{itemize}
\item[--] $N_c \in \type$, for some $\type \in M_i$;

\item[--] if $N_{c'} \in \type'$, for some $\type' \in M_j$, $j \in \{1,2\}$ and $c'\in \sigma$, then $h(N_c,i) = h(N_{c'},j)$.

\end{itemize}
By the definition of mosaics, $h(N_c,1)=h(N_c,2)$ for all $c\in \varrho$.

Given such $h$, let $\mathfrak{S}_{0}^{h} \subseteq \mathfrak{S}_{0}$ comprise all mosaics of the form $h(N_c,i)$ and those that do not have occurrences of nominals. Starting from $\mathfrak{S}_{0}^{h}$, we construct inductively sets $\mathfrak{S}_{n+1}^{h}$, $n \ge 0$, by looking through the mosaics in $\mathfrak{S}_{n}^{h}$ and eliminating the bad ones. We stop when $\mathfrak{S}_{n+1}^{h} = \mathfrak{S}_{n}^{h}$ and set $\mathfrak{S}_{\ast}^{h}= \mathfrak{S}_{n}^{h}$. By Lemma~\ref{lem:decidebad1}, $\mathfrak{S}_{\ast}^{h}$ can be constructed in double exponential time.

\mosaicsGMLnu*

\begin{proof}
$(i) \Rightarrow (ii)$ 
Let $\mathfrak{S}$ be the set of mosaics defined by the elements of the given $\Amf_{1}$ and $\Amf_{2}$, and let $h$ pick the nominal mosaics.  Then $\mathfrak{S}\subseteq \mathfrak{S}_{0}^{h}$ does not contain any bad mosaics, and so $\mathfrak{S}\subseteq \mathfrak{S}_{\ast}^{h}$. Moreover, $\varphi_{i}\in \type(\Amf_i,a_{i})\in M_{i}(\Amf_{i},a_{i})$, for $i=1,2$, as required.
	
$(ii) \Rightarrow (i)$ For any $M = (M_1,M_2)\in \mathfrak{S}_{\ast}^{h}$ and $R \in \sigma$ such that $\Diamond_{R}^{\geq k}\chi\in \type\in M_i$, for some $\Diamond_{R}^{\geq k}\chi$, $\type$ and $i \in \{1,2\}$, we take an $R$-witness $\mathcal{M}_{R,M}\subseteq \mathfrak{S}_{\ast}^{h}$ for $M$ with the corresponding relation $R^{M,\mathcal{M}_{R,M}}$ and define the required structures $\Amf_{i}$, $i=1,2$, as follows. The domain $\dom(\Amf_{i})$ of $\Amf_{i}$ consists of all words
\begin{equation}\label{eq1}
s=(\type_{0},M_{0},i_{0})R_{1}(\type_{1},M_{1},i_{1}) R_2 \dots R_{n}(\type_{n},M_{n},i_{n}) 
\end{equation}
such that $\omega > n\geq 0$, $M_{j}=(M_{j,1},M_{j,2})\in \mathfrak{S}_{\ast}^{h}$, for $j\leq n$, 
\begin{itemize}
\item[--] $\type_{j}\in M_{j,i}$ and $1 \le i_j \le f(M_{j,i})$, for all $j\leq n$; 

\item[--] $M_{j+1}\in \mathcal{M}_{R_{j+1},M_{j}}$ and no nominal occurs in $M_{j+1}$;
		
\item[--] all types in $M_{k}^{\ast}$ and $M_{0,k}$ are $\Du$-equivalent, $k=1,2$.
\end{itemize}
For a propositional variable $A$ and $s$ of the form~\eqref{eq1}, we set $s\in A^{\Amf_{i}}$ iff $A\in \type_{n}$. For a nominal $N_c$, we set $c^{\Amf_{i}} = s$ iff $n=0$, $N_c\in \type_{0}$ and $h(N_c,i) = M_0$. By the definitions of mosaics and functions $h$ and $f$, such an $s$ is unique and $f(M_0) = i_0 = 1$. 
Finally, for a binary $R \in \sigma$ and $s$ given by \eqref{eq1}, we set:
\begin{itemize}
\item[--] $(s',s)\in R^{\Amf_{i}}$  if  $s'$ is obtained by dropping $R_{n}(\type_{n},M_{n},i_{n})$ from $s$, $R=R_{n}$ and
$
(\type_{n-1},M_{n-1})R_{n}^{M_{n-1},\mathcal{M}_{R_n,M_{n-1}}}
(\type_{n},M_{n},i_{n});
$

\item[--] $(s,(\type,M,j))\in R^{\Amf_{i}}$ if $(\type,M,j)\in \dom(\Amf_{i})$, a nominal occurs in $M$ and $(\type_{n},M_{n})R^{M_{n},\mathcal{M}_{R,M_{n}}} (\type,M,j)$.
\end{itemize}
By induction on the construction of formulas we show that, for all $\chi\in \sub(\varphi_1,\varphi_2)$ and all $s \in \dom(\Amf_i)$, we have $\Amf_{i}\models \chi(s)$ iff $\chi \in \type_{n}$. The basis of induction follows immediately from the definition and the Boolean cases in the induction step are trivial. Consider $\Diamond^{\ge k}_R \chi$.

Let $s$ be of the form~\eqref{eq1}. If $\Amf_{i}\models \Diamond^{\ge k}_R \chi(s)$, then there exist distinct $s_1,\dots,s_k \in \dom(\Amf_i)$ with $s R^{\Amf_i} s_l$ and $\Amf_{i}\models \chi(s_l)$, for all $l = 1,\dots,k$. Let $\type_l$ be the type (first component) of the last triple in $s_l$.  By IH, $\chi \in \type_l$, and so $\Diamond^{\ge k}_R\chi \in \type_{n}$ by {\bf (witness)}. Conversely, suppose $\Diamond^{\ge k}_R\chi \in \type_{n}$ in $s$ of the form~\eqref{eq1}. As $s$ is not bad, $M_n$ has an $R$-witness, and so, by condition {\bf (witness)} and the definition of $R^{\Amf_i}$, there are distinct $s_1,\dots,s_k \in \dom(\Amf_i)$ with $s R^{\Amf_i} s_l$ and $\chi \in \type(s_l)$, where $\type_l$ is the type of the last triple in $s_l$ and $l = 1,\dots,k$. By IH, $\Amf_{i}\models \chi(s_l)$ for all such $l$, and so $\Amf_{i}\models \Diamond^{\ge k}_R \chi(s)$, as required.

Finally, the case $\Du\chi$ follows from IH, {\bf ($\Du$-harmony)} and the second item in the definition of bad mosaics. As a consequence, we obtain $\Amf_i\models \varphi_i(\type_i^{\ast},M^{\ast},1)$, for $i=1,2$. 
 	
Define a relation $\bis\subseteq \dom(\Amf_{1})\times\dom(\Amf_{2})$ by taking $s\bis s'$, for $s$ given by \eqref{eq1} and 
\begin{equation}\label{eq2}
s'=(\type_{0}',M_{0}',i_{0}')R_{1}'(\type_{1}',M_{1}',i_{1}') R'_2 \dots R_{m}'(\type_{m}',M_{m}',i_{m}'),
\end{equation}
if $n=m$, $R_{i}=R_{i}'$ and $M_{i}=M_{i}'$, for all $i\leq n$. It follows that we have $(\type_{1}^{\ast},M^{\ast},1) \bis (\type_{2}^{\ast},M^{\ast},1)$. Thus, it remains to show that $\bis$ is a $\MLnu(\varrho)$-bisimulation between $\Amf_{1}$ and $\Amf_{2}$. 
As both components of mosaics are non-empty, $\bis$ is global. That $c^{\Amf_1} \bis c^{\Amf_2}$ follows from the definition of $c^{\Amf_i}$ and the fact that $h(N_c,1)=h(N_c,2)$ for all $c\in \varrho$. The first item in the definition of mosaics  ensures $(i)$ in the definition of bisimulation. Finally, condition {\bf (bisim)} ensures items $(ii)$ and $(iii)$ in that definition. 
\end{proof}

This lemma in tandem with Lemma~\ref{lem:decidebad1} give a $2\ExpTime$ upper bound for deciding (Craig) $\GMLnu/\MLnu$-separation. We can easily modify the construction to establish the same upper bound for $\GMLnu/\MLn$, $\GMLnu/\MLu$, $\GMLu/\MLu$, and $\GMLu/\ML$:  
\begin{itemize}
\item[--] For $\GMLnu/\MLn$, we require that $\Amf_1,a_1 \sim_{\MLn(\varrho)} \Amf_2,a_2$ in Lemma~\ref{lem:main1}~$(i)$. To this end we admit mosaics $(M_{1},M_{2})$ in which $M_{1}$ or $M_{2}$ can be empty (but not both). Observe that now $h(N_c,1)=h(N_c,2)$ does not necessarily hold for all $c\in \varrho$ since the second component of $h(N_c,1)$ and the first component of $h(N_{c},2)$ can be empty, reflecting the fact that $\Amf_{1},c^{\Amf_{1}}$ is not necessarily $\MLn(\varrho)$-bisimilar to $\Amf_{2},c^{\Amf_{2}}$ if $c\in \varrho$ and  $\Amf_1,a_1 \sim_{\MLn(\varrho)} \Amf_2,a_2$.

\item[--] For $\GMLnu/\MLu$, we require that $\Amf_1,a_1 \sim_{\MLu(\varrho)} \Amf_2,a_2$ in Lemma~\ref{lem:main1}~$(i)$. To this end, given $\varphi_{1}$ and $\varphi_{2}$, we replace all $N_{c}$ in $\varphi_{2}$, for $c\in \sig(\varphi_{1})\cap \sig(\varphi_{2})$, by fresh $N_{c'}$ and apply the algorithm above to $\varphi_{1}$ and the resulting formula $\varphi_{2}'$.

\item[--] For $\GMLu/\MLu$, we omit everything related to nominals.

\item[--] For $\GMLu/\ML$, we omit everything related to nominals and allow mosaics with one empty component.
\end{itemize}

A somewhat different argument is required to establish the \textsc{co}\NExpTime upper bounds of Theorem~\ref{thm:gmlu} for $\GMLn/\MLn$- and $\GML/\ML$-separation in the $\Du$-free languages.   
First, we again ignore everything related to $\Du$ and admit mosaics $M = (M_{1},M_{2})$ with $M_{1}=\emptyset$ or $M_{2}=\emptyset$ (but not both). Second, we observe that, without $\Du$, we can significantly simplify the structures $\Amf_{i}$, $i=1,2$, from the proof of Lemma~\ref{lem:main1} $(ii) \Rightarrow (i)$. Let $m = \max \{\md(\varphi_1), \md(\varphi_2)\}$, where $\md(\varphi_i)$ is the \emph{modal depth} of $\varphi_i$, i.e., the maximal number of nested modal operators in $\varphi_i$. Let ${\Amf_{i}}_{|m}$ be the restrictions of the $\Amf_{i}$ to words $s$ of the form \eqref{eq1}  such that 
\begin{itemize}
\item[--] $n\leq m$ and no pair from $R^{\Amf_{i}}$ starts at $s$ of length $m$;

\item[--] in the first triple $(\type_0,M_0,i_0)$ of $s$, either $\type_0 \in \{\type^*_1, \type^*_2\}$ and $M_0 = M^*$ or $h(N_c,i) = M_0$, for some $c \in \sigma$.
\end{itemize}
Thus, the size of the ${\Amf_{i}}_{|m}$ is (singly) exponential in $|\varphi_{1}|+|\varphi_{2}|$. 
As well-known from modal logic~\cite{DBLP:books/daglib/0030819,DBLP:books/cu/BlackburnRV01}, ${\Amf_{i}}_{|m}\models \varphi_{i}(\type_{i}^{\ast},M^{\ast})$, for $i=1,2$. Moreover,
the restriction of $\bis$ to ${\Amf_{1}}_{|m}\times {\Amf_{2}}_{|m}$ is still an $\MLn(\varrho)$-bisimulation, and so ${\Amf_{1}}_{|m},(\type_{1}^{\ast},M^{\ast},1) \sim_{\MLn(\varrho)}{\Amf_{2}}_{|m},(\type_{2}^{\ast},M^{\ast},1)$. The existence of such structures ${\Amf_{i}}_{|m}$ can clearly be checked in \NExpTime. 
In the case of $\GML/\ML$-separation, we ignore everything related to nominals. 

\medskip

We now prove the matching 2$\ExpTime$ and \textsc{co}$\NExpTime$ lower bounds. By Lemma~\ref{l:red}, it suffices to do this for unary Craig separation. 
By Lemma~\ref{criterion}, we need to show 2$\ExpTime$-hardness of the following problem for $\LS\in \{\MLu,\ML\}$: given $\GMLu$-formulas $\varphi$ and $\psi$, decide  whether there exist pointed structures $\Amf,a$ and $\Bmf,b$ such that $\Amf\models \varphi(a)$, $\Bmf\models \psi(b)$, and $\Amf,a \sim_{\LS(\varrho)} \Bmf,b$, where $\varrho$ is the set of all binary predicates $R\in \sig(\varphi)\cup \sig(\psi)$ and all propositional variables $A\in \sig(\varphi)\cap\sig(\psi)$. This also implies 2$\ExpTime$-hardness for the relevant logics with nominals. 
To establish the \textsc{co}$\NExpTime$ lower bound, it suffices to show $\NExpTime$-hardness of the same problem for $\GML$-formulas $\varphi$, $\psi$ and $\LS = \ML$. 

Both proofs are similar to the lower bound proofs for Craig interpolant existence in description logics with nominals and role inclusions given in~\cite[Sections~7, 9]{DBLP:journals/tocl/ArtaleJMOW23}. The underlying idea is to construct formulas $\varphi_n$ and $\psi_n$, for $n < \omega$, in the required language and of polynomial size in $n$ such that if $\Amf\models \varphi_n(a)$, $\Bmf\models \psi_n(b)$, and $\Amf,a \sim_{\LS(\varrho)} \Bmf,b$, then there is a set $E\subseteq \dom(\Bmf)$
of size $2^{n}$ with the following properties:
\begin{itemize}
\item[--] all $e\in E$ are $\LS(\varrho)$-bisimilar, and
	
\item[--] the $e\in E$ satisfy all pairwise distinct combinations of variables $A_{0},\ldots,A_{n-1}$ (which cannot be in $\varrho$).
\end{itemize}
We show how to do this for $L=\GML$ and $\LS\in \{\MLu,\ML\}$. The formula 
$\varphi_n = \Diamond_{R}^{=1}\dots \Diamond_{R}^{=1}\top$ with $n$-many $\Diamond_{R}^{=1}$ is such that $\Amf \models \varphi_n(a)$ iff $\Amf$ has a unique $R$-chain of length $n$ starting from $a$. The formula $\psi_n$ is defined as the conjunction of
\begin{itemize}
\item[--] $\Box_{R}^{i}\Diamond_{R}^{=2}\top$, for $0\leq i <n$;

\item[--] $\Box_{R}^{i}(\Diamond_{R}A_{i} \wedge \Diamond_{R}\neg A_{i})$, for $0\leq i <n$;

\item[--] $\Box_{R}^{i}\big((A_{j} \rightarrow \Box_{R}A_{j}) \wedge (\neg A_{j} \rightarrow \Box_{R}\neg A_{j}) \big)$, for $0 \leq j < i <n$.
\end{itemize} 
We have $\Bmf \models \psi_n(b)$ iff the $R$-chains of depth $\le n$ starting from $b$ comprise a binary $R$-tree of depth $n$ with leaves satisfying counter values from $0$ to $2^{n}-1$ represented by the variables $A_{0},\dots,A_{n-1}$. So, if $\Amf\models \varphi_n(a)$, $\Bmf\models \psi_n(b)$, and $\Amf,a\sim_{\ML(\varrho)} \Bmf,b$, then there is a unique $a'\in \dom(\Amf)$ reachable via an $R$-chain of length $n$ from $a$, and there are leaves $E = \{e_{0},\dots,e_{2^{n}-1}\}$ of the binary tree with root $b$ in $\Bmf$ such that $e_{i}$ satisfies the binary encoding of $i$ by means of $A_{0},\dots,A_{n-1}$, and so 
$\Amf,{a'} \sim_{\LS(\varrho)}  \Bmf,e_{0} \sim_{\LS(\varrho)}\cdots \sim_{\LS(\varrho)}\Bmf,e_{2^{n}-1}$ by the definition of $\LS(\varrho)$-bisimulation and $R \in \varrho$.

For the 2\ExpTime-lower bound, the nodes in $E$ can then be used to encode
the acceptance problem for exponential space-bounded alternating Turing machines (ATMs). Roughly, the idea is to start from each $e\in E$ a simulation of the computation tree of the ATM for an input word of length $n$. On the shared symbols in $\varrho$, these trees are essentially isomorphic. Using the variables $A_{0},\dots,A_{n-1}$ as bits, one can create, for each $m<2^{n}$ on some $e\in E$, 
a $2^{n}$-yardstick starting at $m$ and coordinating the labels of the computation tree at positions $m+l 2^{n}$, $l\geq 0$. This is also achieved using counters encoded by variables that are not in $\varrho$. As one has to look `exponentially deep' into the computation tree, the universal modality $\Du$ of $\GMLu$ is required when writing up the encoding; see~\cite{DBLP:journals/tocl/ArtaleJMOW23} for details.

For the $\NExpTime$-lower bound, the nodes in $E$ can be used to encode
a $2^{n}\times 2^{n}$-torus tiling problem. In this case, assume for simplicity that $|E| = 2^{2n}$. Then, for each grid cell $(i,j)$ with $i,j<2^{n}$, we use some $e_{i,j}\in E$ satisfying the encoding of $i$ via $A_{0},\dots,A_{n-1}$ and that of $j$ via $A_{n},\dots,A_{2n-1}$.
Using successors of the nodes in $E$ (which coincide on $\varrho$, for each $e\in E$),
one can ensure the horizontal and vertical matching conditions of the tiling problem.  
\end{proof}

\section*{Proofs for Section~\ref{inverse-sep}}

\mosaicGMLiu*
\begin{proof}
We start by establishing the upper bound and focusing on Craig separation in view of Lemma~\ref{l:red}. Consider first $\GMLiu/\MLiu$. Let $\varphi_{1}$ and $\varphi_{2}$ be  some $\GMLiu$-formulas, $\sigma = \sig(\varphi_1) \cup \sig(\varphi_2)$ and $\varrho = \sig(\varphi_{1}) \cap \sig(\varphi_{2})$. We employ the notion of type introduced in Section~\ref{sec:gmlu} and adapted in the obvious way for $\GMLiu$. We need to check whether there are pointed $\sigma$-structures $\Amf_{1},a_{1}$ and $\Amf_{2},a_{2}$, and a $\MLiu(\varrho)$-bisimulation $\bis$ between $\Amf_{1}$ and $\Amf_{2}$ such that $\Amf_{1}\models\varphi_{1}(a_{1})$, $\Amf_{2}\models \varphi_{2}(a_{2})$ and $a_{1}\bis a_{2}$.

We construct $\Amf_{1},\Amf_{2}$ and $\bis$ (if they exist) from mosaics, which are appropriately modified to deal with inverse modalities. To this end, we supply each node with some information about the types of its `parents and children'\!. Following~\cite[Sec.\ 8.3]{Pratt23book}, we call the result a star-type. To give a precise definition, we require some notations. 

The \emph{neighbourhood} $N(\Amf,a)$ of a node $a\in \dom(\Amf)$ in a $\sigma$-structure $\Amf$ is the set
$$
N(\Amf,a) = \{ b\in \dom(\Amf)\mid (a,b)\in S^{\Amf},\ S\in \{R,R^{-}\},\ R\in \sigma\}.
$$
A \emph{successor type set} $\frak{s}$ (\emph{for $\varphi_{1}$ and $\varphi_{2}$}) is a set containing, for each type $\type$ and each $S$, exactly one expression of the form $\Diamond_{S}^{=n}\type$, for $0\leq n\leq \mx$, or $\Diamond_{S}^{>\mx}\type$. Let $a\in \dom(\Amf)$, $D\subseteq \dom(\Amf)$, and $n(a,S,\type, D) = |\{b\in D \mid (a,b)\in S^{\Amf},  \type=\type(\Amf,b)\}|$. The successor type set $\frak{s}(\Amf,a,D)$ defined by $a$ and $D$ contains 
\begin{itemize}
\item[--] $\Diamond_{S}^{=n}\type\in \frak{s}(\Amf,a,D)$ if $n = n(a,S,\type,D) \leq \mx$, and 

\item[--] $\Diamond_{S}^{> \mx}\type\in \mathfrak{s}(\Amf,a,D)$ if $n(a,S,\type,D) > \mx$.
\end{itemize}  
Observe that $\frak{s}(\Amf,a,D)=\frak{s}(\Amf,a,D\cap N(\Amf,a))$, so only nodes in the neighbourhood of $a$ are relevant for the definition of $\frak{s}(\Amf,a,D)$.
Occasionally, we are only interested in the existence of an $S$-successor in $D$ with a type $\type$, writing $\Diamond_{S}^{>0}\type\in \frak{s}$ as a shorthand for `$\Diamond_{S}^{=n}\type\in \frak{s}$ 
for some $n$, $1\leq n \leq \mx$, or $\Diamond_{S}^{> \mx}\type\in \mathfrak{s}$'\!. We omit $\Amf$ from $\frak{s}(\Amf,a,D)$, $\type(\Amf,a)$, $N(\Amf,a)$, etc., if understood. We also omit elements of the form $\Diamond^{=0}_S \type$ from successor type sets.


We define a \emph{star-type} as a triple of the form $\ftype = (\type,\frak{s}_{p},\frak{s}_{c})$, where $\type$ is a type,  $\frak{s}_{p}$ a \emph{parent} successor type set, and $\frak{s}_{c}$ a \emph{child} successor type set. To be coherent with the semantics of graded modalities, $\ftype$ should satisfy some natural conditions defined below. 
%
%
The \emph{$S$-profile $P_{S}(\frak{s})$ of $\frak{s}\in \{\frak{s}_{p},\frak{s}_{c}\}$} contains, for any $\chi\in \sub(\varphi_1,\varphi_2)$,
\begin{itemize}
\item[--] $\Diamond_{S}^{=n}\chi$ if $n$ is the sum over all $k$ with $\Diamond_{S}^{=k}\type' \in \frak{s}$ and $\chi\in \type'$;

\item[--] $\Diamond_{S}^{>\mx}\chi$ if this sum exceeds $\mx$ or $\Diamond_{S}^{>\mx} t'\in \frak{s}$, for some $\type' \ni \chi$.
\end{itemize}
The \emph{$S$-profile $P_{S}(\frak{s}_{p},\frak{s}_{c})$ of $\frak{s}_{p}$ and $\frak{s}_{c}$} contains, for all formulas $\chi\in \sub(\varphi_1,\varphi_2)$,
\begin{itemize}
\item[--] $\Diamond_{S}^{=n}\chi$ if $n=n_{1}+n_{2} \leq \mx$ for $\Diamond_{S}^{=n_{1}}\chi$ in the $S$-profile of $\frak{s}_{p}$ and $\Diamond_{S}^{=n_{2}}\chi$ in the $S$-profile of $\frak{s}_{c}$;

\item[--] $\Diamond_{S}^{>\mx}\chi$ otherwise.
\end{itemize}
A star-type $\ftype=(\type,\frak{s}_{p},\frak{s}_{c})$ is called \emph{coherent} if, for all formulas $\Diamond_{S}^{\geq n}\chi\in \sub(\varphi_1,\varphi_2)$, we have $\Diamond_{S}^{\geq n}\chi\in \type$ iff $P_{S}(\frak{s}_{p},\frak{s}_{c})$ contains either $\Diamond_{S}^{=n'}\chi$, for
some $n'\geq n$, or $\Diamond_{S}^{>\mx}\chi$. From now on, all star-types we consider are assumed to be coherent. 

The \emph{star-type realised by $a\in \dom(\Amf)$ and $D\subseteq \dom(\Amf)$} is defined as
$$
\ftype(a,D) = \big( \type(a),\frak{s}(a,D),
\frak{s}(a,\dom(\Amf)\setminus D) \big).
$$
Instead of $\dom(\Amf)\setminus D$, we can take $N(a) \setminus D$.

%


%

A \emph{mosaic} is a pair $M=(M_{1},M_{2})$ of sets of star-types such that
\begin{itemize}
\item[--] $A\in \type$ iff $A\in \type'$, for all $(\type,\frak{s}_{p},\frak{s}_{c}), (\type',\frak{s}_{p}',\frak{s}_{c}')\in M_{1}\cup M_{2}$ and $A\in \varrho$;
	
\item[--] for any $(\type,\frak{s}_{p},\frak{s}_{c})\in M_{1}\cup M_{2}$, all types in $\frak{s}_{p}\cup \frak{s}_{c}$ are $\Du$-equivalent to $\type$.
\end{itemize} 
Given structures $\Amf_{1},\Amf_{2}$ with $a\in \dom(\Amf_{j_{1}})$ and $p\in \dom(\Amf_{j_{2}})$, for some $j_1,j_2 \in \{1,2\}$, the \emph{mosaic defined by $a$ and $p$} is 
$$
M(\Amf_{1},\Amf_{2},a,p)= \big( M_{1}(\Amf_{1},a,p),M_{2}(\Amf_{2},a,p) \big),
$$
where, for $ D_{i}=\{ d\in \dom(\Amf_{i}) \mid \Amf_{i},d\sim_{\MLiu(\varrho)}\Amf_{j_{2}},p\}$ and $i=1,2$,
\begin{multline*}
M_{i}(\Amf_{i},a,p) = \{ \ftype(\Amf_{i},e,D_{i})\mid {}\\
e\in \dom(\Amf_{i}), \Amf_{i},e \sim_{\MLiu(\varrho)} \Amf_{j_{1}},a\}.
\end{multline*}
%
%
We are also interested in \emph{root mosaics} $M$ with $\Diamond_{S}^{=0}\type'\in \frak{s}_{p}$, for all $(\type,\frak{s}_{p},\frak{s}_{c})\in M_{1}\cup M_{2}$, all $\type'$ and all $S$. We obtain a root mosaic $M(\Amf_{1},\Amf_{2},a,\emptyset) = (M_1(\Amf_{1},a,\emptyset),M_2(\Amf_{2},a,\emptyset))$ with $a\in \dom(\Amf_{j})$ and $j \in \{1,2\}$ by taking, for $i = 1,2$,
$$
M_{i}(\Amf_{i},a,\emptyset) = \{ \ftype(\Amf_{i},e,\emptyset)\mid e\in
\dom(\Amf_{i}), \Amf_{i},e \sim_{\MLiu(\varrho)}\Amf_{j},a\}. 
$$




As in Section~\ref{sec:gmlu}, we are going to use mosaics for building  the structures $\Amf_{1}$ and $\Amf_{2}$ we are after. To this end, we require the following definition. Let $M=(M_{1},M_{2})$ be a mosaic and $\mathcal{M}$ a set of mosaics. Set $\Delta=\Delta_{1}\cup \Delta_{2}$ and $\Gamma=\Gamma_{1} \cup \Gamma_{2}$, where, for $i=1,2$,
\begin{align*}
& \Delta_{i} = \{(\ftype,j)\mid j<\omega ,\ \ftype\in M_{i}\},\\ 
&\Gamma_{i} = \{(\ftype',j',M')\mid j<\omega,\ \ftype'\in M_{i}',\ M'=(M_{1}',M_{2}')\in \mathcal{M}\}.
\end{align*}
Intuitively, $\Delta_i$ consists of $\omega$-many copies of each star-type in $M_i$, and $\Gamma_i$ of $\omega$-many copies of each `pointed mosaic' $(\ftype',M')$ in $\mathcal{M}$.

We call $\mathcal{M}$ a \emph{witness} for $M$ if, for each $S$, there is a relation
$$
S^{M,\mathcal{M}} \subseteq (\Delta_{1} \times \Gamma_{1})\cup (\Delta_{2} \cup \Gamma_{2})
$$
such that the following conditions are satisfied:
\begin{description}
\item[(witness$_1$)] if $(\ftype,j)\in \Delta$, $\ftype=(\type,\frak{s}_{p},\frak{s}_{c})$ and $\Diamond_{S}^{=n}\type'\in \frak{s}_{c}$, then $(\ftype,j)$ has exactly $n$ different $S^{M,\mathcal{M}}$-successors of type $\type'$---i.e., of the form $((\type',\frak{s}_{p}',\frak{s}_{c}'),j',M')$---in $\Gamma$; if $\Diamond_{S}^{>\mx}\type'\in \frak{s}_{c}$, then $(\ftype,j)$ has at least $\mx+1$ different $S^{M,\mathcal{M}}$-successors of type $\type'$ in $\Gamma$;


\item[(witness$_2$)] if $(\ftype',k,M')\in \Gamma$, $\ftype'= (\type',\frak{s}_{p}',\frak{s}_{c}')$ and $\Diamond_{S}^{=n}\type\in \frak{s}_{p}'$, then $(\ftype',k,M')$ has exactly $n$ $S^{M,\mathcal{M}}$-successors of type $\type$ in $\Delta$; if $\Diamond_{S}^{>\mx}\type\in \frak{s}_{p}'$, then $(\ftype',k,M')$ has $> \mx$  $S^{M,\mathcal{M}}$-successors of type $\type$ in $\Delta$;


\item[(bisim$_1$)] if $R\in \varrho$, $S \in \{R,R^-\}$ and some element of $\Delta$ has an $S^{M,\mathcal{M}}$-successor $(\ftype',j',M') \in \Gamma$, then every element of $\Delta$ has an $S^{M,\mathcal{M}}$-successor of the form $(\ftype'',j'',M') \in \Gamma$;

\item[(bisim$_2$)] if $R\in \varrho$, $S \in \{R,R^-\}$ and some $(\ftype',j',M') \in \Gamma$ is an $S^{M,\mathcal{M}}$-successor of an element of $\Delta$, then all $(\ftype'',j'',M')$ in $\Gamma$ are  $S^{M,\mathcal{M}}$-successors of some element of $\Delta$.
\end{description}

%
%
Our aim is to show that there exist pointed structures $\Amf_{i},a_{i}$, $i=1,2$,  satisfying $\varphi_{i}$ with $\Amf_{1},a_{1}\sim_{\MLiu(\varrho)}\Amf_{2},a_{2}$ iff there exists a set $\mathfrak{S}$ of mosaics such that
\begin{description}
\item[(m1)] if $M=(M_{1},M_{2})\in \mathfrak{S}$, then $M_{1}\not=\emptyset$ and $M_{2}\not=\emptyset$;

\item[(m2)] there are a root mosaic $M=(M_{1},M_{2})\in \mathfrak{S}$ and star-types $\ftype_{i}\in M_{i}$ containing $\varphi_{i}$ in their types, for $i=1,2$;

\item[(m3)] every $M\in \mathfrak{S}$ has a witness $\mathcal{M}\subseteq \mathfrak{S}$;

\item[(m4)] if $\Du\chi\in \type$ for some $(\type,\frak{s}_{p},\frak{s}_{c})\in M_{i}$ and $M=(M_{1},M_{2})\in \mathfrak{S}$, then there exists a root mosaic $M'=(M_{1}',M_{2}')\in \mathfrak{S}$ such that $\chi\in \type'$ for some $(\type',\frak{s}_{p}',\frak{s}_{c}')\in M_{i}'$ such that $\type$ and $\type'$ are $\Du$-equivalent, for $i=1,2$.
\end{description}
Moreover, we also aim to show that it suffices to look for $\mathfrak{S}$ with mosaics and witnesses of exponential size, which would make it possible to check the existence of such an $\mathfrak{S}$ in double exponential size. With this in mind, we first give a somewhat different, but equivalent `syntactic' definition of witnesses.

We say that star-types $\ftype=(\type,\frak{s}_{p},\frak{s}_{c})$ and $\ftype'= (\type',\frak{s}_{p}',\frak{s}_{c}')$ are \emph{$S$-coherent} and write $\ftype \rightsquigarrow_{S} \ftype'$ if $\Diamond_{S}^{>0}\type'\in \frak{s}_{c}$ and $\Diamond_{S^{-}}^{>0}\type\in \frak{s}'_{p}$.
%
Let $M=(M_{1},M_{2})$ be a mosaic, $\mathcal{M}$ a set of mosaics and 
$$
\mathcal{M}_{i}^{\bullet}= \{(\ftype',M')\mid M'=(M_{1}',M_{2}')\in \mathcal{M}, \ftype'\in M_{i}'\}, \quad i=1,2. 
$$
Call $\mathcal{M}$ is a \emph{syntactic witness} for $M$ if, for every $S$, there is a relation 
$$
\S^{M,\mathcal{M}} \subseteq (M_{1}\times \mathcal{M}_{1}^{\bullet}) \cup (M_{2}\times \mathcal{M}_{2}^{\bullet})
$$
satisfying the following conditions, for $i=1,2$:
\begin{description}
\item[(coherence)] if $(\ftype, (\ftype',M'))\in \S^{M,\mathcal{M}}$, then $\ftype \rightsquigarrow_{S} \ftype'$;

\item[(bisim)] if $R\in \varrho$, $S \in \{R,R^-\}$, and $M'=(M_{1}',M_{2}')\in \mathcal{M}$, then either
\begin{itemize}
\item[--] for any $\ftype\in M_{i}$, there is $\ftype'\in M_{i}'$ with $(\ftype,(\ftype',M'))\in \S^{M,\mathcal{M}}$ and, for any $\ftype'\in M_{i}'$, there is $\ftype\in M_{i}$ with $(\ftype,(\ftype',M'))\in \S^{M,\mathcal{M}}$,
\end{itemize}
in which case $M'$ is called an \emph{$S$-successor of $M$},  
or 
\begin{itemize}
\item[--] $(\ftype,(\ftype',M'))\notin \S^{M,\mathcal{M}}$ for any $\ftype\in M_{i}$ and $\ftype'\in M_{i}'$;
\end{itemize}

\item[(wit$_1$)] if $\ftype=(\type,\frak{s}_{p},\frak{s}_{c}) \in M_{i}$ and $\Diamond_{S}^{>0}\type'\in \frak{s}_{c}$, then there are $M'=(M_{1}',M_{2}')\in \mathcal{M}$ and 
$\ftype'=(\type',\frak{s}_{p}',\frak{s}_{c}')\in M_{i}'$ such that 
$(\ftype,(\ftype',M'))\in \S^{M,\mathcal{M}}$;

\item[(wit$_2$)] if $M'=(M_{1}',M_{2}')\in \mathcal{M}$, $\ftype'=(\type',\frak{s}_{p}',\frak{s}_{c}')\in M_{i}'$ and $\Diamond_{S^{-}}^{>0}\type\in \frak{s}_{p}'$, then there is $\ftype=(\type,\frak{s}_{p},\frak{s}_{c})\in M_{i}$ with $(\ftype,(\ftype',M'))\in \S^{M,\mathcal{M}}$;

\item[(wit$_3$)] if $R\in \varrho$, $S \in \{R,R^-\}$ and $(\ftype,(\ftype',M'))\in \S^{M,\mathcal{M}}$, where $\ftype=(\type,\frak{s}_{p},\frak{s}_{c})\in M_{i}$, then there is a function $g$ that assigns to each $S$-successor $M^*=(M_{1}^*,M_{2}^*)$ of $M$ in $\mathcal{M}$ a star-type in $M_{i}^*$ such that $g(M')= \ftype'$ and, for any $\Diamond_{S}^{=n}\type_1\in \frak{s}_{c}$, we have 
$$
|\{M^*\in \mathcal{M}\mid g(M^*)=(\type_1,\frak{s}_{p}^*,\frak{s}_{c}^*), \text{ for some } \frak{s}_{p}^*,\frak{s}_{c}^*\}|\leq n. 
$$
\end{description}

The intuition behind the concept of a syntactic witness is as follows. If $\mathcal{M}$ is a witness for $M$, then for every $\ftype = (\type,\frak{s}_p,\frak{s}_c)\in M_1\cup M_2$, if $\Diamond^{=n}_S\type'\in \frak{s}_c$ then we require $(\ftype,j)$ to be $S^{M,\mathcal{M}}$-related to exactly $n$ different star-types based on type $\type'$. A similar constraint holds for $(\ftype',\frak{s}_p',\frak{s}_c')\in M_1'\cup M_2'$, where $(M_1',M_2')\in \mathcal{M}$. For a syntactic witness, we relax this requirement: if $\Diamond^{=n}_S\type'\in \frak{s}_c$ for some $n>0$ we do require $\ftype$ to be $\S^{M,\mathcal{M}}$-related to at least one star-type that is based on $\type'$, but we do not require there to be exactly $n$. The reason for this relaxation is that we can create multiple copies of each star-type: if, say, $\Diamond^{=2}_S\type'\in \frak{s}_c$ and $\ftype$ is $\S^{M,\mathcal{M}}$-related to exactly one $(\ftype',M')$, where $\ftype'$ is based on type $\type'$, then every $(\ftype,j)$ will be $S^{M,\mathcal{M}}$-related to $(\ftype',j_1',M')$ and $(\ftype,j_2',M')$ with $j_1'\neq j_2'$.
	
	Similarly, if $\Diamond^{=1}_S\type'\in \frak{s}_c$ and $\ftype$ is $\S^{M,\mathcal{M}}$-related to both $(\ftype',M')$ and $(\ftype'',M'')$, where $\ftype'$ and $\ftype''$ are based on type $\type'$, then there will be different $j_1$ and $j_2$ such that $(\ftype,j_1)$ is $S^{M,\mathcal{M}}$-related to $(\ftype',j',M')$ and $(\ftype,j_2)$ is $s^{M,\mathcal{M}}$-related to $(\ftype'',j'',M'')$.
	
	Such copying does, of course, require there to be at least one $\S^{M,\mathcal{M}}$-successor when $\Diamond^{>0}_S\type'\in \frak{s}_c$. This is enforced by \textbf{(wit$_1$)}, while \textbf{(wit$_2$)} is the corresponding requirement for $\Diamond^{>0}_S\type\in \frak{s}_p'$. 
	
	There is one other requirement, however, that is slightly harder to explain. Suppose that $\Diamond^{=1}\type'\in \frak{s}_c$ and that $(\ftype,(\ftype_1',M')), (\ftype, (\ftype_2',M''))\in \S^{M,\mathcal{M}}$, where $\ftype_1'$ and $\ftype_2'$ are both based on $\type'$ and $M'\neq M''$. Suppose furthermore that $\ftype_1'$ and $\ftype_2'$ are the only $\S^{M,\mathcal{M}}$-successors of $\ftype$ in their respective mosaics. If we were to try to create a witness $S^{M,\mathcal{M}}$ from this $\S^{M,\mathcal{M}}$, we would take $((\ftype,j),(\ftype_1',j_1',M'))\in S^{M,\mathcal{M}}$ and $((\ftype,j),(\ftype_2',j_2',M''))\in S^{M,\mathcal{M}}$. But there would be no $\ftype_2''$ and $j_2''$ such that $((\ftype,j),(\ftype_2'',j_2'',M''))\in S^{M,\mathcal{M}}$ and no $\ftype_1''$ and $j_1''$ such that $((\ftype,j),(\ftype_1'',j_1'',M'))\in S^{M,\mathcal{M}}$. Therefore, the \textbf{(bisim$_1$)} condition would not be satisfied.
	
	In order to prevent this kind of situation, we must ensure that if $(\ftype,(\ftype',M'))\in \S^{M,\mathcal{M}}$, then there is some way to choose $\ftype',M'$ as a successor of $\ftype$ that respects both graded modalities and bisimulation restrictions. This means it must be possible to choose one $\ftype''$ per $S$-successor $M''$ of $M$ such that (i) for the mosaic $M'$ we choose $\ftype'$ and (ii) if $\Diamond^{=n}_S\type''\in \frak{s}_c$ then we choose at most $n$ star-types based on type $\type''$. In our definition of a syntactic witness, this requirement is included as \textbf{(wit$_3$)}.
	
	Now that we have explained the intuition behind syntactic witnesses, we can show that syntactic witnesses are witnesses, and vice versa.

\begin{lemma}\label{lem:syntactic witness1} 
$\mathcal{M}$ is a witness for $M$ iff $\mathcal{M}$ is a syntactic witness for $M$.
\end{lemma}
\begin{proof}
	$(\Rightarrow)$ Suppose $\mathcal{M}$ is a witness for $M=(M_1,M_2)$ and, for each $S$, let $S^{M,\mathcal{M}}$ be the witnessing relation with respect to $S$. Let $\S^{M,\mathcal{M}}$ be the relation given by $(\ftype,(\ftype',M')\in \S^{M,\mathcal{M}}$ iff there are $j,j'$ such that $((\ftype,j),(\ftype',j',N'))\in S^{M,\mathcal{M}}$. 
	
	We show that $\S^{M,\mathcal{M}}$ satisfies the conditions for $\mathcal{M}$ being a syntactic witness for $M$.
	\begin{description}
		\item[(coherence)] Suppose that $(\ftype,(\ftype',M'))\in \S^{M,\mathcal{M}}$, where $\ftype=(\type,\frak{s}_p,\frak{s}_c)$ and $\ftype' = (\type',\frak{s}_p',\frak{s}_c')$. Then, by the definition of $\S^{M,\mathcal{M}}$, there are $j,j'$ such that $((\ftype,j),(\ftype',j',M'))\in S^{M,\mathcal{M}}$.
		
		By the \textbf{(witness$_1$)} condition, $((\ftype,j),(\ftype',j',M'))\in S^{M,\mathcal{M}}$ implies that $\Diamond^{=0}\type'\not\in \frak{s}_c$, and hence that $\Diamond^{>0}\type'\in \frak{s}_c$. Similarly, by \textbf{(witness$_2$)}, it implies that $\Diamond^{=0}_{S^-}\type\not\in \frak{s}'_p$, and hence $\Diamond^{>0}_{S^-}\type\in \frak{s}'_p$. This shows that $\ftype \rightsquigarrow_S \ftype'$, so the \textbf{(coherence)} condition is satisfied.
		
		\item[(bisim)] Consider any $M'=(M_1',M_2')\in \mathcal{M}$ and any $i\in \{1,2\}$. Assume $R\in \varrho$ and $S \in \{R,R^-\}$. We distinguish two cases. As a first case, suppose that there are $\ftype_1\in M_i$, $\ftype_1'\in M_i'$ and $j_1,j_1'<\omega$ such that $((\ftype_1,j_1),(\ftype_1',j_1',M_i')\in S^{M,\mathcal{M}}$. Then  \textbf{(bisim$_1$)} implies that for every $\ftype_2\in M_i$ and every $j_2<\omega$ there are $\ftype_2'\in M_i'$ and $j_2'<\omega$ such that $((\ftype_2,j_2),(\ftype_2',j_2',M')\in S^{M,\mathcal{M}}$. 
		
		So, for every such $\ftype_2$, there is a $\ftype_2'$ such that $(\ftype_2,(\ftype_2,M'))\in \S^{M,\mathcal{M}}$.
		Similarly, \textbf{(bisim$_2$)} implies that for every $\ftype_2'\in M_i'$ there is a $\ftype_2\in M_i$ such that $(\ftype_2,(\ftype_2,M'))\in \S^{M,\mathcal{M}}$. Hence, in this case $M'$ is an $S$-successor of $M$.
		
		As a second case, suppose there are no $\ftype_1\in M_i$, $\ftype_1'\in M_i'$ and $j_1,j_1'<\omega$ such that $((\ftype_1,j_1),(\ftype_1',j',M'))\in S^{M,\mathcal{M}}$. Then for every such $\ftype,\ftype'$ we have $((\ftype,(\ftype',M'))\not\in\S^{M,\mathcal{M}}$.
		
		These cases are exhaustive, so the \textbf{(bisim)} condition is satisfied.
		\item[(wit$_1$)] Suppose $\ftype = (\type,\frak{s}_p,\frak{s}_c)\in M_i$ and $\Diamond^0_S\type'\in \frak{s}_c$. Then, because $S^{M,\mathcal{M})}$ satisfies \textbf{(witness$_1$)}, every $(\ftype,j)$ has $>0$ $S^{M,\mathcal{M}}$-successors $(\ftype',j',M')$. Therefore, $(\ftype',M')$ is an $\S^{M,\mathcal{M}}$ successor of $\ftype$. Hence \textbf{(wit$_1$)} is satisfied.
		\item[(wit$_2$)] As with \textbf{(wit$_1$)}, except here we use the fact that $S^{M,\mathcal{M}}$ satisfies \textbf{(witness$_2$)}.
		\item[(wit$_3$)] Assume $R\in \varrho$ and $S \in \{R,R^-\}$. 
		Take any $(\ftype,(\ftype',M'))\in \S^{M,\mathcal{M}}$. Let $\ftype = (\type,\frak{s}_p,\frak{s}_c)\in M_i$. Fix any $j,j'<\omega$ such that $((\ftype,j),(\ftype',j',M'))\in S^{M,\mathcal{M}}$. For every $M^\ast=(M_1^\ast,M_2^\ast)\in \mathcal{M}$, we say that $M^\ast$ is an $S$-successor mosaic of $(\ftype,j)$ if there are $\ftype^\ast,j^\ast$ such that $((\ftype,j),(\ftype^\ast,j^\ast,M^\ast))\in S^{M,\mathcal{M}}$.
		
		Let $g$ be any function that (1) assigns to every $S$-successor mosaic $M^\ast$ of $(\ftype,j)$ an $\ftype^\ast$ such that there is a $j^\ast$ with $((\ftype,j),(\ftype^\ast,j^\ast,M^\ast))\in S^{M,\mathcal{M}}$ and (2) satisfies $g(M')=\ftype'$.
		
		As $S^{M,\mathcal{M}}$ satisfies \textbf{(witness$_1$)}, we have that if $\Diamond_S^{=n}\type'\in \frak{s}_C$ then there are at most $n$ mosaics $M^\ast$ such that $g(M^\ast)=\ftype^\ast = (t^\ast,\frak{s}^\ast_p,\frak{s}^\ast_c)$ for some $\frak{s}^\ast_p,\frak{s}^\ast_c$. Furthermore, because of the \textbf{(bisim$_1$)} condition, the set of $S$-successor mosaics does not depend on the choice of $j$.
		
		It follows that $g$ is a function that satisfies the requirements of \textbf{(wit$_3$)}, so that property is satisfied for $\S^{M,\mathcal{M}}$.
	\end{description}
	We have now shown that $\S^{M,\mathcal{M}}$ satisfies all of the conditions for $\mathcal{M}$ to be a syntactic witness of $M$.
	
	$(\Leftarrow)$ Suppose $\mathcal{M}$ is a syntactic witness for $M$ and the elements of $\Gamma\cup \Delta$ are enumerated. For $R\in \sigma$ and $S\in \{R,R^{-}\}$, let $\S^{M,\mathcal{M}}$ be the witnessing relation. We will create the relation $S^{M,\mathcal{M}}$ by adding $S$- and $S^-$-edges to elements $x\in \Delta\cup\Gamma$ in the enumeration order. Throughout the process, we keep invariant the property that, by the time we reach $x\in \Delta \cup \Gamma$, there is at most one $S$- or $S^{-1}$-edge to or from $x$ already. Four cases are possible. 
	
	\emph{Case} 1: $x=(\ftype,j)\in \Delta$, $\ftype=(\type,\frak{s}_{p},\frak{s}_{c})$, and no edges from $x$ have been added yet. First, suppose $R\in \sigma\setminus \varrho$. For every type $\type'$, if $\Diamond^{=n}\type'\in \frak{s}_{c}$ for some $n>0$, we fix any $M'$ and $\ftype'=(\type',\frak{s}_p',\frak{s}_c')$ such that $(\ftype,(\ftype',M'))\in \S^{M,\mathcal{M}}$. (At least one such $(\ftype',M')$ exists due to \textbf{(wit$_1$)}.) Now add $S$-edges from $x$ to the first $n$ elements of the form $(\ftype',k,M')\in \Gamma$ that do not yet have any edges to it. Similarly, if $\Diamond^{>\mx}\type'\in \frak{s}_c$, then add $S$-edges from $x$ to the lowest (in the enumeration) $\mx+1$ elements $(\ftype',k,M')\in \Gamma$ that do not yet have edges to it.
	This construction ensures that \textbf{(witness$_1$)} will be satisfied for $S^{M,\mathcal{M}}$.
	
	Now, suppose $R\in \varrho$. In this case, our construction has to be more complex, since we need to satisfy not only \textbf{(witness$_1$)} but also \textbf{(bisim$_1$)}. Let $\mathcal{M}_S\subseteq \mathcal{M}$ be the set of $S$-successors of $M$ in the sense defined in \textbf{(bisim)}. This means that we must ensure that, for every $M'\in \mathcal{M}_S$, there is at least one $(\ftype',j',M')\in \Gamma$ such that $(x,(\ftype',j',M'))\in S^{M,\mathcal{M}}$, while obeying the number restrictions. To this purpose, fix any $(\ftype',M')$ such that $(\ftype,(\ftype',M'))\in \S^{M,\mathcal{M}}$. Let $g$ be the function given by \textbf{(wit$_3$)} and let $\mathcal{E} = \{(\ftype'',M'') \mid g(M'') = \ftype''\}$.
	
	Consider any $\type''$ such that $\Diamond^{=n}_S \type'' \in \frak{s}_{c}$. It is a property of $g$ that there are at most $n$ elements in $\mathcal{E}$ that have a $\type''$ component. If there are less than $n$ such elements, take any $(\ftype'',M'')$ such that $(\ftype,(\ftype'',M''))\in \S^{M,\mathcal{M}}$, and add enough copies of this $(\ftype'',M'')$ to $\mathcal{E}$ to make the total number of elements with a $\type''$ component exactly $n$. We turn $\mathcal{E}$ into a multi-set by doing this. (Such $(\ftype'',M'')$ must exist by \textbf{(wit$_1$)}.)
	
	Similarly, if $\Diamond^{> \mx}_S \type'' \in \frak{s}_{c}$ and there are $\leq \mx$ elements in $\mathcal{E}$ with a $\type''$ component, add copies of $(\ftype'',M'')$ until we have more than $\mx$ such elements.
	
	Now, for every $(\ftype',M')\in \mathcal{E}$, add an $S$-edge from $x = (\ftype,j)$ to the lowest element $(\ftype',k,M')\in \Gamma$ that does not yet have any edges to it. This construction ensures that \textbf{(witness$_1$)} and \textbf{(bisim$_1$)} will be satisfied for $S^{M,\mathcal{M}}$.

%
	
	\emph{Case} 2: $x=(\ftype',k,M')\in \Gamma$ and no edges to $x$ have been added yet. For every $S$, we do the following. If $\Diamond^{=n}_{S^-} \type \in \frak{s}_{p}$ or $\Diamond^{> \mx}_{S^-} \type \in \frak{s}_{p}$, let $\ftype$ be such that $(\ftype, (\ftype',M'))\in \S^{M,\mathcal{M}}$. Create $S$ edges to $x$ from the first $n$ (if $\Diamond^{=n}_{S^-} \type \in \frak{s}_{p}$) or the first $\mx + 1$ (if $\Diamond^{> \mx}_{S^-} \type \in \frak{s}_{p}$) elements of the form $(\ftype,j)\in \Delta$ that do not yet have any outgoing edges.
	
	\emph{Case} 3: Suppose $x = (\ftype, j) \in \Delta$ and there is already an $S$ edge from $x$ to $(\ftype',k,M')\in \Gamma$. Add edges for $S^{-1}$ as in case 1. With regard to $S$, add edges like in case 1 except that instead of taking any $S$-successor we now take the successor $(\ftype',M')$ as a starting point, and take the already existing $S$ edge from $(\ftype,j)$ to $(\ftype',k,M')$ instead of one of the edges that we would otherwise add.
	
	\emph{Case}	4: Suppose $x = (\ftype',k,M')\in \Gamma$ and there is already an $S$ edge from $(\ftype,j)$ to $x$. Add edges for $S^{-1}$ as in case 2. For $S$, do the same except we take the edge from $(\ftype,j)$ to $x$ instead of one edge that would otherwise be added.
	
	It is immediate from the construction that conditions  \textbf{(witness$_1$)}, \textbf{(witness$_2$)} and \textbf{(bisim$_1$)} are satisfied;  condition \textbf{(bisim$_2$)} follows from $\mathcal{M}$ being a syntactic witness.
\end{proof}

Witnesses for mosaics defined by (disjoint) structures $\Amf_{1}$ and  $\Amf_{2}$ can be  extracted as follows. If $M=M(\Amf_{1},\Amf_{2},a,\emptyset)$ is a root mosaic, then 
$$
\mathcal{M}=\{ M(\Amf_{1},\Amf_{2},b,a) \mid b\in N(a'), \text{ for some } a'\sim_{\MLiu(\varrho)} a\}
$$
is a witness for $M$. If $M=M(\Amf_{1},\Amf_{2},a,p)$, for $a,p\in \dom(\Amf_{i})$, then we obtain a witness $\mathcal{M}$ for $M$ by taking 
\begin{multline*}
	\mathcal{M}=\{ M(\Amf_{1},\Amf_{2},b,a) \mid \Amf_{i},  b\not\sim_{\MLiu(\varrho)}\Amf_{i},p,\  b \in N(a'),\\ \text{ for some } a'\sim_{\MLiu(\varrho)} a\}.
\end{multline*}
Alternative and possibly larger witnesses can be obtained by dropping the restrictions `$b \in N(a')$, for some $a'\sim_{\MLiu(\varrho)} a$'\!.
In both cases, we call $\mathcal{M}$ a \emph{witness for $M$ given by $\Amf_{1},\Amf_{2}$}.

It is rather straightforward to see now that there exist $\Amf_{i},a_i$, $i=1,2$, such that $\Amf_{i} \models\varphi_{i}(a_{i})$ and $\Amf_{1},a_{1}\sim_{\MLiu(\varrho)}\Amf_{2},a_{2}$ iff there is a set $\mathfrak{S}$ of mosaics satisfying {\bf (m1)}--{\bf (m4)}. 
%
%
%
%
%
So we focus on showing that it is enough to use mosaics and witnesses of exponential size. 

Call successor sets $\frak{s}$ and $\frak{s}'$ \emph{profile-equivalent} if they have the same $S$-profiles for all $S\in \{R,R^{-}\}$ with $R\in \sigma$. Star-types $\ftype=(\type,\frak{s}_{p},\frak{s}_{c})$ and $\ftype'=(\type',\frak{s}_{p}',\frak{s}_{c}')$ are called \emph{parent profile-equivalent} if $\type=\type'$ and $\frak{s}_{p}$, $\frak{s}_{p}'$ are profile equivalent, and \emph{child profile-equivalent} if $\type=\type'$ and $\frak{s}_{c}$, $\frak{s}_{c}'$ are profile equivalent.
Mosaics $M=(M_{1},M_{2})$ and $M'=(M_{1}',M_{2}')$ are called  \emph{child profile-equivalent} if, for any $\ftype \in M_{i}$, there exists a child-profile equivalent $\ftype'\in M_{i}'$, and vice versa, for $i=1,2$; $M$ and $M'$ are called \emph{parent profile-equivalent} if, for any $\ftype \in M_{i}$, there exists a parent-profile equivalent $\ftype'\in M_{i}'$, and vice versa, for $i=1,2$.
Observe that the number of child/parent profile non-equivalent star-types does not exceed $(|\varphi_1|+|\varphi_2|)^{4(|\varphi_1|+|\varphi_2|)}$.

\expcardinality*

\begin{proof}
We proceed in three steps. First, for every mosaic $M'\in \mathcal{M}$, we reduce the number of star-types in $M'$ to an at most exponential amount. Second, we reduce the number of mosaics in $\mathcal{M}$. Finally, we reduce the number of star-types in $M$. After each of these steps, we show that the result is still a syntactic witness.

	\textit{Step} 1. 
	For each relation $S$, let $\S^{M,\mathcal{M}}$ be the witnessing relation for $M$ and $\mathcal{M}$, and consider any $M'=(M_1',M_2')\in \mathcal{M}$. There can be arbitrarily many star-types in $M_1'\cup M_2'$. We want to select subsets $N_1'\subseteq M_1'$ and $N_2'\subseteq M_2'$ of at most exponential size. Importantly, however, we want to do so in such a way that $(i)$ $N'=(N_1',N_2')$ is child profile-equivalent to $M'$ and $(ii)$ every edge $(\ftype,(\ftype',M'))\in \S^{M,\mathcal{M}}$, where $\ftype'$ is not included in $N'$, can be redirected to some $(\ftype,(\ftype'',N'))$, while still satisfying the conditions of a witnessing relation.
	
	The first of these requirements will be important when we eventually use our collection of mosaics to build a model; because of this child profile-equivalence we will be able to replace $M'$ with $N'$ in a model construction. The second requirement means that the resulting set of mosaics remains a syntactic witness for $M$.
	
	Formally, for any $M'=(M_1,M_2)\in \mathcal{M}$, let $N'=(N_1',N_2')$ be any mosaic such that $(i)$ $N_i'\subseteq M_i'$ for $i\in \{1,2\}$, $(ii)$ $M'$ and $N'$ are child profile-equivalent,
	$(iii)$ for every $\type$, $\type'$ and $S$, if there is a $\ftype' = (\type',\frak{s}'_p,\frak{s}'_c)\in M_i'$ such that $\Diamond_S^{>0}\type\in \frak{s}'_p$ then there is a $\ftype''=(\type',\frak{s}_p'',\frak{s}_c'')\in N_i'$ such that $\Diamond_S^{>0}\type\in \frak{s}''_p$ and $(iv)$ $|N_1|+|N_2|\leq k_{\varphi_{1}\varphi_2}$.
	
	Note that while $\ftype''$ may have different parent and child profiles, it must be based on the same type $\type'$ as $\ftype'$. Such $N'$ exists because $N_1'$ and $N_2'$ need to contain, in the worst case, one element for each child profile (to ensure child profile-equivalence) and one further element for each $S$ and each pair of types. There are exponentially many child profiles, exponentially many types and polynomially many relations, so we can select $N_1'\subseteq M_1'$ and $N_2'\subseteq M_2'$ of the required size.
	
	Let $\mathcal{K}$ be the result of replacing each $M'\in\mathcal{M}$ by the corresponding $N'$. We show that $\mathcal{K}$ is a syntactic witness for $M$ by constructing the witnessing relation $\S^{M,\mathcal{K}}$. For any $(\ftype,(\ftype',M'))\in \S^{M,\mathcal{M}}$, if $\ftype'\in N'$, then we simply take $(\ftype,(\ftype',M'))\in \S^{M,\mathcal{K}}$. If $\ftype'\not \in N'$, let $\ftype = (\type,\frak{s}_p,\frak{s}_c)$ and $\ftype' = (\type',\frak{s}_p',\frak{s}_c')$. 
	
	Because of the \textbf{(coherence)} condition, we have $\ftype \rightsquigarrow_{S}\ftype'$ and therefore, in particular, $\Diamond^{>0}_{S^-}\type\in \frak{s}'_p$. By construction, there must be some $\ftype''=(\type',\frak{s}_p'',\frak{s}_c'')\in N'$ such that $\Diamond_{S^-}\type\in \frak{s}''_p$, which implies that $\ftype \rightsquigarrow_{S}\ftype''$. We then take $(\ftype,(\ftype'',N'))\in \S^{M,\mathcal{K}}$.
	
	It is straightforward to verify that this $\S^{M,\mathcal{K}}$ satisfies all of the requirements for $\mathcal{K}$ being a syntactic witness for $M$:
	\begin{description}
		\item[(coherence)] holds because any $(\ftype,(\ftype',N'))\in \S^{M,\mathcal{K}}$ either existed in $\S^{M,\mathcal{M}}$ or was redirected to some $\ftype'$ with the property that $\ftype\rightsquigarrow_S\ftype'$,
		\item[(bisim) and (wit$_3$)] hold because we redirected edges only within the same mosaic, so $N'$ is an $S$-successor of $M$ (w.r.t.\ $\S^{M,\mathcal{K}}$) iff $M'$ is an $S$-successor of $M$ (w.r.t.\ $\S^{M,\mathcal{M}}$),
		\item[(wit$_1$) and (wit$_2$)] hold because we redirected edges only to star-types based on the same type.
	\end{description}
	
	\textit{Step} 2.
	Now, we select a subset $\mathcal{N}\subseteq \mathcal{K}$ of at most exponential size, while it remains a syntactic witness for $M$. In order to do this, we redirect edges from $\S^{M,\mathcal{K}}$. Importantly, we redirect an edge $(\ftype, (\ftype',N'))\in \S^{M,\mathcal{K}}$ only to $(\ftype',N'')$ where $N''$ is already an $S$-successor of $M$ with respect to $\S^{M,\mathcal{K}}$.
	
	Formally, let $\mathcal{N}$ be any subset of $\mathcal{K}$ such that $(i)$ for every $S$, every $M'\in \mathcal{K}\setminus \mathcal{N}$, every $i\in\{1,2\}$, every $\ftype = (\type,\frak{s}_p,\frak{s}_c)\in M_i$ and every $\ftype' = (\type',\frak{s}_p',\frak{s}_c')\in M_i'$, if $(\ftype, (\ftype',M'))\in \S^{M,\mathcal{K}}$ then there are $M''\in \mathcal{N}$ and $\ftype''= (\type', \frak{s}_p'', \frak{s}_c'')\in M_i''$ such that $\ftype \rightsquigarrow_S\ftype''$ and $M_i''$ is an $S$-successor of $M$ w.r.t. $\S^{M,\mathcal{K}}$ and $(ii)$ $|\mathcal{N}|\leq k_{\varphi_1,\varphi_2}$. 
	Such $\mathcal{N}$ exists because, in the worst case, $\mathcal{N}$ needs to contain one mosaic for every $i$, every $S$ and every $\ftype,\ftype'$, of which there are exponentially many.
	
	We  show that $\mathcal{N}$ is a syntactic witness for $M$ by constructing the witnessing relation $\S^{M,\mathcal{N}}$ from the relation $\S^{M,\mathcal{K}}$. For any $(\ftype,(\ftype',M'))\in \S^{M,\mathcal{K}}$, if $M'\in\mathcal{N}$, then we simply take $(\ftype,(\ftype',M'))\in \S^{M,\mathcal{N}}$. If $M'\not \in \mathcal{N}$, then by construction there are an $S$-successor $M''\in\mathcal{N}$ of $M$ such that $\ftype'\in M''_1\cup M''_2$. We take $(\ftype,(\ftype',M''))\in \S^{M,\mathcal{N}}$.
	
	It is again straightforward to verify that all of the requirements for a syntactic witness are satisfied:
	\begin{description}
		\item[(coherence), (wit$_1$) and (wit$_2$)] hold because we an edge to $(\ftype',M')$ was redirected to $(\ftype',M'')$, i.e., the same star-type but in a different mosaic,
		\item[(bisim) and (wit$_3$)] hold because we redirected edges only to mosaics that were already $S$-successors of $M$.
	\end{description}
	
	\textit{Step} 3. Finally, we select subsets $N_1\subseteq M_1$  and $N_2\subseteq M_2$ such that $(i)$ $N=(N_1,N_2)$ is parent profile-equivalent to $M$ and $(ii)$ every edge $(\ftype,(\ftype',N'))\in \S^{M,\mathcal{N}}$ can be replaced by an edge $(\ftype'',(\ftype',N'))$ where $\ftype''\in N_1\cup N_2$.
	
	Formally, let $N=(N_1,N_2)$ be any mosaic such that $(i)$ $N_i\subseteq M_i$ for $i\in\{1,2\}$, $(ii)$ $M$ and $N$ are parent profile-equivalent, $(iii)$ for every $\type$, $\type'$ and $S$, if there is a $(\type, \frak{s}_p,\frak{s}_c)\in M_i$ such that $\Diamond^{>0}_S\type'\in \frak{s}_c$, then there is a $(\type,\frak{s}_p'',\frak{s}_c'')\in N_i$ such that $\Diamond_S^{>0}\type'\in \frak{s}''_c$ and $(iv)$ $|N_1|+|N_2|\leq k_{\varphi_1\varphi_2}$.
	
	We show that $\mathcal{N}$ is a syntactic witness for $N$ by constructing the witnessing relation $\S^{N,\mathcal{N}}$. For any $(\ftype,(\ftype',N'))\in \S^{M,\mathcal{N}}$, if $\ftype\in N$, then we take $(\ftype,(\ftype',N'))\in \S^{N,\mathcal{N}}$. If $\ftype\not \in N$, take $(\ftype'',(\ftype',N'))\in \S^{N,\mathcal{N}}$, where $\ftype''=(\type'',\frak{s}_p'',\frak{s}_c'')$ is the star-type such that $\Diamond_S^{>0}\type'\in \frak{s}''_c$. Note that in this case $\ftype''\rightsquigarrow_S \ftype'$.
	
	This $\S^{N,\mathcal{N}}$ satisfies all of the requirements of a syntactic witness:
	\begin{description}
		\item[(coherence)] holds because any edge $(\ftype,(\ftype',N'))\in \S^{N,\mathcal{N}}$ either was in $\S^{M,\mathcal{N}}$ or had $\ftype$ chosen in such a way that $\ftype\rightsquigarrow_S\ftype'$,
		\item[(wit$_1$) and (wit$_2$)] hold because $\ftype$ and $\ftype''$ are based on the same type $\type$,
		\item[(bisim) and (wit$_3$)] hold because we did not add any new $S$-successors.
	\end{description}
	We have now shown that $\mathcal{N}$ is a syntactic witness for $N$. Furthermore, $N$ and $\mathcal{N}$ satisfy all of the requirements of the lemma.
\end{proof}

We are now in a position to formulate the mosaic elimination procedure establishing the $2\ExpTime$ upper bound in Theorem~\ref{thm:gmliu}. Let $\mathfrak{S}_{0}$ be the set of all mosaics $M=(M_{1},M_{2})$ with non-empty $M_{1}$, $M_{2}$ and $|M_{1}\cup M_{2}|\leq k_{\varphi_1,\varphi_2}^{2}$. Clearly $\mathfrak{S}_{0}$ can be computed in double exponential time in $|\varphi_1|+|\varphi_2|$. 

Let $\mathfrak{S}\subseteq \mathfrak{S}_{0}$. We call $M\in \mathfrak{S}$ \emph{bad} in $\mathfrak{S}$ if at least one of the following conditions does not hold:
\begin{itemize}
\item[--] there exists a syntactic witness $\mathcal{M}\subseteq \mathfrak{S}$ for $M$ such that $|\mathcal{M}| \leq k_{\varphi_1\varphi_2}$;

\item[--] if $\Du\chi\in \type$ and $(\type,\frak{s}_{p},\frak{s}_{c})\in M_{i}$, for $M=(M_{1},M_{2})\in \mathfrak{S}$ and $i \in \{1,2\}$, then there is a root mosaic $M'=(M_{1}',M_{2}')\in \mathfrak{S}$ such that $\chi\in \type'$, for some $(\type',\frak{s}_{p}',\frak{s}_{c}')\in M_{i}'$, with $\type$ and $\type'$ being $\Du$-equivalent.
\end{itemize}
We compute a sequence $\mathfrak{S}_{0},\mathfrak{S}_{1},\dots$, where $\mathfrak{S}_{i+1}$ is obtained from $\mathfrak{S}_{i}$ by eliminating all bad mosaics, and stop when $\mathfrak{S}_{n+1} = \mathfrak{S}_{n}$, setting $\mathfrak{S}^{\ast} = \mathfrak{S}_{n}$.  Clearly, $\mathfrak{S}^{\ast}$ is computable in double exponential time.

\mainlemmaforgmliu*

\begin{proof}
$(i) \Rightarrow (ii)$ Suppose $\Amf_{i},a_{i}$, $i=1,2$, satisfying $(i)$ are given. Consider the set of all pairs $p=(M, \mathcal{M}_{M})$ such that $\mathcal{M}_{M}$ is a witness for $M$ given by $\Amf_{1}$ and $\Amf_{2}$. The union of all such $\{M\}\cup \mathcal{M}_{M}$ would not contain any bad mosaics if $M$ and $\mathcal{M}_{M}$ were of proper size ($\leq k_{\varphi_1\varphi_2}$), which is not necessarily the case. To deal with pairs $p=(M, \mathcal{M}_{M})$ of wrong size, we apply to them Lemma~\ref{lem:cardin} and obtain pairs $(M^{p},\mathcal{M}_{M}^{p})$ satisfying the conditions of the lemma. However, these pairs are not necessarily defined by $\Amf_{1},\Amf_{2}$. We call the $M^{p}$ \emph{child-good} and all mosaics in any $\mathcal{M}_{M}^{p}$ \emph{parent-good} (as we have found the relevant witnesses for graded modalities in children and, respectively, parents). 
	

Given a parent good mosaic $M=(M_{1},M_{2})$ and a child-good mosaic $M'=(M_{1}',M_{2}')$ that are parent- and child-profile equivalent, we construct a new mosaic $M|M'=(M_{1}|M_{1}',M_{2}|M_{2}')$ which, for any parent- and child-profile equivalent star-types $\ftype=(\type,\frak{s}_{p},\frak{s}_{c})\in M_{i}$ and $\ftype'=(\type',\frak{s}_{p}',\frak{s}_{c}')\in M_{i}'$, contains  the star-type $\ftype|\ftype'=(\type,\frak{s}_{p},\frak{s}_{c}')$ in $M_i|M'_i$. Observe that, by the definition of parent/child equivalence, all $\ftype|\ftype'$ are coherent star-types and $M|M'$ is a mosaic. Also observe that $|M|M'|\leq k_{\varphi_1\varphi_2}^{2}$. Let $\mathfrak{S}$ be the set of all $M|M'$ just constructed. Then $\mathfrak{S}\subseteq \mathfrak{S}_{0}$. It should be clear that $\mathfrak{S}$ does not contain any bad mosaics, and so $\mathfrak{S}\subseteq \mathfrak{S}^{\ast}$. Finally, by applying Lemma~\ref{lem:cardin} to $M=M(\Amf_{1},\Amf_{2},a_{1},\emptyset)$ we obtain the required root mosaic in $\mathfrak{S}$. 

$(ii) \Rightarrow (i)$ Suppose $M^{\ast}=(M_{1}^{\ast},M_{2}^{\ast})\in \mathfrak{S}^{\ast}$ is given by $(ii)$.
%
%
We use $\mathfrak{S}^*$ to construct the required $\Amf_{i},a_{i}$, $i=1,2$,	and $\MLiu(\varrho)$-bisimulation $\bis$ between them. For every $M\in \mathfrak{S}^{\ast}$, take a witness $\mathcal{M}_{M}\subseteq \mathfrak{S}^{\ast}$ for $M$ together with the corresponding relations $S^{M,\mathcal{M}_{M}} \subseteq \Delta^{M} \times \Gamma^{M}$ between $\Delta^{M}=\Delta_1^M \cup \Delta_2^M$ and $\Gamma^{M}=\Gamma_1^M \cup \Gamma_2^M$.
For $i=1,2$, the domain $\dom(\Amf_{i})$ of $\Amf_{i}$ comprises words of the form 
\begin{equation}\label{eq1second}
s = M_{0} \dots M_{n-1} (\ftype_{n},j_{n},M_{n}),
\end{equation}
%
	%
where $M_{j}=(M_{j,1},M_{j,2})\in \mathfrak{S}^{\ast}$, $0 \le j \le n$, 
\begin{itemize}
\item[--] $M_{0}$ is a root mosaic and any type in $M_{0,i}$ is $\Du$-equivalent to any type in $M^{\ast}_{i}$;

\item[--] $M_{j+1}\in \mathcal{M}_{M_{j}}$, for all $j < n$;

\item[--] $\ftype_{n} \in M_{n,i}$ and $j_{n}<\omega$.
\end{itemize}
It follows from the definition of (bad) mosaics that all types used in the construction of  $\dom(\Amf_{i})$ are $\Du$-equivalent to each other.

Let $\type_n$ be the type of $\ftype_{n}$, in which case $\type_n$ is also called the \emph{type} of $s$. We interpret the propositional variables $A \in \sigma$ by taking $s\in A^{\Amf_{i}}$ if $A\in \type_{n}$. For any $R\in \sigma$, we set 
\begin{itemize}
\item[--] $(s',s)\in R^{\Amf_{i}}$ for any $s$ of the form \eqref{eq1second} and any $s'$ of the form
\begin{equation}\label{eq1second'}
s' = M_{0} \dots M_{n-2} (\ftype_{n-1},j_{n-1},M_{n-1})
\end{equation}
such that $(\ftype_{n-1},j_{n-1})R^{M_{n-1},\mathcal{M}_{M_{n-1}}}(\ftype_{n},j_{n},M_{n})$;

\item[--] $(s,s')\in R^{\Amf_{i}}$ for any $s$ of the form \eqref{eq1second} and any $s'$ of the form \eqref{eq1second'}
such that $(\ftype_{n-1},j_{n-1}){(R^{-})}^{M_{n-1},\mathcal{M}_{M_{n-1}}}(\ftype_{n},j_{n},M_{n})$.
\end{itemize}
%
%
%
%
We show by induction that, for any $\chi\in\sub(\varphi_1,\varphi_2)$ and any $s$ of the form \eqref{eq1second}, we have $\Amf_{i}\models \chi(s)$ iff $\chi\in \type_{n}$. The basis of induction and the Boolean cases of the induction step are trivial. Suppose $\Amf_{i}\models \Du\chi(s)$. Then $\Amf_{i}\models \chi(s')$, for some $s' \in \dom(\Amf_{i})$. It $\type'$ is the type of $s'$, then $\chi \in \type'$ by IH, and so $\Du \chi \in \type'$ and $\Du \chi \in \type_n$. 
Conversely, if $\Du \chi\in \type_{n}$, then, by the definition of bad mosaics, $\mathfrak{S}^*$ contains a root mosaic $M'$ with a star-type $\ftype'$ of some type $\type' \ni \chi$ in its $i$th component. By IH, $\Amf_{i}\models \chi(s')$ for any $s' = (\ftype',j,M')$, and so $\Amf_{i}\models \Du\chi(s)$.
Finally, suppose $\Amf_{i}\models \Diamond_S^{\ge k}\chi(s)$ and $\ftype_n = (\type_n,\frak{s}_{p}^n,\frak{s}_{c}^n)$. To show that $\Diamond_S^{\ge k}\chi 
\in \type_n$, we need to prove that $P_{S}(\frak{s}_{p},\frak{s}_{c})$ contains $\Diamond_S^{> \mx}\chi$ or $\Diamond_S^{= m}\chi$, for some $m \ge k$, which follows from IH and conditions {\bf (witness$_1$)} and {\bf (witness$_2$)}. The converse implication also follows from {\bf (witness$_1$)}, {\bf (witness$_2$)} and IH.


We next define $\bis\subseteq \dom(\Amf_{1})\times\dom(\Amf_{2})$ by taking $s\bis s'$, for $s \in \dom(\Amf_{1})$ of the from \eqref{eq1second} and $s'\in \dom(\Amf_{2})$ of the form
\begin{equation}\label{eq2secondsec}
s'= M_{0}'\dots M_{m-1}'(\ftype_{m}',j_{m}',M_{m}')
\end{equation}
if $n=m$ and $M_{j}=M_{j}'$ for all $j\leq n$. It follows from non-emptiness of mosaics in $\mathfrak{S}^*$, the definition of relations $R^{\Amf_{i}}$ and conditions {\bf (bisim$_1$)} and {\bf (bisim$_2$)} that $\bis$ is a $\MLiu(\varrho)$-bisimulation. 
\end{proof}

Lemmas~\ref{criterion}, \ref{l:red}, and ~\ref{lem:gmliuall} give the $2\ExpTime$ upper bound for deciding $\GMLiu$/$\MLiu$-Craig separation in Theorem~\ref{thm:gmliu}.
%
%
%
The $2\ExpTime$ upper bound for $\GMLiu/\MLi$ is obtained by dropping the condition that $M_{1}$ and $M_{2}$ are both non-empty for mosaics $M=(M_{1},M_{2})$. The one for $\GMLiu/\MLu$ is obtained by treating the inverse $R^{-}$ with $R\in \varrho$ in the same way as $R\in \sigma\setminus\varrho$. Both of these modifications give the upper bound for $\GMLiu/\ML$.

It remains to establish $\textsc{co}\NExpTime$ upper bounds for $\GMLi/\MLi$  and $\GMLi/\ML$.
%
%
%
%
%
Suppose $\varphi_{1},\varphi_{2}$ are $\GMLi$-formulas. 
We again ignore everything related to $\Du$ and admit mosaics $M = (M_{1},M_{2})$ with $M_{1}=\emptyset$ or $M_{2}=\emptyset$ (but not both). Let $m = \max \{\md(\varphi_1), \md(\varphi_2)\}$. Then it follows from the proof of Lemma~\ref{lem:gmliuall} (and the definition of $\bis$ in it) that there exist $\Amf_{i},a_{i}$, $i=1,2$, with $\Amf_i \models \varphi_{i}(a_{i})$ and $\Amf_{1},a_{1} \sim_{\MLi(\varrho)} \Amf_{2},a_{2}$ iff we find words 
\begin{equation}\label{eq1secondneww}
s = M_{0}M_{1}\dots M_{n}
\end{equation}
with $n\leq \md(\varphi,\psi)$ and  $M_{j}=(M_{j,1},M_{j,2})$ such that
\begin{itemize}
\item[--] $M_{j}$ a mosaic of size $\leq k_{\varphi_{1},\varphi_{2}}^{2}$, for $j\leq n$;

\item[--] $M_{0}=M^{\ast}$ for a fixed root mosaic $M^*=(M_{1}^{\ast},M_{2}^{\ast})$ such that there are $(\type,\frak{s}_{p},\frak{s}_{c})\in M_{1}^{\ast}$ and $(\type',\frak{s}_{p}',\frak{s}_{c}')\in M_{2}^*$ with $\varphi_{1}\in \type$ and $\varphi_{2}\in \type'$;

\item[--] $M_{j+1}\in \mathcal{M}_{M_{j}}$, for a syntactic witness $\mathcal{M}_{M_{j}}$ for $M_{j}$ with $|\mathcal{M}_{M_{j}}|\leq k_{\varphi,\psi}$, for all $j< n$.
\end{itemize}
A \NExpTime{} algorithm deciding non-separation could guess words of the form~\eqref{eq1secondneww} and then check in exponential time that they satisfy the three conditions above. 

The matching lower bounds are proved as in the previous section. This completes the proof of Theorem~\ref{thm:gmliu}.
\end{proof}


\section*{Proofs for Section~\ref{sec:definability}}
We start by giving the precise recursive definition of the flattening $\Flat(\varphi)\in \FOTNE$ for $\CT$-formulas $\varphi$ without equality:
\[\begin{array}{ll}
	\Flat(A(x))=A(x), & \Flat(R(x,y)) = R(x,y), \\
	\Flat(\exists^{\geq k}x \, \varphi)= \exists x \, \Flat(\varphi), & \Flat(\neg \varphi) = \neg \Flat(\varphi),\\
	\Flat(\varphi\wedge\varphi') = \Flat(\varphi)\wedge\Flat(\varphi'). & 
\end{array}\]
We require the following observation.
\begin{lemma}\label{lem:withoueq}
	Let $\sigma$ be constant-free and let $\kappa\in \mathbb{N}^{\infty}$. Then, for all $\FOTNE(\sigma)$-formulas $\psi_{1}(x)$ and $\psi_{2}(x,y)$, all $a,b\in \dom(\Amf)$, and all $i,j,i',j'<\kappa$, we have\textup{:}
	\begin{itemize}
		\item[--] $\Amf^{\kappa}\models \psi_{1}(a,i)$ iff $\Amf^{\kappa}\models \psi_{1}(a,j)$\textup{;}
		\item[--] $\Amf^{\kappa}\models \psi_{2}((a,i),(b,j))$ iff $\Amf^{\kappa}\models \psi_{2}((a,i'),(b,j'))$.
	\end{itemize}
\end{lemma}

\lemmaequivalenceomega*

\begin{proof}
	Let $\mx<\kappa \in \mathbb{N}^{\infty}$. First, we note that, because the signature contains no constants, for every $a\in \dom(\Amf)$ and every $i<\kappa$, we have $(a,i)\in \dom(\Amf^\kappa)$.
	
	We show that $\Amf^{\kappa}\models \varphi(a,i)$ iff $\Amf^{\kappa}\models \Flat(\varphi)(a,i)$ for all elements $a\in \dom(\Amf)$ and $i<\kappa$. Observe that $\Amf^{\kappa}\models \Flat(\varphi)(a,i)$ iff $\Amf^{\kappa'}\models \Flat(\varphi)(a,i)$ follows from Lemma~\ref{lemma:bisim_omega}, for all $i<\kappa$ and $\kappa'\in \mathbb{N}^{\infty}$. 
	
	Now the proof proceeds by induction on the construction of $\varphi$. Assume first that $\varphi\in \GMLiu$. We only discuss the critical induction step, where $\varphi$ is of the form $\Diamond\!^{\ge k}_R \psi$ for some $\psi$, in which case $\Flat(\varphi)$ is of the form $\Diamond\!_R\,\Flat(\psi)$.
	
	Suppose that $\Amf^\kappa\models \Diamond\!^{\ge k}_R \psi(a,i)$. We assumed that $k>0$, so there is at least one $(b,j)$ such that $\Amf^\kappa \models \psi(b,j)$ and $((a,i), (b,j))\in R^{\Amf^{\kappa}}$. By the induction hypothesis, we then also have $\Amf^\kappa\models \Flat(\psi)(b,j)$, and therefore $\Amf^\kappa\models \Diamond_R\, \Flat(\psi)(a,i)$.
	
	For the other direction, suppose that $\Amf^\kappa\models \Diamond_R\Flat(\psi)(a,i)$, so there is some $(b,j)$ such that $\Amf^\kappa\models \Flat(\psi)(b,j)$ and $((a,i),(b,j))\in R^{\Amf^{\kappa}}$. By Lemma~\ref{lem:withoueq}, this implies that we have $\Amf^\kappa\models \Flat(\psi)(b,j')$ for every $j'<\kappa$. 
	
	By the induction hypothesis, this implies that $\Amf^\kappa \models \psi(b,j')$ for all $j'<\kappa$. Furthermore, by the construction of $\Amf^\kappa$, it follows that $((a,i), (b,j'))\in R^{\Amf^{\kappa}}$ for all $j'<\kappa$. From $\mx< \kappa$ we obtain $\Amf^\kappa\models\Diamond\!^{\ge k}_R\psi(a,i)$. This completes the induction step.
	
	The proof for $\varphi\in \CT$ is a straightforward extension of the proof above. It is again by induction on the construction of $\varphi$. For formulas $\psi(x,y)$, we claim that $\Amf^{\kappa}\models \psi((a,i),(b,j))$ iff $\Amf^{\kappa}\models \Flat(\psi)((a,i'),(b,j'))$. We consider the step $\varphi= \exists^{\geq k}x \,\psi(x,y)$ (we may assume that $y$ is free in $\psi$). 	
	Suppose that $\Amf^\kappa\models \exists^{\geq k}x \, \psi(a,i)$. We assumed that $k>0$, so there is at least one $(b,j)$ such that $\Amf^\kappa \models \psi((b,j),(a,i))$. By the induction hypothesis, we then also have $\Amf^\kappa\models \Flat(\psi)((b,j),(a,i))$, and therefore $\Amf^\kappa\models \exists x \, \Flat(\psi)(a,i)$.
	
	For the other direction, suppose that $\Amf^\kappa\models \exists x \, \Flat(\psi)(a,i)$, so there is some $(b,j)$ such that $\Amf^\kappa\models \Flat(\psi)((b,j),(a,i))$. By Lemma~\ref{lem:withoueq}, this implies that $\Amf^\kappa\models \Flat(\psi)((b,j'),(a,i))$ for every $j'<\kappa$. 
	
	By IH, this implies that $\Amf^\kappa \models \psi((b,j'),(a,i))$ for all $j'<\kappa$. From $\mx< \kappa$ we obtain $\Amf^\kappa\models \exists^{\geq k}x\, \psi(a,i)$. This completes the induction step and thereby the proof.
\end{proof}

The following lemma can be shown similarly to Lemma~\ref{lemma:equivalenceomega}.

\begin{lemma}\label{lem:equivonomega2}
	Suppose $\varphi$ is a $\GMLinu(\sigma)$-formula with finite $\sigma_{c}$ and $\mx<\kappa<\kappa'\in \mathbb{N}^{\infty}$. Then 
	$\Amf^\kappa\models \varphi(a,i)$ iff $\Amf^\kappa\models \Flat_{\sigma_{c}}(\varphi)(a,i)$
	iff $\Amf^{\kappa'}\models \Flat_{\sigma_{c}}(\varphi)(a,i)$ iff $\Amf^{\kappa'}\models \varphi(a,i)$, for all pointed $\sigma$-structures $\Amf,a$ and all $i <\kappa$.
\end{lemma}

Note that, according to the definition of $\L$/$\LS$-definability, defining $\LS$-formulas $\varphi'$ are allowed to contain fresh symbols not occurring in $\varphi$. While fresh propositional variables or modal operators do not affect definability (see below), we 
leave it open whether using a fresh nominal has an effect on definability of formulas without nominals.

\begin{theorem}\label{prop:equivonomega2}
	Let $\L/\LS$ be any pair of modal logics from Table~\ref{table:results}, with $L$ admitting nominals, and let $\sigma$ contain at least one constant. Then an $\L(\sigma)$-formula $\varphi$ is $\LS(\sigma)$-definable iff $\models \varphi \leftrightarrow \Flat_{\sigma_{c}}(\varphi)$.
\end{theorem}
\begin{proof}
	Note that if $\models \varphi \leftrightarrow \varphi'$ and $\varphi'$ uses symbols not in $\sigma$, then we can replace any variable $A\not\in\sigma$ by $\top$, any $\Diamond_{R}\chi$ with $R\not\in\sigma$ by $\neg \top$, and any $N_{c}$ with $c\not\in\sigma$ by an $N_{c'}$ with $c'\in \sigma$, and obtain $\models \varphi \leftrightarrow \varphi''$ for the resulting formula $\varphi''$. Now the proof is similar to the proof of Theorem~\ref{prop:definability1}.
\end{proof}

It remains open whether there are polynomial DAG-size formulas in $\LS$ 
equivalent to $\Flat_{\sigma_{c}}(\Diamond\!^{\ge k}_R \psi)$. Note that the problem is that distinct constants can denote the same domain element.
In contrast, under the unique name assumption, polynomial DAG-size formulas
are easily constructed. Indeed, let $Th_{k}^{n}(p_{0},\dots, p_{n-1})$ be a propositional formula that is true under a propositional assignment iff at least $k$ variables among $p_{0},\dots,p_{n-1}$ are true. It is folklore that there are such formulas of polynomial DAG-size. Let $0,1,\ldots,n-1$ be the constants in $\sigma_{c}$.
Then the formula obtained by replacing in $\Flat_{\sigma_{c}}(\Diamond\!^{\ge k}_R \psi)$ the second conjunct by the formula obtained from $Th_{k}^{n}$ by replacing each $p_{i}$ by $\Diamond_{R}(i\wedge \Flat_{\sigma_{c}}(\psi))$ is as required.

We note that the complexity results for the validity/satisfi\-ability problems for the logics considered in this paper do not depend on whether one defines the size of an input $\varphi$ as the \emph{formula-size} of $\varphi$ (the length of $\varphi$ as a word) or the \emph{DAG-size} $|\varphi|$ of $\varphi$ (the number of subformulas of $\varphi$).
Here, we sketch a reduction from DAG-size to formula-size by introducing 
abbreviations. 
Recall that $\Flat_{\sigma_{c}}(\varphi)$ is equivalent to the $\L$-formula $f(\varphi)$ defined recursively in the same way as 
$\Flat_{\sigma_{c}}(\varphi)$ except that
$$
f_{\sigma_{c}}(\Diamond\!^{\ge k}_R\psi))=\big[\Diamond_R \big(\Flat_{\sigma_{c}}(\psi)\wedge \bigwedge_{c\in \sigma_{c}} \neg N_c\big)\vee{} \Diamond_{R}^{\geq k}(\Flat_{\sigma_{c}}(\psi) \wedge \bigvee_{c\in \sigma_{c}}N_{c}) \big].
$$
Let $\Flat'_{\sigma_{c}}$ be also defined in the same way as $\Flat_{\sigma_{c}}$ except that we set, for a fresh propositional variable $A_\psi$,
$$
\Flat'_{\sigma_{c}}(\Diamond\!^{\ge k}_R \psi) = 
\big[\Diamond_R \big(A_\psi \wedge \bigwedge_{c\in \sigma_{c}} \neg N_c\big)\vee{} \Diamond_{R}^{\geq k}(A_\psi \wedge \bigvee_{c\in \sigma_{c}}N_{c})\big].
$$
The formula-size of $\Flat'_{\sigma_{c}}(\varphi)$ is clearly polynomial in the formula-size of $\varphi$. 
If $\L$ contains the universal modality, then $\models \varphi \leftrightarrow \Flat_{\sigma_{c}}(\varphi)$ iff the following $\L$-formula (of polynomial formula-size) is valid:
$$
\blacksquare \bigwedge \{A_\psi \leftrightarrow \Flat'_{\sigma_{c}}(\psi) \mid \Diamond\!^{\ge k}_R \psi \in \sub (\varphi)\} \rightarrow \big( \varphi \leftrightarrow \Flat'_{\sigma_{c}}(\varphi) \big).
$$
If $\L$ does not admit the universal modality, instead of $\blacksquare$ we take all sequences of the form $\Box_{R_1} \dots \Box_{R_m}$ such that $\varphi$ contains nested $\Diamond^{\ge k_1}_{R_1} \dots \Diamond^{\ge k_m}_{R_m}$, with each $\Diamond^{\ge k_i}_{R_i}$, $i =1,\dots,m-1$, being an immediate predecessor of $\Diamond^{\ge k_{i+1}}_{R_{i+1}}$.

\medskip

We now give the detailed definition of $\Flat(\varphi)$ for $\CT/\FOTNE$. Recall that the straightforward extension of the flattening $\Flat$ defined above for $\CT$ without equality does not work. So we now define a new flattening $\Flat$ from $\CT$ to $\FOTNE$. Suppose we are given a $\CT$-formula $\varphi$.
As before, only the definition of $\Flat(\exists^{\geq k}x \, \psi)$ is non-trivial. Assume that $\psi$ is a Boolean combination $\beta$ of the form 
$$
\psi = \beta(\rho_{1},\ldots,\rho_{k_{1}}, \gamma_{1}(x),\dots,\gamma_{k_{2}}(x),\xi_{1}(y),\dots,\xi_{k_{3}}(y)),
$$
where
\begin{itemize}
	\item[--] the $\rho_{i}$ are binary atoms $R(x,y)$, $R(y,x)$, or $x=y$; 
	
	\item[--] the $\gamma_{i}(x)$ are unary atoms of the form $A(x)$, $x=x$, or $R(x,x)$, or a formula of the form $\exists^{\geq k'} y\, \gamma_{i}'$; 
	
	\item[--] the $\xi_{i}(y)$ are unary atoms of the form $A(y)$, $y=y$ or $R(y,y)$, or a formula of the form $\exists^{\geq k'} x \, \xi_{i}'$.
\end{itemize}
We may assume that $y$ is free in $\psi$. Consider the formula
$$
\varphi_{1} = \exists^{\geq k} x \, \big(((x=y) \wedge \psi_{1}) \vee ((x\not=y) \wedge \psi_{2})\big),
$$
in which $\psi_{1}$ is obtained from $\psi$ by replacing 
\begin{itemize}
	\item[--] all $\rho_{i}$ of the form $x=y$ by $\top$;
	
	\item[--] all $\gamma_{i}(x)$ of the form $x=x$ by $\top$;
	
	\item[--] all $\xi_{i}(y)$ of the form $y=y$ by $\top$; 
\end{itemize}
and $\psi_{2}$ is obtained from $\psi$ by replacing
\begin{itemize}
	\item[--] all $\rho_{i}$ of the form $x=y$ by $\neg \top$;
	
	\item[--] all $\gamma_{i}(x)$ of the form $x=x$ by $\top$;
	
	\item[--] all $\xi_{i}(y)$ of the form $y=y$ by $\top$.
\end{itemize}
Clearly, $\varphi$ and $\varphi_{1}$ are logically equivalent. Also, the formula $\exists x\, ((x=y)\wedge \psi_{1})$
is logically equivalent to the formula $\psi_{1}'$ obtained from $\psi_{1}$ by replacing all
$\rho_{i}$ of the form $R(x,y)$ or $R(y,x)$ by $R(y,y)$, all $\gamma_{i}(x)$ of the form $A(x)$ 
by $A(y)$, and all $\gamma_{i}(x)$ of the form $\exists^{\geq k'}\! y\, \gamma_{i}'$ by the result of swapping $x$ and $y$. Now set, recursively, 
\begin{align*}
	& \Flat(\exists^{\geq k}x\,\psi)= \Flat(\psi_{1}')\vee \exists x\, \Flat(\psi_{2}), \quad \text{for $k=1$},\\
	& \Flat(\exists^{\geq k}x\,\psi)= \exists x\, \Flat(\psi_{2}), \hspace*{1.4cm} \text{for $k>1$}.
\end{align*}
Observe that we do not define $\Flat(\psi(x,y))$ for formulas $\psi(x,y)$ in $\CT$ with occurrences of equality that are not in the scope of a quantifier.

It is now straightforward to prove analogues of Lemma~\ref{lemma:equivalenceomega}
and Theorem~\ref{prop:definability1} for $\CT/\FOTNE$.

\begin{lemma}\label{lem:equivonomega27}
	Let $\sigma$ be constant-free, let $\varphi(x)$ be a $\CT(\sigma)$-formula $($with equality$)$, and let $\mx<\kappa<\kappa'\in \mathbb{N}^{\infty}$. Then, for all pointed $\sigma$-structures $\Amf,a$ and all $i <\kappa$, we have 
	$\Amf^\kappa\models \varphi(a,i)$ iff $\Amf^\kappa\models \Flat(\varphi)(a,i)$
	iff $\Amf^{\kappa'}\models \Flat(\varphi)(a,i)$ iff $\Amf^{\kappa'}\models \varphi(a,i)$.
\end{lemma}
\begin{proof}
As before, the interesting part is to show that we have $\Amf^{\kappa}\models \varphi(a,i)$ iff $\Amf^{\kappa}\models \Flat(\varphi)(a,i)$, for all $a\in \dom(\Amf)$ and $i<\kappa$. 
	
	The proof is by induction on the construction of $\varphi$. For $\CT$-formulas $\psi(x,y)$ without occurrences of equality that are not in the scope of a quantifier, we claim that $\Amf^{\kappa}\models \psi((a,i),(b,j))$ iff $\Amf^{\kappa}\models \Flat(\psi)((a,i'),(b,j'))$, for all $a,b\in \dom(\Amf)$ and all $i,j,i',j'<\kappa$. Observe that the induction step for such formulas follows from Lemma~\ref{lem:withoueq}. 
	
	We consider the step $\varphi= \exists^{\geq k}x \, \psi(x,y)$, where we may assume that $y$ is free in $\psi$. Assume from above:
	\begin{itemize}
		\item[--] $\psi=\gamma(\rho_{1},\ldots,\rho_{k_{1}}, \gamma_{1}(x),\ldots,\gamma_{k_{2}}(x),\xi_{1}(y),\ldots,\xi_{k_{3}}(y))$;
		
		\item[--] $
		\varphi_{1} = \exists^{\geq k} x \, \big[ \big((x=y) \wedge \psi_{1} \big) \vee \big((x\not=y) \wedge \psi_{2}\big)\big]$;
		
		\item[--] 
		for $k=1$, 
		$\Flat(\exists^{\geq k}x \, \psi)= \Flat(\psi_{1}')\vee \exists x \, \Flat(\psi_{2})$;
		
		\item[--] for $k>1$,
		$
		\Flat(\exists^{\geq k}x \, \psi)= \exists x \, \Flat(\psi_{2}).
		$
	\end{itemize}
Suppose that $\Amf^\kappa\models \exists^{\geq k}x \, \psi(a,i)$. 

If $k=1$, then we have either $\Amf^{\kappa}\models \psi((a,i),(a,i))$, and so $\Amf^{\kappa} \models \psi_{1}'(a,i)$, or $\Amf^{\kappa}\models \psi((b,j),(a,i))$, for some $(b,j)\not=(a,i)$, and so $\Amf^{\kappa} \models \psi_{2}((b,j),(a,i))$. Then, by IH, $\Amf^{\kappa} \models \Flat(\psi_{1}')(a,i)$ or $\Amf^{\kappa} \models \Flat(\psi_{2})((b,j),(a,i))$. In both cases $\Amf^\kappa\models \Flat(\varphi)$.
	
	If $k>1$, then $\Amf^{\kappa}\models \psi((b,j),(a,i))$, for some $(b,j)\not=(a,i)$, and so again $\Amf^{\kappa} \models \psi_{2}((b,j),(a,i))$ and $\Amf^\kappa\models \Flat(\varphi)$ by IH.
	
	For the other direction, suppose $\Amf^\kappa\models \Flat(\varphi)(a,i)$. 
	
	Let $k=1$. If $\Amf^\kappa\models \Flat(\psi_{1}')(a,i)$, then by IH, $\Amf^{\kappa}\models \psi_{1}'(a,i)$, and so $\Amf^{\kappa}\models \varphi(a,i)$. If $\Amf^{\kappa}\models \exists x \, \Flat(\psi_{2})(a,i)$, then there is some $(b,j)$ such that $\Amf^\kappa\models \Flat(\psi_{2})((b,j),(a,i))$. By Lemma~\ref{lem:withoueq}, we have $\Amf^\kappa\models \Flat(\psi_{2})((b,j'),(a,i))$ for some $(b,j')\not=(a,i)$.  
	By the induction hypothesis, this implies that $\Amf^\kappa \models \psi_{2}((b,j'),(a,i))$. Thus, we obtain $\Amf^\kappa\models \exists^{\geq k}x \, \psi(a,i)$. 
	
	Let $k>1$. Then $\Amf^{\kappa}\models \exists x \, \Flat(\psi_{2}(a,i)$. We then have some $(b,j)$ such that $\Amf^\kappa\models \Flat(\psi_{2})((b,j),(a,i))$. By Lemma~\ref{lem:withoueq}, $\Amf^\kappa\models \Flat(\psi_{2})((b,j'),(a,i))$ for every $j'<\kappa$. 
	By the induction hypothesis, this implies that $\Amf^\kappa \models \psi_{2}((b,j'),(a,i))$ for all $j'<\kappa$. Since $\mx< \kappa$, we obtain $\Amf^\kappa\models \exists^{\geq k}x \, \psi(a,i)$, as required. 
\end{proof}

For the polynomial reduction of checking $\models\varphi \leftrightarrow \Flat(\varphi)$ to $\CT$-validity, 
observe that, 
although $\Flat(\varphi)$ is exponential in formula-size as we duplicate the formulas $\gamma_{i}(x)$ and $\xi_{i}(y)$, we can again   use the DAG-representation instead. 

\begin{corollary}
	$\CT/\FOTNE$-definability is polytime reducible to $\CT$-validity.
\end{corollary}

\lemtransss* 

\begin{proof}
$(i)$ is trivial. $(ii)$ We construct $\varphi'$ by induction. The case of atomic formulas is clear. Assume $\chi(x,y)$ is a Boolean combination $\beta$ of the form 
	$$
	\chi = \beta(\gamma_{1}(x),\dots,\gamma_{k_{1}}(x),\xi_{1}(y),\dots,\xi_{k_2}(y)),
	$$
	where
	\begin{itemize}
		\item[--] the $\gamma_{i}(x)$ are unary atoms of the form $A(x)$, $\top(x)$, or a formula of the form $\exists y\, \gamma_{i}^\ast$; 
		
		\item[--] the $\xi_{i}(y)$ are unary atoms of the form $A(y)$, $\top(y)$, or a formula of the form $\exists x \, \xi_{i}^\ast$.
	\end{itemize} 
	A \emph{$\gamma$-type} $t$ is a maximal satisfiable subset of $$
	\{\gamma_{1}(x),\dots,\gamma_{k_{1}}(x),\neg\gamma_{1}(x),\dots,\neg\gamma_{k_{1}}(x)\}.
	$$
	A \emph{$\xi$-type} $t$ is a maximal satisfiable subset of 
	$$
	\{\xi_{1}(y),\dots,\xi_{k_{2}}(y),\neg\xi_{1}(y),\dots,\neg\xi_{k_{2}}(y)\}.
	$$
	Let $S$ be the set of all pairs $(t_{1},t_{2})$ of $\gamma$-types $t_{1}$ and $\xi$-types $t_{2}$ such that $t_{1}\cup t_{2}\models \chi$. Consider now, inductively, the equivalent translations $\gamma_{i}'$ and $\xi_{i}'$ of $\gamma_{i}(x)$ and $\xi_{i}(y)$ into $\L$, respectively. For any $\gamma$-type $t$, let $t'=\bigwedge_{\gamma\in t}\gamma'$ and, for any $\xi$-type $t$, let 
	$t'=\bigwedge_{\xi\in t}\xi'$. Now define
\begin{itemize}
\item[--] $(\exists y \, (R(x,y) \wedge \chi(x,y)))'= \bigvee_{(t_{1},t_{2})\in S}(t_{1}' \wedge \Diamond_{R}t_{2}')$;

\item[--] $(\exists x \, (R(y,x) \wedge \chi(x,y)))'= \bigvee_{(t_{1},t_{2})\in S}(t_{2}' \wedge \Diamond_{R}t_{1}')$;

\item[--] $(\exists y \, (R(y,x) \wedge \chi(x,y)))'= \bigvee_{(t_{1},t_{2})\in S}(t_{1}' \wedge \Diamond_{R}^{-}t_{2}')$;

\item[--] $(\exists x \, (R(x,y) \wedge \chi(x,y)))'= \bigvee_{(t_{1},t_{2})\in S}(t_{2}' \wedge \Diamond_{R}^{-}t_{1}')$;

\item[--] $(\exists y \, \chi(x,y))' = \bigvee_{(t_{1},t_{2})\in S}(t_{1}' \wedge \Du t_{2}')$;

\item[--] $(\exists x \, \chi(x,y))' = \bigvee_{(t_{1},t_{2})\in S}(t_{2}' \wedge \Du t_{1}')$.
\end{itemize} 
It is straightforward to see that $\cdot'$ is a map from suc$\L$ to $\L$ and that
$\Amf\models \varphi(a)$ iff $\Amf\models \varphi'(a)$, for all pointed $\Amf,a$.
\end{proof}

We now consider $\FOTNE/\MLiu$-definability, where $\MLiu$ is represented in suc$\MLiu$. For every $\mathfrak{A}$, we take the \emph{twin-unfolding} $\Amf^{\ast}$ of $\mathfrak{A}$ whose domain $\dom(\Amf^{\ast})$ is the set of 
\begin{equation}\label{eq:def1}
	\vec{a}=a_{0}^{j}S_{0}\dots S_{n-1}a_{n}^{j}
\end{equation}
with $n\geq 0$, $j\in \{0,1\}$, $a_{i}\in \dom(\Amf)$ for all $i\leq n$, and $S_{i}$ either a binary predicate symbol in $\sig(\varphi)$ or its inverse, and such that $(a_{i},a_{i+1})\in S_{i}^{\Amf}$ for $i\leq n$. For $\vec{a}$ of the form \eqref{eq:def1}, let $\text{tail}(\vec{a})=a_{n}^{j}$. Set $\vec{a}\in A^{\Amf^{\ast}}$ if $a_{n}\in A^{\Amf}$ for 
$\text{tail}(\vec{a})=a_{n}^{j}$.
It remains to define $R^{\Amf^{\ast}}$ for a binary predicate symbol $R$. Assume $\vec{a}$ takes the form \eqref{eq:def1} and
let $\vec{a}'=a_{0}^{j}S_{0}\cdots a_{n}^{j}S_{n}a_{n+1}^{j}$ be an extension of $\vec{a}$. Then $(\vec{a},\vec{a}')\in R^{\Amf^{\ast}}$ if $R=S_{n}$ and 
$(\vec{a}',\vec{a})\in R^{\Amf^{\ast}}$ if $R^{-}=S_{n}$.  

Observe that $\Amf,a\sim_{\MLiu} \Amf^{\ast},a^{j}$ for $j=0,1$ since $\vec{a}\mapsto \text{tail}(\vec{a})$ is an $\MLiu$-bisimulation between $\Amf^{\ast}$ and $\Amf$.

\begin{lemma}\label{lem:twinunfoldingtrans}
	For all $\varphi\in \MLiu$ and all pointed $\Amf,a$: $\Amf^{\ast}\models \varphi(a^0)$ iff $\Amf\models\varphi(a)$.
\end{lemma}

\begin{lemma}\label{lem:twinunfoldingequiv}
For all pointed structures $\Amf,a$ and $j=0,1$, we have $\Amf^{\ast}\models \varphi(a^{j})$ iff $\Amf^{\ast}\models m(\varphi)(a^{j})$.
\end{lemma}
\begin{proof}
	The proof uses the following properties of $\Amf^{\ast}$: (1) for all $a,b\in \dom(\Amf^{\ast})$,
	if $(a,b)\in R^{\Amf^{\ast}}$, then $(a,b)\not\in S^{\Amf^{\ast}}$ for any $S\not=R$ and $(b,a)\not\in S^{\Amf^{\ast}}$ for any $S$; (2) for all $a\in \dom(\Amf^{\ast}$ and all Boolean combination of subformulas $\psi$ of $m(\varphi)$ with one free variable, if $\Amf^{\ast}\models \psi(a)$, then there exists $b\in \dom(\Amf^{\ast})$ such that $\Amf^{\ast}\models \psi(b)$, $b\not=a$, and $(a,b),(b,a)\not\in S^{\Amf^{\ast}}$ for any $S\in \sig(\varphi)$.
\end{proof}

We are now in a position to show that $m(\varphi)$ is a suc$\MLiu$-formula defining $\varphi$ whenever such a definition exists.

\begin{theorem}\label{prop:definability111}
An $\FOTNE$-formula $\varphi(x)$ 
is suc$\MLiu$-definable iff $\models \varphi \leftrightarrow m(\varphi)$.
\end{theorem}
\begin{proof}
Follows from Lemmas~\ref{lem:twinunfoldingtrans} and~\ref{lem:twinunfoldingequiv}.
\end{proof}

As $m(\varphi)$ is constructed in polytime in DAG-representation, we obtain:
\begin{corollary}
	$\CT/\MLiu$-definability is polytime reducible to $\CT$-validity.
\end{corollary}

We next consider the elimination of the universal and inverse modalities.

Define the \emph{$\omega$-unfolding} $\Amf^{\updownarrow\omega}$ of
a structure $\Amf$ as follows. Its domain $\dom(\Amf^{\updownarrow\omega})$ is the set of 
all
\begin{equation}\label{eq:def5}
	\vec{a}=a_{0}^{j_{0}}S_{0}\cdots S_{n-1}a_{n}^{j_{n}},
\end{equation}
where $n\geq 0$, $j_{i}\in \omega$, $a_{i}\in \dom(\Amf)$ for all $i\leq n$, and $S_{i}$ is either a binary predicate symbol in $\sig(\varphi)$ or its inverse such that $(a_{i},a_{i+1})\in S_{i}^{\Amf}$ for $i\leq n$. Let $\text{tail}(\vec{a})=a_{n}^{j}$.
For $\vec{a}$ of the form \eqref{eq:def5}, we set $\vec{a}\in A^{\Amf^{\updownarrow\omega}}$ iff $a_{n}\in A^{\Amf}$. Let $\vec{a}$ be of the form \eqref{eq:def5} and let $\vec{a}'= a_{0}^{j_{0}}S_{0}\cdots a_{n}^{j_{n}}S_{n}a_{n+1}^{j_{n+1}}$ be an extension of $\vec{a}$. Then $(\vec{a},\vec{a}')\in R^{\Amf^{\updownarrow\omega}}$ iff $R=S_{n}$ and 
$(\vec{a}',\vec{a})\in R^{\Amf^{\updownarrow\omega}}$ iff $R^{-}=S_{n}$.
\begin{lemma}\label{lem:updownarrow}
	For all pointed structures $\Amf,a$, we have $\Amf^{\updownarrow\omega},a^{0}\sim_{\MLiu}\Amf,a$.
\end{lemma}

$\Amf_{0}$ is a \emph{generated substructure} of $\Amf_{1}$ if it is a substructure of $\Amf_{1}$ and $a\in \text{dom}(\Amf_{0})$ and $(a,a')\in R^{\Amf_{1}}$ or
$(a',a)\in R^{\Amf_{1}}$ imply $a'\in \dom(\Amf_{0})$.
$\Amf_{0}$ is a \emph{forward generated substructure} of $\Amf_{1}$ if it is a substructure of $\Amf_{1}$ and $a\in \text{dom}(\Amf_{0})$ and $(a,a')\in R^{\Amf_{1}}$ imply $a'\in \dom(\Amf_{0})$. By $\Amf_{a}$ we denote the smallest
generated substructure of $\Amf$ containing $a\in \dom(\Amf)$.    

Define the \emph{forward $\omega$-unfolding} $\Amf^{\uparrow\omega}$  as the
substructure of $\Amf^{\updownarrow\omega}$, in which $\dom(\Amf^{\uparrow\omega})$ is the set of all
$$
\vec{a}=a_{0}^{j_{0}}R_{0}\cdots R_{n-1}a_{n}^{j_{n}}, 
$$
where each $R_{i}$ is a binary predicate symbol in $\sig(\varphi)$. Forward $\omega$-unfoldings have been considered before in the context of $\ML$-definability~\cite{DBLP:journals/igpl/AndrekaBN95,DBLP:conf/concur/JaninW96}.
Note that $\Amf^{\uparrow\omega}$ is a forward generated substructure of $\Amf^{\updownarrow\omega}$ and we have the following:

\begin{lemma}\label{lem:uparrow}
	$\Amf^{\uparrow\omega},a^{0}\sim_{\MLu}\Amf,a$.
\end{lemma}

For pointed structures $\Amf,a$ and $\Bmf,b$, we write $\Amf,a\cong\Bmf,b$ if there is an isomorphism from $\Amf$ onto $\Bmf$ mapping $a$ to $b$. 

\begin{lemma}\label{lem:lemfrombisimtoiso} 
For any at most countable structures $\Amf,\Bmf$ with 
	$a\in \dom(\Amf)$ and $b\in \dom(\Bmf)$,
		\begin{itemize}
		\item[--] if $\Amf,a\sim_{\MLi} \Bmf,b$, then
		$\Amf_{a^{0}}^{\updownarrow\omega},a^{0}\cong \Bmf_{b^{0}}^{\updownarrow\omega},b^{0}$\textup{;}
		\item[--] if $\Amf,a\sim_{\ML} \Bmf,b$, then
		$\Amf_{a^{0}}^{\uparrow\omega},a^{0}\cong\Bmf_{b^{0}}^{\uparrow\omega},b^{0}$.	
	\end{itemize}
\end{lemma}

We say that $\varphi$ is \emph{invariant under} (\emph{forward}) \emph{generated substructures} if, for all pointed $\Amf,a$ and all (forward) generated substructures $\Amf'$ of $\Amf$ with $a\in \dom(\Amf')$, we have $\Amf\models\varphi(a)$ iff $\Amf'\models \varphi(a)$.

\begin{lemma}\label{lem:semcrit}
	The following hold\textup{:}
	\begin{enumerate}
		\item $\varphi\in \MLiu$ is $\MLi$-definable iff $\varphi$ is invariant under generated substructures\textup{;}
		
		
		\item $\varphi\in \MLiu$ is $\ML$-definable iff $\varphi$ is invariant under forward generated substructures.
	\end{enumerate}
\end{lemma}
\begin{proof}
	(1) It is straightforward to show that every formula in $\MLi$ is invariant under generated substructures. Hence the direction from left to right follows. Conversely, assume that $\varphi\in\MLiu$ is not $\MLi$-definable but is invariant under generated substructures. We find countable $\Amf,a$ and $\Bmf,b$ with $\Amf\models \varphi(a)$, $\Bmf\models\neg\varphi(b)$, and $\Amf,a\sim_{\MLi}\Bmf,b$.
	By Lemma~\ref{lem:lemfrombisimtoiso}, we have 
	$\Amf_{a^{0}}^{\updownarrow\omega},a^{0}\cong \Amf_{b^{0}}^{\updownarrow\omega},b^{0}$.
	By invariance of $\varphi$ under generated substructures, $\Amf_{a^{0}}^{\updownarrow\omega}\models \varphi(a^{0})$ because
	$\Amf_{a^{0}}^{\updownarrow\omega}$ is a generated substructure of
	$\Amf^{\updownarrow\omega}$ and $\Amf^{\updownarrow\omega}\models \varphi(a^{0})$
	because, by Lemma~\ref{lem:updownarrow}, $\Amf^{\updownarrow\omega},a^{0}\sim_{\MLiu}\Amf,a$.
	
	By invariance of $\varphi$ under generated substructures, we obtain $\Bmf_{b^{0}}^{\updownarrow\omega}\models \neg\varphi(b^{0})$ because
	$\Bmf_{b^{0}}^{\updownarrow\omega}$ is a generated substructure of
	$\Bmf^{\updownarrow\omega}$ and $\Bmf^{\updownarrow\omega}\models \neg\varphi(b^{0})$
	because, by Lemma~\ref{lem:updownarrow}, $\Bmf^{\updownarrow\omega},b^{0}\sim_{\MLiu}\Bmf,b$.
	
	We have derived a contradiction as isomorphic pointed structures satisfy the same formulas.
	
	(2) It is straightforward to show that every formula in $\ML$ is invariant under forward generated substructures. Hence the direction from left to right follows. Conversely, assume that $\varphi\in\MLiu$ is not $\ML$-definable but is invariant under forward generated substructures. We find countable $\Amf,a$ and $\Bmf,b$ with $\Amf\models \varphi(a)$, $\Bmf\models\neg\varphi(b)$, and $\Amf,a\sim_{\ML}\Bmf,b$.
	By Lemma~\ref{lem:lemfrombisimtoiso}, 
	$\Amf_{a^{0}}^{\uparrow\omega},a^{0}\cong \Amf_{b^{0}}^{\uparrow\omega},b^{0}$.
	By invariance of $\varphi$ under forward generated substructures, $\Amf_{a^{0}}^{\uparrow\omega}\models \varphi(a^{0})$ because
	$\Amf_{a^{0}}^{\uparrow\omega}$ is a forward generated substructure of
	$\Amf^{\updownarrow\omega}$ and $\Amf^{\updownarrow\omega}\models \varphi(a^{0})$
	because, by Lemma~\ref{lem:updownarrow}, $\Amf^{\updownarrow\omega},a^{0}\sim_{\MLiu}\Amf,a$.
	
	By invariance of $\varphi$ under forward generated substructures, $\Bmf_{b^{0}}^{\uparrow\omega}\models \neg\varphi(b^{0})$ because
	$\Bmf_{b^{0}}^{\uparrow\omega}$ is a forward generated substructure of
	$\Bmf^{\updownarrow\omega}$ and $\Bmf^{\updownarrow\omega}\models \neg\varphi(b^{0})$
	because, by Lemma~\ref{lem:updownarrow}, $\Bmf^{\updownarrow\omega},b^{0}\sim_{\MLiu}\Bmf,b$.  
	
	We have derived a contradiction as isomorphic pointed structures satisfy the same formulas.    
\end{proof}

\begin{theorem}
Both $\CT/\MLi$- and $\CT/\ML$-definability are \coNExpTime-complete.
\end{theorem}
\begin{proof}
	We consider $\CT/\ML$-definability, the proof for $\CT/\MLi$-definability is similar. Observe that we can encode the problem of deciding whether a formula $\varphi(x)$ in suc$\MLiu$ is invariant under forward generated substructures as a validity problem in suc$\MLiu$. To show this, let $A$ be a fresh unary predicate symbol and let $\varphi(x)_{|A}$ be the relativisation of $\varphi(x)$ to $A$. Observe that, for all pointed $\Amf,a$, we have $\Amf\models \varphi_{|A}(a)$ iff $\Amf_{|A}\models \varphi(a)$, where $\Amf_{|A}$ is the substructure of $\Amf$ induced by $A^{\Amf}$
	and $a\in A^{\Amf}$. Then $\varphi(x)$ is invariant under forward generated substructures iff,
	for 
	$$
	\chi_{A}=\bigwedge_{R\in \sig(\varphi)} \forall x (A(x)\rightarrow \forall y (R(x,y) \rightarrow A(y))),
	$$
	the formula
	$$
	\chi_{A} \rightarrow \forall x ((\varphi(x) \wedge A(x)) \leftrightarrow \varphi(x)_{|A})
	$$
	is valid. Hence $\varphi(x)\in \CT$ is not $\ML$-definable iff at least one of the following conditions holds:
	\begin{enumerate}
		\item $\varphi \leftrightarrow \Flat(\varphi)$ is not valid;  
		\item $\Flat(\varphi) \leftrightarrow m(\Flat(\varphi))$ is not valid;
		\item $m(\Flat(\varphi))$ is not invariant under forward generated substructures.
	\end{enumerate}
The claim follows since each of these conditions can be checked in \NExpTime.    
\end{proof}

\section*{Proofs for Section~\ref{sec:uniformseparation}}

\thmuniformsepone*

\begin{proof}
	We first complete the proof of $(a) \Rightarrow (d)$ for the remaining pairs $\L/\LS$ of modal logics. If the universal modality is not admitted in $\LS$, then we take the set $\Gamma'$ of all $\LS_{2\md(\varphi)+2}(\varrho)$-formulas $\chi$ with $\Bmf^{\mx+1}\models\chi(a,0)$ 
	and $\gamma' = \bigwedge_{\chi\in \Gamma'}\chi$. Then, for every sequence $\Box=\Box_{R_{1}}\cdots\Box_{R_{n}}$ with $n\leq \md(\varphi)$ and $R_{1},\dots,R_{n}\in \varrho$, every $\chi\in \sub(\Flat(\varphi))$, and every $R\in \varrho$, $\Gamma$ contains
	$$
	\Box \Big( \Diamond_{R} \chi \rightarrow \big(\Diamond_{R}(\chi \wedge A_{0}) \wedge \cdots \wedge \Diamond_{R} (\chi \wedge A_{\mx+1})\big) \Big)
	$$
	and, for every sequence $\Box=\Box_{R_{1}}\cdots\Box_{R_{n}}$ with $n\leq 2\md(\varphi)$ and $R_{1},\dots,R_{n}\in \varrho$, it contains the formula 
	$$
	\Box(A_{i}\rightarrow \neg A_{j}), \quad \text{for $0\leq i < j \leq \mx$}. 
	$$
	This again gives $\gamma' \models \varphi \leftrightarrow \Flat(\varphi)$.
	If $\LS$ admits inverse modalities, in the formulas $\Gamma'$ above we have to consider not only $R\in \varrho$ but also $R^{-}$ with $R\in\varrho$. This allows us to show again that
	$\gamma' \models \varphi \leftrightarrow \Flat(\varphi)$.
	
	The implication $(b) \Rightarrow (a)$ is trivial.
	
	$(c) \Rightarrow (b)$ Let $\varphi\models \Flat(\varphi)$. To prove  that $\Flat(\varphi)$ is a uniform $\LS(\sigma)$-separator for $\varphi$, suppose $\varphi\models\chi$, for an $\LS(\sigma)$-formula $\chi$, aiming to show $\Flat(\varphi)\models \chi$.	Let $\Amf\models \Flat(\varphi)(a)$. Then $\Amf^{\omega}\models \Flat(\varphi)(a)$, by Lemma~\ref{lemma:bisim_omega}. Then $\Amf^{\omega}\models \varphi(a)$ 
	by Lemma~\ref{lemma:equivalenceomega}, and so $\Amf^{\omega}\models \chi(a)$. Using again Lemma~\ref{lemma:bisim_omega}, we obtain $\Amf\models \chi(a)$, as required.
	
	$(d) \Rightarrow (c)$ Suppose $\Amf\models \varphi(a)$. As $\varphi$ is  preserved under $\omega$-expansions, $\Amf^{\omega}\models \varphi(a)$. Hence $\Amf^{\omega}\models \Flat(\varphi)(a)$, by Lemma~\ref{lemma:equivalenceomega}. By Lemma~\ref{lemma:bisim_omega}, $\Amf\models \Flat(\varphi)(a)$, as required.
	
	The proof for $\CT/\FOTNE$ is similar and left to the reader.
\end{proof}

\thmuniformsepthree*

\begin{proof}
	We first give the construction of $\text{diag}_{\sigma_{f}}(\Amf,a)$ for the pair $\GMLinu/\MLinu$. Consider a finite pointed $\sigma_{f}$-structure $\Amf,a$, where $\sig(\varphi)\subseteq \sigma_{f}\subseteq \sigma$ and $\sigma_{f}$ is finite. Identify the elements of $\dom(\Amf)$ with some constants in $\sigma\setminus\sigma_{f}$. Then define $\text{diag}_{\sigma_{f}}(\Amf,a)$  as the conjunction of the following formulas:
	\begin{itemize}
		\item[--] $N_{a}$,
		
		\item[--] $\blacksquare\bigvee_{c\in \dom(\Amf)}N_{c}$,
		
		\item[--] $\Du N_{c}$, for $c\in \dom(\Amf)$,
		
		\item[--] $\blacksquare(N_{c}\rightarrow \neg N_{c'})$, for distinct $c,c'\in \dom(\Amf)$, 
		
		\item[--] $\Du(N_{c} \wedge A)$, for $c\in A^{\Amf}$,
		
		\item[--] $\Du(N_{c} \wedge N_{d})$, for $c=d^{\Amf}$ and $d\in \sigma_{f}$,
		
		\item[--] $\Du(N_{c} \wedge \neg A)$, for $c\not\in A^{\Amf}$,
		
		\item[--] $\Du(N_{c} \wedge \Diamond_{R}N_{d})$, for $(c,d)\in R^{\Amf}$,
		
		\item[--] $\Du(N_{c} \wedge \neg\Diamond_{R}N_{d})$, for $(c,d)\not\in R^{\Amf}$. 
	\end{itemize} 
	It is easy to see that $\text{diag}_{\sigma_{f}}(\Amf,a)$ has the properties (i) and (ii) from the proof sketch in the main paper.
	
	We now show the directions $(a) \Leftrightarrow (b)$ and $(c) \Rightarrow (b)$ for all pairs $\L/\LS$ from the theorem.
	
	$(a) \Rightarrow (b)$ Consider a uniform $\LS(\sigma)$-separator $\varphi'$ for $\varphi$.
	Clearly, $\varphi'$ is a uniform $\LS(\sigma')$-separator for $\varphi$, for any signature $\sigma'$ obtained from $\sigma$ by removing all but a finite set $\sigma_{0}'$ of constants that contains all constants occurring in $\varphi$. By Theorem~\ref{thm:uniformsep2}, $\Flat_{\sigma_{0}'}(\varphi)$ is also a uniform $\L(\sigma')$-separator for $\varphi$. By the definition of uniform separators, we have:
	\begin{itemize}
		\item[--] $\varphi'\models \Flat_{\sigma_{0}'}(\varphi)$ for any such $\sigma_{0}'$;
		
		\item[--] $\Flat_{\sigma_{0}'}(\varphi)\models \varphi'$ for any such $\sigma_{0}'$ with $\sigma_{0}'\supseteq \sigma_{0}$.
	\end{itemize}
	It follows that $\models \varphi'\leftrightarrow \Flat_{\sigma_{0}}(\varphi)$, as required.
	
	$(b) \Rightarrow (a)$ is trivial.
	
	$(c) \Rightarrow (b)$ Suppose $\varphi\models \chi$, for an $\LS(\sigma)$-formula $\chi$. We have to show that $\Flat_{\sigma_{0}}(\varphi)\models \chi$. Assume this is not the case. By the FMP of $\LS$, we have a finite pointed structure $\Amf,a$ with $\Amf\models \Flat_{\sigma_{0}}(\varphi)(a)$ and $\Amf\not\models\chi(a)$. But then, in view of $\Flat_{\sigma_{0}}(\varphi)\models_{\it fin} \varphi$ and $\varphi\models \chi$, we have $\Amf\models \chi(a)$, which is impossible.
	
	\smallskip
	
	It remains to show $(b) \Rightarrow (c)$ for the remaining pairs $\L/\LS$.
	For $\GMLnu/\MLnu$, the proof is the same as in the main part of the paper since $\text{diag}_{\sigma_{f}}(\Amf,a)$ is in $\MLnu$.
	
	For $\GMLin/\MLin$ and $\GMLn/\MLn$, we construct $\text{diag}_{\sigma_{f}}(\Amf,a)$ without using $\Du$. 
	We still require condition~(ii) but only the following weaker versions of condition~(i):
	\begin{itemize}
		\item[(i)] for $\GMLin$: if $\Bmf\models \text{diag}_{\sigma_{f}}(\Amf,a)(b)$, then $\Amf_{a},a \cong_{\sigma_{f}}\Bmf_{b},b$ for the smallest generated substructures $\Amf_{a}$ and $\Bmf_{b}$ of $\Amf$ and $\Bmf$ containing
		$a$ and $b$, respectively; 
		\item[(i)] for $\GMLn$: if $\Bmf\models \text{diag}_{\sigma_{f}}(\Amf,a)(b)$, then $\Amf_{a},a \cong_{\sigma_{f}}\Bmf_{b},b$ for the smallest forward generated substructures $\Amf_{a}$ and $\Bmf_{b}$ of $\Amf$ and $\Bmf$ containing
		$a$ and $b$, respectively. 
	\end{itemize}
	Recall from the previous section that 
	then the domain of $\Amf_{a}$ comprises all $a' \in \dom(\Amf)$ such that $a'=a$ or $a'$ is reachable from $a$ via a path
	$$
	a = a_0 S_1 a_1 S_2 \dots S_n a_n = a',
	$$
	where all $a_i$ are in $\dom(\Amf)$ and all $S_i$ are binary relations from $\sigma_f$ (for $\MLn$) and inverses thereof (for $\MLin$).
	For each $a' \in \dom(\Amf)$, $a' \ne a$, we fix a shortest path of this form and set
	$\Diamond_{a'} = \Diamond_{S_1} \dots \Diamond_{S_n}$, $\Box_{a'} = \Box_{S_1} \dots \Box_{S_n}$, 
	regarding both $\Diamond_a$ and $\Box_a$ as blank. We can now replace the conjuncts of the form $\Du \psi$ and $\blacksquare \psi$ in $\text{diag}_{\sigma_{f}}(\Amf,a)$ by
	$$
	\bigvee_{c \in \dom(\Amf_{a})} \hspace*{-2mm} \Diamond_c \psi \quad \text{and} \quad \bigwedge_{c \in \dom(\Amf_{a})} \hspace*{-2mm} \Box_c \psi,
	$$
	respectively, and add the conjuncts
	\begin{itemize}
		\item[--] $\Box_c (N_c \to \Box_S \hspace*{-2mm} \displaystyle \bigvee_{(c,d) \in S^{\Amf}} \hspace*{-2mm} N_d)$, for every relevant binary relation $S$.
	\end{itemize}
	Denote the resulting $\LS(\sigma_f)$-formulas by $\text{diag}_{\sigma_{f}}^{\LS}(\Amf,a)$, where $\LS\in \{\MLin,\MLn\}$. It is readily checked that they satisfy the respective versions of (i) for $\LS$.
	
	It is now straightforward to lift the proof of $(b) \Rightarrow (c)$ to the pairs $\GMLin/\MLin$ and $\GMLn/\MLn$.  
\end{proof}	

It follows that, for the pairs $\L/\LS$ considered in Theorem~\ref{thm:uniformsep3}, uniform $\LS(\sigma)$-separator existence for $\sigma$ containing infinitely many constants reduces in polytime to 
finite and general $\L$-valdity. For logics with the FMP, it follows from Theorems~\ref{thm:uniformsep3} and~\ref{prop:equivonomega2} that uniform separator existence and definability coincide: 

\begin{theorem}\label{prop:defanduniformthesame}
	Let $\L/\LS$ be any of the pairs $\GMLnu/\MLnu$, $\GMLu/\MLu$, $\GMLn/\MLn$, or $\GML/\ML$. Then, for any signature $\sigma$ with infinite $\sigma_{c}$ and any $\L(\sigma)$-formula $\varphi$, a uniform $\LS(\sigma)$-separator for $\varphi$ is an $\LS$-definition of $\varphi$.
\end{theorem}

We consider uniform $\FOTNE/\MLiu$-separation.

\begin{theorem}\label{thm:uniformsep20}
	For any signature $\sigma$ with $\sigma_{c}=\emptyset$ and any $\FOTNE(\sigma)$-formula $\varphi$, the following conditions are equivalent\textup{:}
	\begin{enumerate}
		\item[$(a)$] $\varphi$ has a uniform $\MLiu(\sigma)$-separator\textup{;}
		
		\item[$(b)$] $m(\varphi)$ is a uniform $\MLiu(\sigma)$-separator for $\varphi$\textup{;}
		
		\item[$(c)$] $\varphi\models m(\varphi)$\textup{;}
		
		\item[$(d)$] $\varphi$ is preserved under twin-unfoldings of $\sigma$-structures $\Amf$ in the sense that $\Amf\models \varphi(a)$ implies $\Amf^{\ast}\models \varphi(a^{0})$, for any $\Amf,a$.
	\end{enumerate}
\end{theorem}
\begin{proof}
	$(a) \Rightarrow (d)$. Given a uniform $\MLiu(\sigma)$-separator $\varphi'$ for $\varphi$, we need to show that $\varphi$ is preserved under twin unfoldings of $\sigma$-structures. Suppose this is not so. Take a pointed $\sigma$-structure $\Amf,a$ such that $\Amf\models \varphi(a)$ and $\Amf^{\ast}\not\models \varphi(a^0)$. 
	From (a) we obtain $\Amf\models \varphi'(a)$, and so, by By Lemma~\ref{lem:twinunfoldingtrans}, $\Amf^{\ast}\models\varphi'(a^{0})$.
	
	For any signature $\varrho$ and $k\geq 0$, let suc$\MLiu_{k}(\varrho)$ be the set of all suc$\MLiu(\varrho)$-formulas of rank $\le k$. If $\varrho$ is finite, then suc$\MLiu(\varrho)$ is finite modulo logical equivalence. 
	
	The proof is now similar to the proof of Theorem~\ref{thm:uniformsep1}. We construct a formula $\gamma(x) \in \text{suc}\MLiu$ using fresh unary relation symbols such that $\Amf^{\ast}$ can always be expanded to a model of $\gamma$ and $\gamma\models \varphi\leftrightarrow m(\varphi)$. The latter is achieved by constructing $\gamma$
	in such a way that it ensures conditions (1) and (2) of the proof of Lemma~\ref{lem:twinunfoldingequiv}.
	
	Consider fresh unary $A_{0}^{S},A_{1}^{S},A_{2}^{S},B_{0},B_{1}$,  for 
	$S\in \mathcal{R}(\varphi)$ with $\mathcal{R}(\varphi)=\{ R \mid R\in \sig(\varphi)\} \cup \{R^{-} \mid R\in \sig(\varphi)\}$. Let $\varrho$ be the union of the set of fresh symbols and $\sig(\varphi)\cup \sig(\varphi')$. We first ensure condition~2 by stating that $B_{0},B_{1}$ partition the domain, satisfy the same formulas in suc$\MLiu_{k}(\varrho)$, and that elements satisfying $B_{0}$ and $B_{1}$ are not related by any $R\in \sigma$. So let $\gamma$ contain the conjuncts
	\begin{align*}
		& \blacksquare ((B_{0}\rightarrow \neg B_{1}) \wedge (B_{0} \vee B_{1})),\\
		& \blacksquare ((B_{0} \rightarrow \neg \Diamond_{R}B_{1}) \wedge (B_{1} \rightarrow \neg \Diamond_{R}B_{0})), \\ 
		& \Du (\chi \wedge B_{0}) \leftrightarrow \Du(\chi \wedge B_{1}),
	\end{align*}  
	for all $\chi\in \text{suc}\MLiu_{k}(\varrho)$. To enforce condition~1, $\gamma$ states that the $A_{0}^{S},A_{1}^{S},A_{2}^{S}$ partition the domain:
	\begin{align*}
		&\blacksquare \bigvee_{0\leq i\leq 2,S\in \mathcal{R}(\varphi)}A_{i}^{S},\\
	    &\blacksquare \bigwedge_{A_{i}^{S}\not=A_{j}^{S'}}(A_{I}^{S} \rightarrow \neg A_{j}^{S'}),\\
	\end{align*}
	and that the relations $S$ can link the $A_{i}^{S'}$ only in a certain way: for all  $S,S',S''\in \mathcal{R}(\varphi)$ and all $i\leq 2$: 
	\begin{align*}
		& \blacksquare (A_{i}^{S}\rightarrow \Box_{S'}\neg A_{i\oplus_{3} 1}^{S''}), \text{ if $S'\not=S''$},\\
		& \blacksquare (A_{i}^{S}\rightarrow \Box_{S'}\neg A_{i}^{S''}).
	\end{align*}
	Then $\gamma$ is satisfied in any $\varrho$-structure $\Bmf^{\ast}$ obtained from some $\Amf^{\ast}$ by satisfying $B_{0}$ in all nodes of the form \eqref{eq:def1} with $j=0$, $B_{1}$ in all nodes of the form \eqref{eq:def1} with $j=1$ and by satisfying $A_{i}^{S}$ in a node of the form \eqref{eq:def1} if there is $k\geq 0$ with $n=3k+i$ and $S_{n-1}=S$. Then $\Bmf^{\ast}\models \gamma(a^{0})$.
	By construction, $\gamma \models \varphi \leftrightarrow m(\varphi)$.
	As a consequence, we have $\varphi\models \gamma \rightarrow m(\varphi)$, and so, by the definition of uniform $\MLiu(\sigma)$-separators, $\varphi'\models \gamma \rightarrow m(\varphi)$. Since $\Bmf^{\ast}\models \varphi'(a^{0})$ and $\Bmf^{\ast}\models \gamma(a^{0})$, we obtain $\Bmf^{\ast}\models m(\varphi)(a^{0})$. Therefore, $\Amf^{\ast}\models m(\varphi)(a^{0})$ as $m(\varphi)$ does not contain any fresh symbols. It follows from Lemma~\ref{lem:twinunfoldingtrans} that $\Amf^{\ast}\models\varphi(a^{0})$, which is impossible.
	
	The implication $(b) \Rightarrow (a)$ is trivial.
	
	$(c) \Rightarrow (b)$ Let $\varphi\models m(\varphi)$. To prove  that $m(\varphi)$ is a uniform $\MLiu(\sigma)$-separator for $\varphi$, suppose $\varphi\models\chi$, for an $\MLiu(\sigma)$-formula $\chi$, aiming to show $m(\varphi)\models \chi$. Let $\Amf\models m(\varphi)(a)$. Then $\Amf^{\ast}\models m(\varphi)(a)$, by Lemma~\ref{lem:twinunfoldingequiv}. Then $\Amf^{\ast}\models \varphi(a)$ 
	by Lemma~\ref{lem:twinunfoldingtrans}, and so $\Amf^{\ast}\models \chi(a)$. Using again Lemma~\ref{lem:twinunfoldingequiv}, we obtain $\Amf\models \chi(a)$, as required.
	
	$(d) \Rightarrow (c)$ Suppose $\Amf\models \varphi(a)$. As $\varphi$ is  preserved under twin-unfoldings, $\Amf^{\ast}\models \varphi(a^{0})$. Hence $\Amf^{\ast}\models m(\varphi)(a^{0})$, by Lemma~\ref{lem:twinunfoldingtrans}. By Lemma~\ref{lem:twinunfoldingequiv}, $\Amf\models m(\varphi)(a)$, as required.
\end{proof}


\bibliographystyle{../IEEEtran}
\bibliography{../bibliobeth,../local}

\end{document}